\documentclass[11pt]{article}

\usepackage{amsfonts}
\usepackage{amssymb}
\usepackage{amstext}
\usepackage{amsmath}
\usepackage{enumerate}
\usepackage{algorithm}
\usepackage{algorithmic}

\usepackage{amsthm}

\usepackage{thmtools}

\usepackage[usenames,dvipsnames,svgnames,table]{xcolor}

\usepackage[pagebackref,letterpaper=true,colorlinks=true,pdfpagemode=UseNone,citecolor=OliveGreen,linkcolor=BrickRed,urlcolor=BrickRed,pdfstartview=FitH]{hyperref}

\usepackage{framed}
\usepackage{enumitem}


\newtheorem{theorem}{Theorem}[section]
\newtheorem{fact}[theorem]{Fact}
\newtheorem{lemma}[theorem]{Lemma}
\newtheorem{definition}[theorem]{Definition}

\newtheorem{corollary}[theorem]{Corollary}

\newtheorem{claim}[theorem]{Claim}
\newtheorem{question}[theorem]{Question}
\newtheorem{problem}[theorem]{Problem}
\newtheorem*{problem*}{Problem}
\newtheorem{remark}[theorem]{Remark}
\newtheorem*{remark*}{Remark}

\newtheorem{example}[theorem]{Example}

\newcommand{\opt}[0]{{\ensuremath{\sf{opt}}}}

\numberwithin{equation}{section}
\numberwithin{table}{section}

\renewcommand{\preceq}{\preccurlyeq}
\renewcommand{\succeq}{\succcurlyeq}

\renewcommand{\tilde}{\widetilde}

\newcommand{\R}{\ensuremath{\mathbb R}}
\newcommand{\Z}{\ensuremath{\mathbb Z}}

\newcommand{\E}[1]{{\mathbb{E}}\left[#1\right]}

\newcommand{\junk}[1]{}

\newcommand{\norm}[1]{\left\lVert#1\right\rVert}

\newcommand{\vertiii}[1]{{\left\vert\kern-0.25ex\left\vert\kern-0.25ex\left\vert #1 \right\vert\kern-0.25ex\right\vert\kern-0.25ex\right\vert}}

\newenvironment{proofof}[1]{{\medbreak\noindent \em Proof of #1.  }}{\hfill\qed\medbreak}

\def\b1{{\bf 1}}
\def\eps{{\epsilon}}

\def\R{\mathbb{R}}

\def\diag{\operatorname{diag}} 
\def\opt{\sf{opt}} 
\def\Reff{\text{Reff}}

\def\tr{\operatorname{tr}}

\global\long\def\E{\mathbb{E}}

\global\long\def\R{\mathbb{R}}

\newcommand{\inner}[2]{\langle #1, #2 \rangle} 

\DeclareMathOperator{\argmin}{argmin}

\title{A Spectral Approach to Network Design}

\author{Lap Chi Lau\footnote{School of Computer Science, University of Waterloo. Supported by NSERC Discovery Grant 2950-120715 and NSERC Accelerator Supplement 2950-120719. Email: \href{mailto:lapchi@uwaterloo.ca}{lapchi@uwaterloo.ca}},~~~~~
Hong Zhou\footnote{School of Computer Science, University of Waterloo. Supported by NSERC Discovery Grant 2950-120715 and NSERC Accelerator Supplement 2950-120719. Email: \href{mailto:h76zhou@uwaterloo.ca}{h76zhou@uwaterloo.ca}}}

\date{}

\begin{document}

\begin{titlepage}
\def\thepage{}
\thispagestyle{empty}

\maketitle

\begin{abstract}
We present a spectral approach to design approximation algorithms for network design problems.
We observe that the underlying mathematical questions are the spectral rounding problems,
which were studied in spectral sparsification and in discrepancy theory.
We extend these results to incorporate additional non-negative linear constraints, and show that they can be used to significantly extend the scope of network design problems that can be solved.
Our algorithm for spectral rounding is an iterative randomized rounding algorithm based on the regret minimization framework.
In some settings, this provides an alternative spectral algorithm to achieve constant factor approximation for the classical survivable network design problem, and partially answers a question of Bansal about survivable network design with concentration property.
We also show many other applications of the spectral rounding results, including weighted experimental design and additive spectral sparsification.
\end{abstract}

\end{titlepage}

\thispagestyle{empty}


\newpage

\section{Introduction}

Network design is a central topic in combinatorial optimization, approximation algorithms and operations research.
The general setting of network design is to find a minimum cost subgraph satisfying certain requirements.
The most well-studied problem is the survivable network design problem~\cite{GGP+94,AKR95,GW95,GGT+09}, where the requirement is to have at least a specified number $f_{u,v}$ of edge-disjoint paths between every pair of vertices $u,v$. 
A seminal work of Jain~\cite{Jai01} introduced the iterative rounding method for linear programming to design a $2$-approximation algorithm for the survivable network design problem,
and this method has been extended to various more general settings~\cite{FJW01,Gab05,CVV06,LNSS09,LRS11,EV14,FNR15,LZ15,Ban19}.
There are also other linear programming based algorithms such as randomized rounding~\cite{Sri01,GKR00,BGRS13,AGM+14,GLL19} to obtain important algorithmic results for network design.
It is widely recognized that linear programming is the most general and powerful approach in designing approximation algorithms for network design problems.

In the past decade, spectral techniques have been developed to make significant progress in designing graph algorithms~\cite{SS11,CKM+11,BSS12,AZLO15,AO15,Sch18}.
One striking example is the spectral sparsification problem introduced by Spielman and Teng~\cite{ST11}, where the objective is to find a sparse edge-weighted graph $H$ to approximate the input graph $G$ so that $(1-\eps)L_G \preceq L_H \preceq (1+\eps)L_G$ where $L_G$ and $L_H$ are the Laplacian matrices of the graph $G$ and $H$.
The spectral condition $(1-\eps)L_G \preceq L_H \preceq (1+\eps)L_G$ implies that $H$ is also a cut sparsifier of $G$ such that the total weight on each cut in $H$ is approximately the same as that in $G$.
Batson, Spielman, Srivastava~\cite{BSS12} proved that every graph $G$ has a spectral sparsifier $H$ with only $O(n/\eps^2)$ edges.
This improves upon the influential result of Bencz\'ur and Karger~\cite{BK96} that every graph $G$ has a cut sparsifier $H$ with $O(n \log n / \eps^2)$ edges, which has many applications in designing fast algorithms for graph problems.
From a technical perspective, the spectral approach introduces linear algebraic concepts and continuous optimization techniques in solving graph problems, and the results in spectral sparsification~\cite{BSS12,AZLO15,AO15} show that it is algorithmically more convenient to control the spectral properties of the graph in order to control its combinatorial properties.
 
Inspired by these developments, we are motivated to study whether there is a spectral approach to design approximation algorithms for network design problems.  
The general way to designing approximation algorithms is to solve a convex program to obtain a fractional solution $x$ in polynomial time, and then to round $x$ into an integral solution $z$ that well approximates $x$ (with respect to the constraints and the objective function) as an approximate solution.
We observe that the following spectral rounding question, where the objective is to approximate the spectral properties of $x$, underlies a large class of problems including the survivable network design problem.

\begin{question}[Spectral Rounding]
For each edge $e$ in a graph, let $L_e$ be the Laplacian matrix of $e$ and $c_e$ be its cost.  Given $x_e \in \R_{+}$ for each edge $e$, characterize when we can find $z_e \in \Z_{+}$ for each $e$ such that
\[
\sum_{e} x_e L_e \approx \sum_e z_e L_e 
\quad {\rm and} \quad 
\sum_e c_e x_e \approx \sum_e c_e z_e.
\]
\end{question}

When spectral rounding is possible, we notice that the integral solution $z$ not only approximately preserves the cost and the pairwise edge connectivity properties of $x$ as required by the survivable network design problem, but also many other properties of $x$ including pairwise effective resistances, the graph expansion, and degree constraints. 
This would significantly extend the scope of useful properties that a network designer could control simultaneously to design better networks.

\subsection{General Survivable Network Design} \label{ss:network-intro}

The main conceptual contribution of this paper is to show that the techniques in spectral graph theory and discrepancy theory can be used to significantly extend the scope of network design problems that can be solved.

In network design, we are given a graph $G=(V,E)$ where each edge has a cost $c_e$, and the objective is to find a minimum cost subgraph that satisfies certain requirements.
In survivable network design~\cite{GGP+94,Jai01}, 
the requirements are pairwise edge-connectivities, that every pair of vertices $u,v$ should have at least $f_{uv}$ edge-disjoint paths for $u,v \in V$. 
This captures several classical problems as special cases, including minimum Steiner tree~\cite{BGRS13}, minimum Steiner forest~\cite{AKR95,GW95}, and minimum $k$-edge-connected subgraph~\cite{GGT+09}.
Jain introduced the iterative rounding method for linear programming to design a $2$-approximation algorithm for the survivable network design problem~\cite{Jai01}.
His proof exploits the nice structures of the connectivity constraints to show that there is always a variable $x_e$ with value at least $\frac12$ in any extreme point solution to the linear program.
His work leads to many subsequent developments in network design~\cite{FJW01,CVV06,Gab05,GGT+09,CV14},
and the iterative rounding algorithm is still the only known constant factor approximation algorithm for the survivable network design problem.

Motivated by the need of more realistic models for the design of practical networks, researchers study generalizations of survivable network design problems where we can incorporate additional useful constraints.
One well-studied problem is the degree-constrained survivable network design problem, where there is a degree upper bound $d_v$ on each vertex $v$ to control its workload.
There is a long line of work on this problem~\cite{Rag96,RMR+01,Goe06,LNSS09,EV14,FNR15,LZ15} and the iterative rounding method has been extended to incorporate degree constraints into survivable network design successfully.
In the general setting~\cite{LNSS09,LV10,LZ15}, there is a polynomial time algorithm to find a subgraph that violates the cost and the degree constraints by a multiplicative factor of at most $2$.
For interesting special cases such as finding a spanning tree~\cite{Goe06,SL15} or a Steiner tree~\cite{LS13,LZ15}, there is a polynomial time algorithm that returns a solution that violates the degree constraint by an additive constant.

More generally, one can consider to add linear packing constraints and linear covering constraints into survivable network design~\cite{BGRS04,BKK+13,OZ18,LSwa18}, but not as much is known about how to approximately satisfy these constraints simultaneously especially when the linear constraints are unstructured.

Another natural constraint is to control the shortest path distance between pairs of vertices, but unfortunately this is proved to be computationally hard~\cite{DK99} to incorporate into network design.

In~\cite{CLSWZ19}, together with Chan, Schild, and Wong, we propose to incorporate the effective resistance metric into network design, as an interpolation of shortest path distance and edge-connectivity between vertices.
Incorporating effective resistances can also allow one to control some natural quantities about random walks on the resulting subgraph, such as the commute time between vertices~\cite{CRR+96} and the cover time~\cite{Mat88,DLP11}.  
We note that effective resistances have interesting connections to many other graph problems, including spectral sparsification~\cite{SS11}, maximum flow computation~\cite{CKM+11}, asymmetric traveling salesman problem~\cite{AO15}, and random spanning tree generation~\cite{MST15,Sch18}.
We believe that it is a useful property to be incorporated into network design.

There are many other natural constraints that could help in designing better networks, including total effective resistances~\cite{GBS08}, algebraic connectivity (and graph expansion)~\cite{GB06}, and the mixing time of random walks~\cite{BDX04}.
These constraints are also well-motivated and were studied individually before (without taking other constraints together into consideration, e.g.~connectivity requirements), but not much is known about approximation algorithms with nontrivial approximation guarantees for these constraints (see Section~\ref{s:spectral}). 

It would be ideal if a network designer can control all of these properties simultaneously to design a good network that suits their need.
We can write a convex programming relaxation for this general network design problem incorporating all these constraints.
\begin{alignat}{2} \label{eq:P-intro}
    & {\sf cp} := \min_x & &~ \inner{c}{x} \nonumber \\
    &  & & \begin{aligned}
          & x(\delta(S)) \geq f(S)  & & \quad \forall S \subseteq V &  & \text{\quad (connectivity constraints)} \\
          & x(\delta(v)) \leq d_v  & & \quad \forall v \in V & & \text{\quad (degree constraints)}\\
          & A x \leq a & & \quad A \in \R_+^{p \times m}, a \in \R_+^p & & \text{\quad (linear packing constraints)}\\
          & B x \geq b & & \quad B \in \R_+^{q \times m}, b \in \R_+^q & & \text{\quad (linear covering constraints)}\\
        & \Reff_{x}(u,v) \leq r_{uv} & & \quad \forall u,v \in V& & \text{\quad (effective resistance constraints)}\\
        & L_{x} \succcurlyeq M  & & \quad M \succcurlyeq 0 & & \text{\quad (spectral constraints)}\\
        & \lambda_2(L_x) \geq \lambda & & & & \text{\quad (algebraic connectivity constraint)} \\
        & 0 \leq x_e \leq 1 & & \quad \forall e \in E & & \text{\quad (capacity constraints)}
    \end{aligned} \tag{CP}
\end{alignat}

The connectivity constraints are specified by a function $f$ on vertex subsets,
e.g.~in survivable network design $f(S) := \max_{u,v} \{ f_{u,v} \mid u\in S, v\notin S\}$.
The matrix $L_x$ is the Laplacian matrix of the fractional solution $x$.
More explanations about this convex program can be found in Section~\ref{ss:convex}.

Our first result for network design is the following approximation algorithm for this general problem.
We remark that the degree constraints are not handled in the following result.

\begin{theorem}[Informal] \label{t:network-zero-one-intro}
Suppose we are given an optimal solution $x$ to the convex program~\eqref{eq:P-intro}.
There is a polynomial time randomized algorithm 
to return an integral solution $z$ to \eqref{eq:P-intro} that simultaneously satisfies all the connectivity constraints, the effective resistance constraints, the spectral constraints, the algebraic connectivity constraint and the capacity constraints exactly with high probability.
The objective value of the integral solution $z$ is
\[
\inner{c}{z} \leq (1+O(\eps)) \cdot {\sf cp} + O \left(\frac{n c_{\infty}}{\eps} \right)
\]
with high probability, 
where $n$ is the number of vertices in the graph and 
$c_{\infty} := \norm{c}_{\infty}$ is the maximum cost of an edge.
Furthermore, the linear packing constraints and the linear covering constraints are satisfied approximately with high probability (see Theorem~\ref{t:zero-one} and Theorem~\ref{t:network-zero-one} for the approximation guarantees for these constraints).
\end{theorem}

Note that this provides a $(1+O(\eps))$-approximation algorithm if ${\sf cp} \gtrsim nc_{\infty}/\eps^2$, and a constant factor approximation algorithm if ${\sf cp} \gtrsim nc_{\infty}$.
We remark that, for survivable network design, the $(1+O(\eps))$-approximation algorithm does not improve on the $2$-approximation algorithm of Jain's result, as Jain's algorithm always returns a solution with cost at most ${\sf cp} + 2nc_{\infty}$.

The main advantage of the spectral approach is that it significantly extends the scope of useful properties that can be incorporated into network design, while previously there are no known non-trivial approximation algorithms even for some individual constraints.
We demonstrate the use of Theorem~\ref{t:network-zero-one-intro} with one concrete setting.

\begin{example} \label{e:example}
Suppose the connectivity requirement satisfies $f_{u,v} \geq k$ for all $u,v \in V$ (e.g.~to find a $k$-edge-connected subgraph).
Assume the cost $c_e$ of each edge $e$ satisfies $1 \leq c_e \leq O(k)$.
Then Theorem~\ref{t:network-zero-one-intro} provides a constant factor approximation algorithm for this survivable network design problem.
 To our knowledge, the only known constant factor approximation algorithm even restricted to this special case is Jain's iterative rounding algorithm.
 The algorithm in Theorem~\ref{t:network-zero-one-intro} provides a completely different spectral algorithm to achieve constant factor approximation in this special case.

Furthermore, the constant factor approximation algorithm can be achieved while incorporating additional effective resistance constraints (e.g.~to upper bound commute times between pairs of vertices), spectral constraints (e.g.~to dominate another graph/topology in terms of the number of edges in cuts), algebraic connectivity constraint (e.g.~to lower bound graph expansion).
Also, additional linear packing and covering constraints can be satisfied approximately, even when they are unstructured.
See Section~\ref{s:network} for a more in-depth discussion.
\end{example}

Recently, Bansal~\cite{Ban19} designed a rounding technique that achieves the guarantees by iterative rounding and randomized rounding simultaneously, and he showed various interesting applications of his techniques.
However, he left it as an open question whether there is an $O(1)$-approximation algorithm for survivable network design while satisfying some concentration property of the output.
Theorem~\ref{t:network-zero-one-intro} provides some progress towards his question (e.g.~in the setting in Example~\ref{e:example}), as the guarantees on the linear packing and linear covering constraints satisfy some concentration property as shown in Theorem~\ref{t:network-zero-one}. 
We defer to Section~\ref{ss:Bansal} for details.

Our second result for network design is a strong upper bound on the integrality gap of the convex program that incorporates degree constraints as well, assuming the fractional solution $x$ satisfies some additional properties.

\begin{theorem}[Informal] \label{t:network-two-sided-intro}
Suppose we are given a solution $x$ to the convex program~\eqref{eq:P-intro}.
Assume that $\Reff_x(u,v) \leq \eps^2$ for every $uv \in E$ and $c_{\infty} \leq \eps^2 \inner{c}{x}$ for some $\eps \in [0,1]$.
Then, there exists an integral solution $z$ that approximately satisfies all the connectivity constraints, degree constraints, effective resistance constraints, spectral constraints, algebraic connectivity constraints, and capacity constraints with  
$\inner{c}{z} \leq (1+O(\eps))\inner{c}{x}$.
\end{theorem}

We remark that Theorem~\ref{t:network-two-sided-intro} does not provide a polynomial time algorithm to find such an integral solution, as it is proved using the non-constructive results in discrepancy theory.
Also, we note that Theorem~\ref{t:network-two-sided-intro} does not handle linear covering and packing constraints.
The assumption $\Reff_x(u,v) \leq \eps^2$ for every $uv \in E$ may not be satisfied in applications, and we will explain in Section~\ref{ss:two-sided} when it will be satisfied and show that it is not too restrictive.

\subsection{Previous Work on Spectral Rounding}

The most relevant works for spectral rounding are from spectral sparsification and discrepancy theory.
There are two previous theorems that imply non-trivial results for spectral rounding.

\subsubsection{Spectral Sparsification}

There are various algorithms for spectral sparsifications, by random sampling~\cite{SS11}, by barrier functions~\cite{BSS12}, by regret minimization~\cite{AZLO15,SHS16}, and by some combinations of these ideas~\cite{LS18,LS17}.
Most of these algorithms need to work with arbitrary weights and cannot guarantee that the output subgraph has only integral weights.
There are some algorithms which guarantee that the output has only integral weights, but they only achieve considerably weaker spectral approximation~\cite{AGM14,AZLO15,BST19}.

Allen-Zhu, Li, Singh, and Wang~\cite{AZLSW20} formulated and proved the following spectral rounding theorem, using the framework of regret minimization developed for spectral sparsification~\cite{AZLO15}.

\begin{theorem}[\cite{AZLSW20}] \label{t:swap}
Let $v_1, v_2, \ldots, v_m \in \R^n$, $x \in [0,1]^m$ and $k = \sum_{i=1}^m x_i$.
Suppose $\sum_{i=1}^m x_i v_i v_i^T = I_n$ and $k \geq 5n/\eps^2$ for some $\eps \in (0,\frac13]$. 
Then there is a polynomial time algorithm to return a subset $S \subseteq [m]$ with
\[|S| \leq k
\quad {\rm and} \quad 
\sum_{i \in S} v_i v_i^T \succeq (1-3\eps)I_n.\]
\end{theorem}

Theorem~\ref{t:swap} can be understood as a one-sided spectral rounding result, where the fractional solution $x$ is rounded to a zero-one solution while the budget constraint is satisfied and the spectral lower bound is approximately satisfied.
Through a general reduction, this theorem implies near-optimal approximation algorithms for a large class of experimental design problems~\cite{AZLSW20}.

We remark that Theorem~\ref{t:swap} can be modified to prove similar but more restrictive results for network design as in Theorem~\ref{t:network-zero-one-intro}, when the objective function $c$ is the all-one vector and there are no linear covering and packing constraints.
This already extends the scope of unweighted network design significantly, but this connection was not made before.
For network design, it is desirable to have different costs on edges, and these weighted problems are usually more difficult to solve than the unweighted problems (e.g. minimum $k$-edge-connected subgraphs~\cite{GGT+09} vs \cite{Jai01}, minimum bounded degree spanning trees~\cite{FR94} vs \cite{Goe06}, etc).

\subsubsection{Discrepancy Theory}

The techniques in spectral sparsification have been extended greatly to prove discrepancy theorems in spectral settings~\cite{MSS15b,AO15,KLS19}, most notably in the solution to Weaver's conjecture that resolves the Kadison-Singer problem~\cite{MSS15a,MSS15b} and its extension and surprising application to the asymmetric traveling salesman problem~\cite{AO15}.
The following recent result by Kyng, Luh, and Song~\cite{KLS19} provides the most refined formulation in the discrepancy setting,
using the method of interlacing polynomials and the barrier arguments developed in~\cite{MSS15b,AO15}.

\begin{theorem}[\cite{KLS19}] \label{t:4std}
Let $v_1, ..., v_m \in \R^n$, and $\xi_1, ..., \xi_m$ be independent random scalar variables with finite support.
There exists a choice of outcomes $\eps_1, ..., \eps_m$ in the support of $\xi_1, ..., \xi_m$ such that
\[
\left\| \sum_{i=1}^m \E[\xi_i] v_i v_i^T - \sum_{i=1}^m \eps_i v_i v_i^T \right\|_{{\rm op}} \leq 4 \left\| \sum_{i=1}^m {\bf Var}[\xi_i] (v_i v_i^T)^2 \right\|_{{\rm op}}^{1/2}.
\]
\end{theorem}

We note that Theorem~\ref{t:4std} implies the following two-sided spectral rounding result, which is very similar to Corollary~1.7 in~\cite{KLS19} but with a weaker assumption, where we only need $\norm{\sum_{i=1}^m x_i v_i v_i^T}_{\rm op} \leq 1$ instead of $\norm{\sum_{i=1}^m v_i v_i^T}_{\rm op} \leq 1$ as in~\cite{KLS19}.
The proof will be presented in Section~\ref{s:discrepancy} in a more general setting.

\begin{corollary} \label{c:4std}
Let $v_1, ..., v_m \in \R^n$ and $x \in [0,1]^m$. 
Suppose $\sum_{i=1}^m x_i v_i v_i^T = I_n$ and $\|v_i\| \leq \eps$ for all $i \in [m]$.
Then there exists a subset $S \subseteq [m]$ satisfying 
\[(1-O(\eps))I_n \preceq \sum_{i \in S} v_i v_i^T \preceq (1+O(\eps)) I_n.\]
\end{corollary}

Comparing to Theorem~\ref{t:swap}, the advantage of Corollary~\ref{c:4std} is that it provides a two-sided spectral approximation.
On the other hand, Corollary~\ref{c:4std} requires the assumption that all vectors are short, and it has no guarantee on the size of $S$. 
Also, it is important to point out that the proof of Corollary~\ref{c:4std} does not provide a polynomial time algorithm to find such a subset.

\subsection{Our Technical Contributions}

We extend the previous results on spectral rounding to incorporate non-negative linear constraints and to satisfy the requirements for network design problems.
These results have interesting applications in many other problems besides network design; see Section~\ref{s:applications-intro} and Section~\ref{s:applications}.

Our main technical result considers one-sided spectral rounding.

\begin{theorem} \label{t:zero-one}
Suppose we are given $v_1, ..., v_m \in \R^n$ and $x \in [0,1]^m$ such that $\sum_{i=1}^m x_i v_i v_i^T = I_n$. 
For any $\eps \in (0, \frac14)$, there is a polynomial time randomized algorithm that returns a solution $z \in \{0,1\}^m$ such that
\[
\sum_{i=1}^m z_i v_i v_i^T \succcurlyeq I_n
\]
with probability at least $1-\exp\left( -\Omega(n) \right)$. 
Furthermore, for any $c \in \R^m_+$, the solution $z$ satisfies the upper bound
\[
\inner{c}{z} \leq (1+6\eps) \inner{c}{x} + \frac{15 n c_{\infty}}{\eps}
\]
with probability at least $1 - \exp( - \Omega(n) )$, and the solution $z$ satisfies the lower bound
\[
\inner{c}{z} \geq \inner{c}{x} - \delta n c_{\infty}
\]
with probability at least $1-\exp\left(-\Omega\left(\min\{\eps \delta, \eps \delta^2\} \cdot n\right)\right)$ for $\delta > 0$.
\end{theorem}

The main advantage of Theorem~\ref{t:zero-one} over Theorem~\ref{t:swap} is that we can prove that $\inner{c}{z}$ is not too far from $\inner{c}{x}$ for an arbitrary vector $c \in \R^m_{+}$ with high probability.
This allows us to bound the cost of the returned solution to network design problems, and when $nc_{\infty} \lesssim \inner{c}{x}$ we can conclude that $z$ is a constant factor approximate solution.
Note that the guarantee on linear constraints can be applied to up to exponentially many constraints.
This allows us to incorporate additional linear packing and covering constraints into network design and have some non-trivial guarantees.
Another advantage is that we construct a solution that satisfies the spectral lower bound exactly, by allowing the solution to choose more than $k = \sum_{i=1}^m x_i$ vectors.
This is important in network design problems where we would like to construct a solution that satisfies all the constraints (instead of approximately satisfying all the constraints), by allowing the cost of the solution to be higher than the cost of the optimal solutions.

We note that there are examples showing that the additive error term $O(nc_{\infty}/\eps)$ in Theorem~\ref{t:zero-one} is tight up to a constant factor (see Section~\ref{ss:integrality}).

Using the proof techniques in Theorem~\ref{t:zero-one}, we can strengthen a recent deterministic algorithm by Bansal, Svensson and Trevisan~\cite{BST19} to construct unweighted spectral sparsifiers, to ensure that there will be no parallel edges in the sparsifier.
See Section~\ref{s:additive} for details.

For two-sided spectral rounding, we show that Corollary~\ref{c:4std} can be extended to incorporate one given non-negative linear constraint. 

\begin{theorem} \label{t:two-sided}
Let $v_1, ..., v_m \in \R^n$, $x \in [0,1]^m$ and $c \in \R^m_{+}$.
Suppose $\sum_{i=1}^m x_i v_i v_i^T = I_n$,  
$\|v_i\| \leq \eps < \frac{1}{8}$ for all $i \in [m]$ and
$c_{\infty} \leq \eps^2 \inner{c}{x}$.
Then there exists $z \subseteq \{0,1\}^m$ such that 
\[(1-8\eps)I_n \preceq \sum_{i = 1}^m z_i v_i v_i^T \preceq (1 + 8\eps)I_n 
\quad  \text{and} \quad 
(1-8\eps)\inner{c}{x} \leq \inner{c}{z} \leq (1 + 8\eps) \inner{c}{x}.\]
\end{theorem}

Note that the linear constraint $c$ in Theorem~\ref{t:two-sided} is required to be given as part of the input, while it is not required so in Theorem~\ref{t:zero-one}.
Theorem~\ref{t:two-sided} is useful in bounding the integrality gap for convex programs for network design problems, showing strong approximation results when the assumptions are satisfied (see Section~\ref{ss:two-sided}).
Also, we will show in Section~\ref{s:additive} that it can be used in the study of additive unweighted spectral sparsification~\cite{BST19}, proving an optimal existential result.

\subsubsection{Techniques}

The main technical contribution is an iterative randomized rounding algorithm for Theorem~\ref{t:zero-one}. 
Our algorithms is based on the regret minimization framework developed in~\cite{AZLO15,AZLSW20} for spectral sparsification and one-sided spectral rounding.
Let us first review the previous work.
To prove Theorem~\ref{t:swap}, 
Allen-Zhu, Li, Singh, and Wang~\cite{AZLSW20} analyzed a local search algorithm where they start from an arbitrary subset $S_0$ of $k$ vectors, and in each iteration $t \geq 1$ they find a pair of vectors $i \in S_{t-1}$ and $j \notin S_{t-1}$ so that roughly speaking $\lambda_{\min}(\sum_{l \in S_{t-1} - i + j} v_l v_l^T) > \lambda_{\min}(\sum_{l \in S_{t-1}} v_l v_l^T)$, and then they set $S_t = S_{t-1} - i_{t} + j_{t}$.
Using the framework of regret minimization, with the $l_{1/2}$-regularizer introduced in~\cite{AZLO15}, they proved that the task of finding a pair to improve the minimum eigenvalue can be reduced to finding a pair $i_t \in S_{t-1}$ and $j_t \notin S_{t-1}$ so that 
\begin{equation} \label{eq:progress}
\frac{\inner{v_{j_t}v_{j_t}^T}{A_t}}{1+2\alpha\inner{v_{j_t}v_{j_t}^T}{A_t^{1/2}}} - 
\frac{\inner{v_{i_t}v_{i_t}^T}{A_t}}{1-2\alpha\inner{v_{i_t}v_{i_t}^T}{A_t^{1/2}}} \geq \Delta > 0,
\end{equation}
where $A_t$ is the matrix defined in~(\ref{e:closed-form}) based on the current solution $S_{t-1}$.
Using a delicate argument, they proved that if $i_t \in S_{t-1}$ (subjecting to the restriction that $2\alpha \inner{v_i v_i^T}{A_t^{1/2}} < 1$) is chosen to minimize $\inner{v_iv_i^T}{A_t}/\big(1-2\alpha\inner{v_iv_i^T}{A_t^{1/2}}\big)$ and $j_t \notin S_{t-1}$ is chosen to maximize $\inner{v_jv_j^T}{A_t}/\big(1+2\alpha\inner{v_jv_j^T}{A_t^{1/2}}\big)$, then this pair $i_t,j_t$ satisfies the above inequality with $\Delta = \eps/k$ as long as $\lambda_{\min}(\sum_{l \in S_{t-1}} v_l v_l^T) \leq 1 - 3\eps$.
This implies, by the regret minimization framework, that the local search algorithm will succeed to find a solution $S_\tau$ with $\lambda_{\min}(\sum_{l \in S_\tau} v_l v_l^T) \geq 1-3\eps$ within $\tau \leq k/\eps$ iterations. 
We will review more about the regret minimization framework in Section~\ref{s:regret-min}.

To incorporate non-negative linear constraints, our idea is to turn the deterministic local search algorithm into an iterative randomized rounding algorithm.  
In this randomized rounding algorithm, we first construct an initial solution $S_0$ by adding each $i$ into $S_0$ with probability $x_i$ independently.
This will ensure that $c(S_0) \approx \inner{c}{x}$ with high probability.
In each iteration $t \geq 1$, based on the current solution $S_{t-1}$, we construct a probability distribution to sample a vector $v_{i_t}$ to be removed from $S_{t-1}$, and a probability distribution to sample a vector $v_{j_t}$ to be added to $S_{t-1}$.
To maintain $c(S_t) \approx \inner{c}{x}$, the basic idea is to remove a vector $v_i$ with probability proportional to $1-x_i$ and add a vector $v_j$ with probability proportional to $x_j$, but doing so will not satisfy the spectral lower bound with high probability.
Instead, we prove that if we recompute the sampling probability so that a vector $v_i$ is removed with probability proportional to $(1-x_i)(1-2\alpha\inner{v_iv_i^T}{A_t^{1/2}})$ and a vector $v_j$ is added with probability proportional to $x_j(1+2\alpha\inner{v_iv_i^T}{A_t^{1/2}})$, then \eqref{eq:progress} is satisfied with expected progress $\E[\Delta] \geq \eps/k$ as long as $\lambda_{\min}(\sum_{l \in S_{t-1}} v_l v_l^T) \leq 1-2\eps$.
Informally, a vector pointing to a direction that is not well covered by the current solution is more likely to be added and less likely to be removed, to ensure that the spectral lower bound will be satisfied.
However, this changes the expectation on the linear constraint, 
but we can bound the error by the additive term $O(n c_{\infty}/\eps)$.
Note that there are examples showing that this additive error is unavoidable if our goal is to satisfy the spectral lower bound exactly (see Section~\ref{ss:integrality}), so our analysis is tight up to a constant factor.
Compared to the deterministic approach in~\cite{AZLSW20}, this randomized approach uses the fractional solution $x$ more crucially in the rounding procedure, and we note that it can be used to give a simpler proof of the deterministic local search algorithm in~\cite{AZLSW20} (see Remark~\ref{r:simple}).

The advantage of the randomized approach is that we can prove that the random variables are concentrated around their expected values, so that we can handle multiple non-negative linear constraints simultaneously.
Since the sampling probabilities change over time based on the previous samples, the random variables that we consider are not a sum of independent random variables and thus Chernoff type bounds cannot be applied.
For the spectral lower bound, we will define a martingale and use Freedman's inequality to prove that the total progress we make in \eqref{eq:progress} is concentrated around its expected value.
For the non-negative linear constraints, we show that they satisfy an interesting ``self-adjusting'' property, such that if $c(S_t) - \inner{c}{x}$ is (more) positive then $\E[c(S_{t+1})] - c(S_t)$ is (more) negative and vice versa, so intuitively $c(S_t) \approx \inner{c}{x}$ with high probability for any $t$.
This sequence of random variables is not a martingale and so Freedman's inequality cannot be applied.
Instead, we prove a new concentration inequality for this self-adjusting process that provides a quantitative bound similar to that in Freedman's inequality.
We note that the iterative randomized rounding algorithm does not even need to know the linear constraint $c$ in advance in order to return a solution $S$ with $c(S) \approx \inner{c}{x}$.
This property is quite similar to that of a recent rounding algorithm by Bansal~\cite{Ban19} combining iterative rounding and randomized rounding as we will discuss in Section~\ref{ss:Bansal}.

We remark that our approach to turn a deterministic algorithm into a randomized algorithm is inspired by the fast algorithm for spectral sparsification by Lee and Sun~\cite{LS18}, where they turned the deterministic algorithm by Batson, Spielman and Srivastava~\cite{BSS12} into a randomized algorithm that recomputes the sampling probabilities in different phases.
In their algorithm, the advantage of the randomized algorithm is to sample many vectors in parallel instead of carefully choosing one vector at a time as in~\cite{BSS12}.
In our algorithm, the advantage of the randomized algorithm is to approximately preserves many linear constraints simultaneously using arguments about expectation and concentration, while it is not clear how to modify the proofs in the deterministic local search algorithm in~\cite{AZLSW20} to prove that there is always a pair of vectors $v_i, v_j$ which makes enough progress in \eqref{eq:progress} and at the same time $c_j-c_i$ is small, even if there is just have one constraint $c$ and it is given in advance.
We believe that this probabilistic approach will be useful in designing algorithms using the regret minimization framework.

\subsection{Other Applications} \label{s:applications-intro}

The spectral rounding results are quite general and have many other applications besides network design.
We mention some of these results and defer the details to Section~\ref{s:applications}.

\subsubsection{Weighted Experimental Design} \label{ss:experimental-intro}

Experimental design is an important class of problems in statistics and has found new applications in machine learning~\cite{Ang88,Puk06}.
The one-sided spectral rounding result of Allen-Zhu, Li, Singh and Wang~\cite{AZLSW20} was used to give near optimal approximation algorithms for many well-known experimental design problems.
We will explain these previous work in Section~\ref{s:experimental},
and show that our results can be used to design approximation algorithms for the more general setting where different experiments may have different costs
while incorporating some additional linear constraints;
see Theorem~\ref{t:exact-design} and the discussions thereafter.

\begin{theorem}[Informal] \label{t:exact-design-intro}
We are given $m$ design points that are represented by $n$-dimensional vectors $v_1, ..., v_m \in \R^n$, a cost vector $c \in \R^m_+$ and a cost budget $C \in \R_+$. 
For any $\eps \in (0, \frac12]$, if $C \geq 15 nc_{\infty} / \eps^2$,
there is a randomized polynomial time algorithm that returns a subset of vectors with total cost at most $C$ so that the objective value of A/D/E/V/G-design is at most $(1+O(\eps))$ times of that of the optimal solution.
\end{theorem}

\subsubsection{Spectral Network Design}

There are several previous work on network design problems with spectral requirements, including maximizing algebraic connectivity~\cite{GB06,KMS+10}, minimizing total effective resistances~\cite{GBS08}, and network design for $s$-$t$ effective resistances~\cite{CLSWZ19}.
In Section~\ref{s:spectral}, we will see that these problems are special cases of the general network design problem and the weighted experimental design problem in Section~\ref{s:network} and Section~\ref{s:experimental}, and our results provide improved approximation algorithms for these problems and also generalize these problems to incorporate many additional constraints.

We provide the first non-trivial approximation algorithm for the problem of maximizing algebraic connectivity subject to a knapsack constraint, proposed by Ghosh and Boyd~\cite{GB06}.

\begin{theorem} \label{t:lambda1-intro}
Let $G=(V,E)$ be a graph where each edge has cost $c_e$
and $C$ be a given cost budget.
Suppose $C \geq 15|V|c_{\infty}/\eps^2$ for some $\eps \leq 1/2$.
There is a randomized polynomial time algorithm which returns a subgraph $H$ of $G$ with 
\[\sum_{e \in H} c_e \leq C 
\quad {\rm and} \quad
\lambda_2(L_H) \geq (1-O(\eps)) \lambda_{\opt},
\]
where $\lambda_{\opt}$ is the maximum $\lambda_2$ that can be achieved by a solution with cost at most $C$.
\end{theorem}

We also provide a similar result for the problem of minimizing total effective resistance, proposed by Ghosh, Boyd and Saberi~\cite{GBS08}.

\begin{theorem} \label{t:Reff-intro}
Let $G=(V,E)$ be a graph where each edge has cost $c_e$ 
and $C$ be a given cost budget.
Suppose $C \geq 15 |V|c_{\infty}/\eps^2$ for some $\eps \leq 1/2$.
There is a randomized polynomial time algorithm which returns a subgraph $H$ of $G$ with 
\[\sum_{e \in H} c_e \leq C 
\quad {\rm and} \quad 
\sum_{u,v} \Reff_H(u,v) \leq (1+O(\eps)) \opt,
\] 
where $\opt$ is the minimum total effective resistance that can be achieved by a solution with cost at most $C$.
\end{theorem}

These results can be extended to incorporate additional constraints (e.g. connectivity constraints).
See Section~\ref{s:spectral} for details about these results,
including the related work~\cite{KMS+10,NST19}.

\subsubsection{Additive Spectral Sparsification}

Recently, Bansal, Svensson and Trevisan~\cite{BST19} study whether there is a non-trivial notion of unweighted spectral sparsification with which linear-sized spectral sparsification is always possible.
They provide randomized and deterministic algorithms to construct ``additive'' unweighted spectral sparsifiers, a notion suggested by Oveis Gharan.
In Section~\ref{s:additive}, we will explain their results and show that our spectral rounding results can be applied to this problem.
Using Theorem~\ref{t:two-sided}, we prove an optimal existential result for the problem.

\begin{theorem} \label{t:additive-cost-intro}
Suppose we are given a graph $G=(V,E)$ with $n$ vertices, $m$ edges, and maximum degree $d$.
Let $\tilde{m} = n/\eps^2$.
For any $\eps \in (0,1)$, there exists a subset of edges $F \subseteq E$ with $|F| \leq 8n/\eps^2 $ such that
\[
-8\sqrt{2} \eps d I_n \preceq L_G - \frac{m}{\tilde{m}} \sum_{e \in F} b_e b_e^T \preceq 8\sqrt{2} \eps d I_n.
\]
\end{theorem}

Using the proof techniques in Theorem~\ref{t:zero-one}, we provide an improved deterministic algorithm to construct additive unweighted spectral sparsifiers with no parallel edges (where the result in~\cite{BST19} may produce parallel edges).

\begin{theorem} \label{t:additive-deterministic-intro}
Given a graph $G=(V,E)$ with $n$ vertices, $m$ edges, maximum degree $d$, and $\eps \in (0,1)$, there is a polynomial time deterministic algorithm that finds a subset $F$ of edges with size $\tilde{m} = |F| = O(n/\eps^2)$ such that $\tilde{G} = (V,F)$ satisfies
\[
\frac{2m}{\tilde{m}} D_{\tilde{G}} - 2 D_G -\eps d I \preceq \frac{m}{\tilde{m}} L_{\tilde{G}} - L_G \preceq \eps d I,
\]
where $D_G$ and $D_{\tilde{G}}$ are the diagonal degree matrix of $G$ and $\tilde{G}$ respectively.
\end{theorem}

\section{Preliminaries}

We review some basic linear algebra and spectral graph theory in Section~\ref{s:algebra} and Section~\ref{s:graphs}.
Then we review the regret minimization framework for one-sided spectral rounding in Section~\ref{s:regret-min}, and state some concentration inequalities for the analysis of our randomized algorithm in Section~\ref{s:martingale}.

\subsection{Linear Algebra} \label{s:algebra}

We write $\R$ and $\R_+$ as the sets of real numbers and non-negative real numbers, and $\Z$ and $\Z_+$ as the sets of integers and non-negative integers.

All the vectors in this paper only have real entries.
Let $\R^n$ denote the $n$-dimensional Euclidean space. 
We write $\vec{1}_n$ as the $n$-dimensional all-one vector.
Given a vector $x$, we write $\norm{x}$ as its $\ell_2$-norm, $\norm{x}_1$ as its $\ell_1$-norm, and $\norm{x}_\infty$ as its $\ell_\infty$-norm. 
A vector $v \in \R^n$ is a column vector, and its transpose is denoted by $v^T$.
Given two vectors $x,y \in \R^n$, the inner product is defined as $\inner{x}{y} := \sum_{i=1}^n x_i y_i$.
The Cauchy-Schwarz inequality says that $\inner{x}{y} \leq \norm{x} \norm{y}$.

We write $I_n$ as the $n \times n$ identity matrix, and $J_n$ as the $n \times n$ all-one matrix.
All matrices considered in this paper are real symmetric matrices.
It is a fundamental result that any $n \times n$ real symmetric matrix has $n$ real eigenvalues $\lambda_1 \leq \ldots \leq \lambda_n$ and an orthonormal basis of eigenvectors.
We write $\lambda_{\max}(M)$ and $\lambda_{\min}(M)$ as the maximum and the minimum eigenvalue of a matrix $M$.
The trace of a matrix $M$, denoted by $\tr(M)$, is defined as the sum of the diagonal entries of $M$.
It is well-known that $\tr(M) = \sum_{i=1}^n \lambda_i(M)$ where $\lambda_i(M)$ denotes the $i$-th eigenvalue of $M$.

A matrix $M$ is a positive semidefinite (PSD) matrix, denoted as $M \succeq 0$, if $M$ is symmetric and all the eigenvalues are nonnegative, or equivalently, the quadratic form $x^T M x \geq 0$ for any vector $x$. 
We use $A \succeq B$ to denote $A - B \succeq 0$ for matrices $A$ and $B$. 
We write $\mathbb{S}^n_+$ as the set of all $n$-dimensional PSD matrices.
Let $M \succeq 0$ be a PSD matrix with eigendecomposition $M = \sum_i \lambda_i v_i v_i^T$, where $\lambda_i \geq 0$ is the $i$-th eigenvalue and $v_i$ is the corresponding eigenvector. 
The square root of $M$ is $M^{1/2}:= \sum_i \sqrt{\lambda_i} v_i v_i^T$.

Given two matrices $A$ and $B$ of the same size, the Frobenius inner product of $A, B$ is denoted as $\langle A, B \rangle := \sum_{i,j} A_{ij} B_{ij} = \tr(A^T B)$. The following are two standard facts
\[
A, B \succeq 0 \quad \Longrightarrow \quad \langle A, B \rangle \geq 0 \qquad \text{and} \qquad A \succeq 0, B \succeq C \succeq 0 \quad \Longrightarrow \quad \langle A, B \rangle \geq \langle A, C \rangle.
\]

We write $\norm{M}_{\rm op} := \max_{\norm{x}=1} \norm{Mx}$ as the operator norm of a matrix $M$. 
For symmetric matrices, the operator norm is just the largest absolute value of its eigenvalues.
For positive semidefinite matrices, the operator norm is just its largest eigenvalue.

\subsection{Graphs and Laplacian Matrices} \label{s:graphs}

Let $G=(V,E)$ be an undirected graph with edge weight $x_e\geq 0$ on each edge $e \in E$.
The number of vertices and the number of edges are denoted by $n:=|V|$ and $m:=|E|$.
For a subset of edges $F \subseteq E$, the total weight of edges in $F$ is $x(F) := \sum_{e \in F} x_e$.
For a subset of vertices $S \subseteq V$, the set of edges with one endpoint in $S$ and one endpoint in $V-S$ is denoted by $\delta(S)$.
For a vertex $v$, the set of edges incident on a vertex $v$ is $\delta(v):=\delta(\{v\})$, and the weighted degree of $v$ is $\deg(v) := x(\delta(v))$.
The expansion of a set $\phi(S) := |\delta(S)|/|S|$ is defined as the ratio of the number of edges on the boundary of $S$ to the size of $S$. The expansion of a graph $G$ is defined as $\phi(G) := \min_{0 \leq |S| \leq \frac{n}{2}} \phi(S)$.

The adjacency matrix $A \in \R^{n \times n}$ of the graph is defined as $A_{u,v} = x_{u,v}$ for all $u,v \in V$.
The Laplacian matrix $L \in \R^{n \times n}$ of the graph is defined as $L = D - A$ where $D \in \R^{n \times n}$ is the diagonal degree matrix with $D_{u,u} = \deg(u)$ for all $u \in V$. Similarly, the signless Laplacian matrix $L^+ \in \R^{n \times n}$ is defined as $L^+ = D + A$.
For each edge $e = uv \in E$, let $b_e := \chi_u - \chi_v$ where $\chi_u \in \R^n$ is the vector with one in the $u$-th entry and zero otherwise. 
The Laplacian matrix with respect to weights $x$ can be written as 
\[
L_x := \sum_{e \in E} x_e b_e b_e^T.
\]
Let $\lambda_1 \leq \lambda_2 \leq \ldots \leq \lambda_n$ be the eigenvalues of $L$ with corresponding orthonormal eigenvectors $v_1, v_2, \ldots, v_n$
so that $L = \sum_{i=1}^n \lambda_i v_i v_i^T$.
It is well-known that the Laplacian matrix is positive semidefinite, 
$\lambda_1=0$ with $v_1 = \vec{1}/\sqrt{n}$ as the corresponding eigenvector,
and $\lambda_2 > 0$ if and only if $G$ is connected.
The following fact is useful for eigenvalue maximization.
\begin{fact}[\cite{GB06}]
$\lambda_2(L_x)$ is a concave function with respect to $x$ for $x \geq 0$.
\end{fact}
The pseudo-inverse of the Laplacian matrix $L$ of a connected graph is defined as 
\[L^\dagger = \sum_{i=2}^n \frac{1}{\lambda_i} v_i v_i^T,\]
which maps every vector $b$ orthogonal to $v_1$ to a vector $y$ such that $Ly=b$.
The effective resistance between two vertices $s$ and $t$ on a graph $G$ with weight $x$ is defined as
\[
\Reff_x(s,t) := b_{st}^T L_x^{\dagger} b_{st}.
\]
We will use the following fact for the formulation of the convex programming relaxation in \eqref{eq:P-intro}.
\begin{fact}[\cite{GBS08}]
$\Reff_x(s,t)$ is a convex function with respect to the weights $x$ for $x \geq 0$.
\end{fact}

\subsection{Regret Minimization and Spectral Rounding} \label{s:regret-min}

We use the regret minimization framework developed by Allen-Zhu, Liao and Orecchia for spectral sparsification~\cite{AZLO15} and present the results in~\cite{AZLO15,AZLSW20}.
This is an online optimization setting.
In each iteration $t$, the player chooses an action matrix $A_t$ from the set of density matrices $\Delta_n = \{A \in \R^{n \times n} \mid A \succeq 0, \tr(A)=1\}$.
We can intrepret the player action as choosing a probability distribution over the set of unit vectors.
The player then observes a feedback matrix $F_t$ and incurs a loss of $\inner{A_t}{F_t}$.
After $\tau$ iterations, the regret of the player is defined as 
\[
R_\tau := \sum_{t=1}^\tau \inner{A_t}{F_t} - \inf_{B \in \Delta_n} \sum_{t=1}^\tau \inner{B}{F_t} 
= \sum_{t=1}^\tau \inner{A_t}{F_t} - \lambda_{\min}\Bigg(\sum_{t=1}^\tau F_t\Bigg),
\]
which is the difference between the loss of the player actions and the loss of the best fixed action $B$, that can be assumed to be a rank one matrix $vv^T$.
The objective of the player is to minimize the regret.
A well-known algorithm for regret minimization is Follow-The-Regularized-Leader which plays the action
\[A_t = \argmin_{A \in \Delta_n} \left\{ w(A) + \alpha \cdot \sum_{l=0}^{t-1} \inner{A}{F_l} \right\},
\]
where $w(A)$ is a regularization term and $\alpha$ is a parameter called the learning rate that balances the loss and the regularization. Note that $F_0$ is an initial feedback which is given before the game started.
Different choice of regularization gives different algorithm for regret minimization.
One choice is the entropy regularizer $w(A) = \inner{A}{\log A - I}$ and this gives the well-known matrix multiplicative update algorithm.
The choice that we will use is the $\ell_{1/2}$-regularizer $w(A) = -2\tr(A^{1/2})$ introduced in~\cite{AZLO15}, which plays the action
\begin{equation} \label{e:closed-form}
A_t = \left( l_t I + \alpha \sum_{l=0}^{t-1} F_l \right)^{-2},
\end{equation}
where $l_t$ is the unique constant that ensures $A_t \in \Delta_n$.
Allen-Zhu, Liao and Orecchia~\cite{AZLO15} prove upper bounds on the regret of this algorithm for positive or negative semidefinite feedback matrices.

\begin{theorem}[Theorem 3.2 and 3.3 in \cite{AZLO15}] \label{t:regret-psd-nsd}
Suppose $F_0 = 0$ and each feedback matrix $F_t \in \R^{n \times n}$ is either a positive or negative semidefinite matrix with $\alpha A_t^{1/4} F_t A_t^{1/4} \succcurlyeq - \frac14 I$ for all $t \geq 1$, and the action matrix $A_t \in \R^{n \times n}$ is of the form in~(\ref{e:closed-form}) for some $\alpha > 0$.
Then
\[
R_\tau \leq O(\alpha) \sum_{t=1}^\tau \inner{A_t}{|F_t|} \cdot \left\| A_t^{1/4} F_t A_t^{1/4} \right\|_{\rm op} + \frac{2\sqrt{n}}{\alpha}.
\]
When each feedback matrix $F_t$ is of the form $u_t u_t^T$ for some $u_t \in \R^n$ for all $t \geq 1$, it holds that
\begin{equation} \label{e:regret-rank-one}
\lambda_{\min}\left(\sum_{t=1}^\tau u_t u_t^T\right) 
\geq \sum_{t=1}^\tau \frac{\inner{u_t u_t^T}{A_t}}{1+\alpha \inner{u_t u_t^T}{A_t^{1/2}}} - \frac{2\sqrt{n}}{\alpha}.
\end{equation}
\end{theorem}

For one-sided spectral rounding, the goal is to choose a subset $S$ of vectors to maximize $\lambda_{\min}(\sum_{i \in S} v_i v_i^T)$.
Using this regret minimization framework,
the second part of Theorem~\ref{t:regret-psd-nsd} reduces this problem to the simpler task of finding a vector $u_t$ that maximizes $\inner{u_t u_t^T}{A_t} / (1+\alpha \inner{u_t u_t^T}{A_t^{1/2}})$.
Using the condition that $\sum_{i=1}^m x_i v_i v_i^T = I_n$ and $\sum_{i=1}^m x_i = k$, it can be shown~\cite{AZLSW17c} that there is always a vector $v_j$ with $\inner{v_j v_j^T}{A_t} / (1+\alpha \inner{v_j v_j^T}{A_t^{1/2}}) \geq 1/(k + \alpha \sqrt{n})$.
Setting $\alpha = \sqrt{n}/\eps$ and $\tau=k$ and using the assumption that $k \geq n/\eps^2$, this gives $\lambda_{\min}(\sum_{t=1}^k u_t u_t^T) \geq 1-3\eps$ and proves Theorem~\ref{t:swap} in the easier setting when a vector can be chosen more than once (i.e.~the with repetition setting in experimental design).
This greedy algorithm can be extended to the more difficult setting when every vector can be chosen at most once, but only achieving a $\Theta(1)$-approximation~\cite{AZLSW17c}.

To prove Theorem~\ref{t:swap} when the output must be a zero-one solution, Allen-Zhu, Li, Singh and Wang~\cite{AZLSW20} derived the following regret minimization bound for rank two feedback matrices.

\begin{theorem}[Lemma 2.5 and 2.7 in~\cite{AZLSW20}] \label{t:regret-rank-two}
Suppose the action matrix $A_t \in \R^{n \times n}$ is of the form in \eqref{e:closed-form} for some $\alpha > 0$. 
Suppose the initial feedback matrix $F_0 \in {\mathbb S}^n$ is a symmetric matrix, and for all $t \geq 1$ each feedback matrix $F_t$ is of the form $v_{j_t} v_{j_t}^T - v_{i_t} v_{i_t}^T$ for some $v_{j_t}, v_{i_t} \in \R^n$ such that $\alpha \inner{v_{i_t} v_{i_t}^T}{A_t^{1/2}} < \frac12$, then
\begin{equation*} 
\lambda_{\min}\left( \sum_{t=0}^\tau F_t \right) 
\geq \sum_{t=1}^\tau \left( 
\frac{\inner{v_{j_t} v_{j_t}^T}{A_t} }{1 + 2\alpha \inner{v_{j_t} v_{j_t}^T}{A_t^{1/2}} } 
- \frac{\inner{v_{i_t} v_{i_t}^T }{A_t}}{1 - 2\alpha \inner{v_{i_t} v_{i_t}^T}{A_t^{1/2}}} \right)
- \frac{2\sqrt{n}}{\alpha}.
\end{equation*}
\end{theorem}

With Theorem~\ref{t:regret-rank-two},
they analyzed a deterministic local search algorithm where they start from an arbitrary solution $S_0$ of $k$ vectors, and in each iteration $t \geq 1$ they find a $j_t \notin S_{t-1}$ that maximizes $\inner{v_jv_j^T}{A_t}/\big(1+2\alpha\inner{v_jv_j^T}{A_t^{1/2}}\big)$ and an $i_t \in S_{t-1}$ that minimizes $\inner{v_iv_i^T}{A_t}/\big(1-2\alpha\inner{v_iv_i^T}{A_t^{1/2}}\big)$ subjecting to the restriction that $2\alpha \inner{v_i v_i^T}{A_t^{1/2}} < 1$,
and define $S_t := S_{t-1} - i_t + j_t$ as the new solution.
Using a delicate argument, they proved that so long as $\lambda_{\min}(\sum_{l \in S_{t-1}} v_l v_l^T) \leq 1 - 3\eps$, the pair $i_t,j_t$ always satisfies
\begin{equation*}
\frac{\inner{v_{j_t}v_{j_t}^T}{A_t}}{1+2\alpha\inner{v_{j_t}v_{j_t}^T}{A_t^{1/2}}} - 
\frac{\inner{v_{i_t}v_{i_t}^T}{A_t}}{1-2\alpha\inner{v_{i_t}v_{i_t}^T}{A_t^{1/2}}} \geq \frac{\eps}{k}.
\end{equation*}
This implies by Theorem~\ref{t:regret-rank-two} that the local search algorithm will succeed in finding a solution $S_t$ with $\lambda_{\min}(\sum_{l \in S_t} v_l v_l^T) \geq 1-3\eps$ within $k/\eps$ iterations.

One technical point used in~\cite{AZLSW17c,AZLSW20} is that the partial solution $Z_{t-1} := \sum_{l=0}^{t-1} F_l$ at time $t$ and the action matrix $A_t$ at time $t$ have the same eigenbasis due to (\ref{e:closed-form}).
This allows one to bound $\inner{Z_{t-1}}{A_t}$ and $\inner{Z_{t-1}}{A_t^{1/2}}$ as follows.

\begin{lemma}[Claim 2.11 in~\cite{AZLSW20}] \label{l:cospectral}
Let $Z \succeq 0$ be an $n \times n$ positive semidefinite matrix
and $A = (\alpha Z + lI)^{-2}$ for some $\alpha>0$ where $l$ is the unique constant such that $A$ is a density matrix.
Then, it holds that
\[
\inner{Z}{A} \leq \frac{\sqrt{n}}{\alpha} + \lambda_{\min}(Z) \qquad \text{and} \qquad \alpha \inner{Z}{A^{1/2}} \leq n + \alpha \sqrt{n} \cdot \lambda_{\min}(Z).
\]
\end{lemma}

This lemma will be used in constructing a zero-one solution for Theorem~\ref{t:zero-one}.
It will also be used in strengthening the result in~\cite{BST19} to guarantee that the unweighted additive spectral sparsifier returned by the regret minimization algorithm has no parallel edges.

\subsection{Martingale and Concentration Inequalities} \label{s:martingale}

A sequence of random variables $Y_1, \ldots, Y_\tau$ is a martingale with respect to a sequence of random variables $Z_1, \ldots, Z_\tau$ if for all $t > 0$, it holds that 
\begin{enumerate}
\item $Y_t$ is a function of $Z_1, \ldots, Z_{t-1}$;
\item $\E[|Y_t|] < \infty$; 
\item $\E[Y_{t+1} | Z_1, \ldots, Z_t] = Y_t$. 
\end{enumerate}

We will use the following theorem by Freedman to bound the probability that $Y_\tau$ is large.

\begin{theorem}[\cite{Fre75,Tro11}] \label{t:Freedman}
Let $\{Y_t\}_t$ be a real-valued martingale with respect to $\{Z_t\}_t$,
and $\{X_t=Y_t - Y_{t-1}\}_t$ be the difference sequence. 
Assume that $X_t \leq R$ deterministically for $1 \leq t \leq \tau$. 
Let $W_t := \sum_{j=1}^t \E[X_j^2 | Z_1, ..., Z_{j-1}]$ for $1 \leq t \leq \tau$.
Then, for all $\delta \geq 0$ and $\sigma^2 > 0$,
\[
\Pr\left( \exists t \in [\tau]: Y_t \geq \delta~{\rm and}~W_t \leq \sigma^2 \right) \leq \exp\left( \frac{-\delta^2/2}{\sigma^2 + R\delta/3} \right).
\]
\end{theorem}

Recently, some variants of Freedman's inequality for martingales have been used to obtain algorithmic discrepancy results~\cite{BG17,BDGL19}. 
In this paper, for the analysis of the non-negative linear constraints,
we prove another variant which applies to non-martingales with a ``self-adjusting'' property, 
that if $Y_t$ is (more) positive then $E[Y_{t+1}]-Y_t$ is (more) negative and vice versa.
With this self-adjusting property, intuitively $Y_t$ cannot be too far away from zero, and the following theorem provides a quantitative bound that is similar to that in Freedman's inequality. 

\begin{theorem} \label{t:adjusting}
Let $\{Y_t\}_t$ be a sequence of random variables,
and $X_t := Y_t - Y_{t-1}$ be the difference sequence.
Suppose that there exist $\gamma \in (0,\frac12]$, $\beta_u, \beta_l \geq 0$ and  $\sigma > 0$ such that the following properties hold for all $t \geq 1$.
\begin{enumerate}
\item 
{\em (Bounded difference:)} $|X_t| \leq 1$ with probability one.
\item 
{\em (Self adjusting:)} 
$-\gamma Y_{t-1} - \beta_l \leq \E[X_t \mid Y_0, ..., Y_{t-1}] \leq -\gamma Y_{t-1} + \beta_u$.
\item 
{\em (Bounded variance:)} 
$\E[X_t^2 \mid Y_0, \ldots, Y_{t-1}] \leq \gamma Y_{t-1} + \sigma$.
\item 
{\em (Initial concentration:)} For any $a \in [-1,1]$, 
the random variable $Y_0$ satisfies
$\E\left[e^{a Y_0}\right] \leq e^{a^2 \sigma/\gamma}$. 
\end{enumerate}
Then, for any $\eta > 0$ and any $t \geq 0$, it holds that
\[
\Pr\left[ Y_t \geq \frac{\beta_u}{\gamma} + \eta\right]
\leq \exp\left[ -\frac{\eta^2}{4(\sigma+\beta_u)/\gamma + 2\eta}\right]
\]
and
\[
\Pr\left[ Y_t \leq -\frac{\beta_l}{\gamma} - \eta\right]
\leq \exp\left[ -\frac{\eta^2}{4\sigma/\gamma + \eta}\right].
\]
\end{theorem}

The proof of Theorem~\ref{t:adjusting} will be presented in Section~\ref{ss:concentration}.

\section{Spectral Rounding} \label{s:rounding}

We will first present the iterative randomized rounding algorithm for one-sided spectral rounding in Section~\ref{s:swap}. 
Then we will present the proof of Theorem~\ref{t:two-sided} for two-sided spectral rounding in Section~\ref{s:discrepancy}, and some examples showing the tightness of our results in Section~\ref{ss:integrality}.

\subsection{Iterative Randomized Rounding for One-Sided Spectral Rounding} \label{s:swap}

We modify the deterministic local search algorithm in~\cite{AZLSW20} to an iterative randomized rounding algorithm so as to approximately satisfy arbitrary non-negative linear constraints.
In this randomized algorithm, we first construct an initial solution $S_0$ by adding each vector $v_i$ into $S_0$ with probability $x_i$ independently.
In each iteration $t \geq 1$, based on the current solution $S_{t-1}$, we construct a probability distribution to sample a vector $v_{i_t}$ to be removed from $S_{t-1}$, and a probability distribution to sample a vector $v_{j_t}$ to be added to $S_{t-1}$.
The basic idea is that a vector $v_i$ is removed with probability proportional to $1-x_i$ and a vector $v_j$ is added with probability proportional to $x_j$, but the probability is also adjusted based on the vector's contribution to the minimum eigenvalue of the current solution. 
We remark that it is possible that no vector is removed and/or no vector is added in an iteration. 
The algorithm stops when the minimum eigenvalue of the current solution is at least $1-2\eps$.
The following is the formal description of the algorithm.

\begin{framed}
{\bf Iterative Randomized Swapping Algorithm}

Input: $v_1, ..., v_m \in \R^n$ and $x \in [0, 1]^m$ with $\sum_{i=1}^m x_i v_i v_i^T = I_n$, and an error parameter $\eps \in (0, \frac12)$.

Output: a subset $S \subseteq [m]$ such that $\sum_{i \in S} v_i v_i^T \succcurlyeq (1-2\eps) I_n$ and $c(S) \approx \inner{c}{x}$ for any $c \in \R_+^m$ with high probability.

    \begin{enumerate}
    
    \item Initialization: $t:= 1$, $S_0:=\emptyset$, $\alpha := \sqrt{n}/\eps$, $k := m+2n/\eps$.

    \item Add $i$ into $S_0$ independently with probability $x_i$ for each $i \in [m]$. Let $Z_0 := \sum_{i \in S_0} v_i v_i^T$.

    \item While $\lambda_{\min}(Z_{t-1}) < 1 - 2\eps$ do
    \begin{enumerate}

	\item Compute the action matrix $A_t := (\alpha Z_{t-1} - l_t I_n)^{-2}$, where $l_t \in \R$ is the unique value such that $A_t \succ 0$ and $\tr(A_t) = 1$.

        \item Define $S'_{t-1} := \{i \in S_{t-1}: 2\alpha \inner{v_i v_i^T}{ A_t^{1/2}} < \frac12 \}$.
	
	\item Sample $i_t$ from the following probability distribution:
	\[
	   \Pr\left( i_t = i \right) = \frac{1}{k} (1- x_i) (1 - 2 \alpha \langle v_i v_i^T, A_t^{1/2} \rangle) \quad {\rm for~} i \in S_{t-1}',
	\]
	and $\Pr\left( i_t = \emptyset \right) = 1 - \sum_{i \in S'_{t-1}} \Pr(i_t = i)$.

        \item Sample $j_t$ from the following probability distribution:
	    \[
           \Pr\left( j_t = j \right) = \frac{x_j}{k} (1 + 2 \alpha \langle v_j v_j^T, A_t^{1/2} \rangle) \quad {\rm for~} j \in [m] \backslash S_{t-1},
	    \]
	    and $\Pr\left( j_t = \emptyset \right) = 1 - \sum_{j \in [m] \backslash S_{t-1}} \Pr(j_t = j)$.

        \item Set $S_{t} := S_{t-1} \cup \{ j_t\} \backslash \{i_t\}$, $Z_t := \sum_{i \in S_t} v_i v_i^T$ and $t := t+1$.
    \end{enumerate}
    \item Return $S = S_{t-1}$ as the solution. 
\end{enumerate}
\end{framed}

Before we state the main result of this algorithm, we first check that the algorithm is well-defined.

\begin{claim}
The probability distributions in each iteration of the iterative randomized swapping algorithm are well-defined.
\end{claim}
\begin{proof}
To verify that the probability distribution for sampling $i_t$ is well-defined, we need to show that $\Pr(i_t = i) \geq 0$ for $i \in S_{t-1}'$ and $\sum_{i \in S_{t-1}'} \Pr(i_t = i) \leq 1$.
Since $A_t \succ 0$ and $x_i \in [0,1]$ and $2\alpha \inner{v_iv_i^T}{A_t^{1/2}}\leq 1/2$ for $i \in S_{t-1}'$, it follows that for $i \in S_{t-1}'$ we have
\[
0 \leq \Pr(i_t = i) = \frac{1}{k} (1-x_i)(1-2\alpha \inner{v_iv_i^T}{A_t^{1/2}}) \leq \frac{1}{k},
\]
and this implies that $\sum_{i \in S_{t-1}'} \Pr(i_t = i) \leq |S_{t-1}'|/k \leq m/k < 1$ by the definition of $k$.

Next we verify that the probability distribution for sampling $j_t$ is well-defined.
It is clear that $\Pr(j_t = j) \geq 0$ as $A_t \succ 0$ and $x_j \in [0,1]$.
We claim that 
$\sum_{j \in [m] \setminus S_{t-1}} \Pr(j_t = j) \leq \sum_{j \in [m]} \Pr(j_t = j) \leq 1$ as
\[\sum_{j \in [m]} \Pr(j_t = j)
= \frac{1}{k} \sum_{j=1}^m x_j (1+2\alpha \inner{v_j v_j^T}{A_t^{1/2}})
= \frac{1}{k} \left( \sum_{j=1}^m x_j + 2\alpha \tr(A_t^{1/2})\right)
\leq \frac{1}{k} \left(m  + \frac{2n}{\eps}\right) = 1,
\]
where the second equality is by the assumption that $\sum_{j=1}^m x_j v_j v_j^T = I_n$, the last equality is by the definition of $k$,
and the inequality uses that $x_j \in [0,1]$, $\alpha = \sqrt{n}/\eps$ and the bound that $\tr(A_t^{1/2}) \leq \sqrt{n}$. 
To see that $\tr(A_t^{1/2}) \leq \sqrt{n}$,
let $\lambda_1, \ldots, \lambda_n$ be the eigenvalues of $A_t$, then 
\begin{equation} \label{e:trace}
\tr(A_t^{1/2}) = \sum_{i=1}^n \sqrt{\lambda_i} \leq \sqrt{n} \sum_{i=1}^n \lambda_i = \sqrt{n} \tr(A_t) = \sqrt{n},
\end{equation} 
where the inequality is by Cauchy-Schwarz and the last equality is by the definition of $A_t$.
\end{proof}

\begin{remark} \label{r:probability}
The reader may wonder why we do not define the probability distribution for sampling $i_t$ by
\[
\Pr\left( i_t = i \right) 
= \frac{(1- x_i) (1 - 2 \alpha \langle v_i v_i^T, A_t^{1/2} \rangle)}{\sum_{j \in S_{t-1}'}(1- x_j) (1 - 2 \alpha \langle v_j v_j^T, A_t^{1/2} \rangle)} \quad {\rm for~} i \in S_{t-1}',
\]
so that $\sum_{i \in S_{t-1}'} \Pr(i_t=i) = 1$ and likewise for sampling $j_t$, so that we always remove a vector from $S_{t-1}$ and add another vector to $S_{t-1}$ in each iteration.
This is our initial approach and we believe that this should also work, 
but it turns out that the calculations for the linear constraints simplify considerably by having a common denominator $k$ for these two probability distributions.
\end{remark}

The following is the main technical result for one-sided spectral rounding.

\begin{theorem} \label{t:approx}
Suppose we are given $v_1, ..., v_m \in \R^n$, $x \in [0,1]^m$ such that $\sum_{i=1}^m x_i v_i v_i^T = I_n$. 
For any $\eps \in (0, \frac12)$, the iterative randomized swapping algorithm returns a subset $S \subseteq [m]$ satisfying
\[
\sum_{i \in S} v_i v_i^T \succcurlyeq (1-2\eps) I_n
\]
within $qk/\eps$ iterations with probability at least $1-\exp\left( -\Omega(q\sqrt{n}) \right)$ for $q \geq 2$. 
Furthermore, for any $c \in \R^m_+$ and any $\delta_1 \in [0,1]$, $\delta_2 \in [0,1]$ and $\delta_3 > 0$, the probability that the returned solution $S$ satisfies the cost upper bound is
\[
\Pr\left[ c(S) \leq (1+\delta_1) \inner{c}{x} + \frac{15 n c_{\infty}}{\eps} \right] \geq 1 - \exp\left[-\Omega\left( \frac{\delta_1 n}{\eps} \right) \right],
\]
and the probability that the returned solution $S$ satisfies the cost lower bound is
\[
\Pr\Big[ c(S) \geq (1-\delta_2) \inner{c}{x} - \delta_3 n c_{\infty} \Big]
\geq 1-\exp\Big[-\Omega\big(\min\{\delta_2 \delta_3, \eps \delta_3^2\} \cdot n\big)\Big].
\]
\end{theorem}

\begin{remark}
If we set $\delta_1 = \delta_2 = \eps$ and $\delta_3 = 1/\eps$,
then Theorem~\ref{t:approx} states that the returned solution $S$ satisfies
\[
(1-\eps)\inner{c}{x} - \frac{nc_{\infty}}{\eps} \leq c(S) \leq (1+\eps)\inner{c}{x} + \frac{15nc_{\infty}}{\eps}
\]
with probability at least $1-\exp(-\Omega(n))$ for any $c \in \R^m_+$.
We introduce $\delta_1, \delta_2, \delta_3$ to have a more refined control of the failure probability of the lower bound, and this will be relevant in showing that linear covering constraints can be almost satisfied.
\end{remark}

{\bf Organization:}
The remainder of this subsection is organized as follows.
We will first prove that the spectral lower bound will be approximately satisified with high probability within polynomial time in Section~\ref{ss:time}, and then prove the guarantees on the linear constraints in Section~\ref{ss:cost}.
Then, we will use Theorem~\ref{t:approx} to prove the exact one-sided spectral rounding result in Theorem~\ref{t:zero-one} in Section~\ref{ss:exact}.
Finally, we provide a proof of the concentration inequality in Theorem~\ref{t:adjusting} in Section~\ref{ss:concentration}.

\subsubsection{Bounding the Minimum Eigenvalue} \label{ss:time}

The goal in this subsection is to prove that the probability that the algorithm does not terminate within $\tau \geq qk/\eps$ iterations is at most $\exp(-\Omega(q\sqrt{n}))$ for $q \geq 2$.

We will bound the minimum eigenvalue of the solution using the regret minimization framework developed in~\cite{AZLO15,AZLSW20}.
The initial feedback matrix is $F_0 = Z_0$, which is constructed randomly using $x$.
In each iteration $t \geq 1$, after computing the action matrix $A_t$,
the algorithm responds with the feedback matrix $F_t = v_{j_t}v_{j_t}^T - v_{i_t}v_{i_t}^T$.
Note that $Z_{\tau} = \sum_{t=0}^\tau F_t$.
Define 
\[\Delta_t^+ :=  \frac{\inner{v_{j_t} v_{j_t}^T}{A_t} }{1 + 2\alpha \inner{v_{j_t} v_{j_t}^T}{A_t^{1/2}} }
\quad {\rm and} \quad
\Delta_t^- := \frac{\inner{v_{i_t} v_{i_t}^T }{A_t}}{1 - 2\alpha \inner{v_{i_t} v_{i_t}^T}{A_t^{1/2}}}
\quad {\rm and} \quad
\Delta_t := \Delta_t^+ - \Delta_t^-.
\]
Note that $2\alpha\inner{v_{i_t}v_{i_t}^T}{A_t^{1/2}} < \frac12 < 1$ for $1 \leq t \leq \tau$ by the definition of $S_{t-1}'$, and so $\Delta_t^-$ is well-defined for $1 \leq t \leq \tau$.
The regret minimization Theorem~\ref{t:regret-rank-two} proves that
\begin{equation} \label{e:regret-rank-two}
\lambda_{\min}\left( Z_{\tau} \right) 
= \lambda_{\min} \Big( \sum_{t=0}^{\tau} F_t \Big)
\geq \sum_{t=1}^\tau \Delta_t - \frac{2\sqrt{n}}{\alpha}
= \sum_{t=1}^\tau \Delta_t - 2 \eps.
\end{equation}
To lower bound the minimum eigenvalue,
we will prove that $\sum_{t=1}^\tau \Delta_t \geq 1$ with high probability.
In the following, we bound the expected value of $\sum_{t=1}^\tau \Delta_t$, 
and then use Freeman's martingale inequality to bound the probability that $\sum_{t=1}^\tau \Delta_t$ deviates significantly from its expected value.

\begin{lemma} \label{l:spec-exp-t}
Let $\lambda := \max_{0 \leq t \leq \tau} \lambda_{\min}(Z_t)$.
Then
\[
\sum_{t=1}^\tau 
\E \left[ \Delta_t \mid S_{t-1}\right] 
\geq \sum_{t=1}^\tau \frac{1}{k} (1-\eps-\lambda_{\min}(Z_{t-1}))
\geq \frac{\tau}{k} (1-\eps-\lambda).
\]
\end{lemma} 
\begin{proof}
We first consider the expected gain of adding the vector $j_t$.
By the definition of the probability distribution of $j_t$,
\begin{equation} \label{eq:add-spec-exp}
\begin{aligned}
    \E\left[ \Delta_t^+ \mid S_{t-1} \right] 
& = \frac{1}{k} \sum_{j \in [m] \backslash S_{t-1}} x_j (1+2\alpha \inner{v_j v_j^T}{A^{1/2}_t}) \cdot \frac{\inner{v_{j} v_{j}^T}{A_t}}{1 + 2 \alpha \inner{v_{j} v_{j}^T}{A_t^{1/2}}} 
\\
& = \frac{1}{k} \sum_{j \in [m] \backslash S_{t-1}} x_j \inner{v_j v_j^T}{A_t} 
\\ & = \frac{1}{k} \Big( 1 - \sum_{j \in S_{t-1}} x_j \inner{v_j v_j^T}{A_t} \Big),
\end{aligned}
\end{equation}
where the last equality is by $\sum_{j=1}^m x_j v_j v_j^T = I_n$ and $\tr(A_t)=1$ by the definition of $A_t$.

Then we consider the expected loss of removing the vector $i_t$.
By the definition of the probability distribution of $i_t$,
\begin{equation} \label{eq:del-spec-exp}
\begin{aligned} 
    \E \left[ \Delta_t^- \mid S_{t-1} \right] 
& = \sum_{i \in S'_{t-1}} \frac{1}{k} (1-x_i) (1-2\alpha \inner{v_i v_i^T}{A_t^{1/2}}) \cdot \frac{\inner{v_i v_i^T}{A_t}}{1-2\alpha \inner{v_i v_i^T}{A_t^{1/2}}} \\
    & = \frac{1}{k} \sum_{i \in S'_{t-1}} (1-x_i) \inner{v_i v_i^T}{A_t} \\
    & \leq \frac{1}{k} \sum_{i \in S_{t-1}} (1-x_i) \inner{v_i v_i^T}{A_t} \\
    & \leq \frac{1}{k} \Big( \lambda_{\min}(Z_{t-1}) + \eps - \sum_{i \in S_{t-1}} x_i \inner{v_i v_i^T}{A_t} \Big),
\end{aligned}
\end{equation}
where the first inequality is because $x_i \in [0,1]$ and $\inner{v_iv_i^T}{A_t} \geq 0$ as $A_t \succ 0$,
and the last inequality follows from Lemma~\ref{l:cospectral} that $\inner{Z_{t-1}}{A_t} \leq \sqrt{n}/\alpha + \lambda_{\min}(Z_{t-1})$ and $\alpha = \sqrt{n}/\eps$.

The lemma follows by combining (\ref{eq:add-spec-exp}) and (\ref{eq:del-spec-exp}) and summing over $t$ and using $\lambda = \max_{t} \lambda_{\min}(Z_t)$.
\end{proof}

\begin{remark} \label{r:simple}
If we use the probability distributions stated in Remark~\ref{r:probability},
then we can start with a solution with $l:=\sum_{i=1}^m x_i + O(n/\eps)$ vectors and guarantee that the solution at each iteration still has exactly $l$ vectors.
A similar statement about the expected progress as in Lemma~\ref{l:spec-exp-t} can be proved.
This implies that there exists a good pair $i_t \in S_{t-1}$ and $j_t \in S_{t-1}$, which gives a solution set of size $l$ satisfying the spectral lower bound approximately. 
Together with a preprocessing step as in Section~\ref{ss:exact}, 
this gives a simpler proof of the deterministic algorithm of~\cite{AZLSW20}.
\end{remark}

\begin{lemma} \label{l:spec-whp-t}
Let $\lambda := \max_{0 \leq t \leq \tau} \lambda_{\min}(Z_t)$.
Then, for any $\eta > 0$,
\[
\Pr \left[ \sum_{t=1}^\tau \Delta_t \leq \left( \sum_{t=1}^\tau \E[\Delta_t \mid S_{t-1}]\right) - \eta \right] 
\leq \exp\left( -\frac{\eta^2 k\sqrt{n}/2}{\tau \eps(1+\lambda+\eps) + \eta k\eps/3} \right).
\]
\end{lemma}
\begin{proof}
We define the following sequences of random variables where
$X_t := \E[\Delta_t \mid S_{t-1}] - \Delta_t$ and $Y_t := \sum_{l=1}^t X_l$.
Observe that $\{Y_t\}_t$ is a martingale with respect to $\{S_t\}_t$.
We use Freedman's inequality to bound $\Pr(Y_\tau \geq \eta)$.
To apply Freedman's inequality, we need to upper bound $X_t$ and $E[X_t^2 \mid S_{t-1}]$.
Note that
\[0 \leq \Delta_t^+ = \frac{\inner{v_{j_t} v_{j_t}^T}{A_t} }{1 + 2\alpha \inner{v_{j_t} v_{j_t}^T}{A_t^{1/2}} }
\leq \frac{\inner{v_{j_t} v_{j_t}^T}{A_t} }{2\alpha \inner{v_{j_t} v_{j_t}^T}{A_t^{1/2}} } \leq \frac{1}{2\alpha},
\]
where the last inequality holds as $0 \prec A_t \preceq I$.
Also,
\[
0 \leq \Delta_t^- = \frac{\inner{v_{i_t} v_{i_t}^T }{A_t}}{1 - 2\alpha \inner{v_{i_t} v_{i_t}^T}{A_t^{1/2}}}
\leq \frac{\inner{v_{i_t} v_{i_t}^T }{A_t^{1/2}}}{1 - 2\alpha \inner{v_{i_t} v_{i_t}^T}{A_t^{1/2}}} 
\leq \frac{1}{2\alpha},
\]
where the second last inequality is by $0 \prec A_t \preceq I$,
and the first and last inequality are because $i_t$ is chosen from the set $S'_{t-1} := \{i \mid 4\alpha\inner{v_iv_i^T}{A_t^{1/2}} < 1\}$.
(We remark that this upper bound on $\Delta_t^-$ is exactly the reason for the definition of $S_{t-1}'$.)
As these lower and upper bounds on $\Delta_t^+$ and $\Delta_t^-$ hold with probability one, we have the deterministic upper bound $X_t \leq R := \eps/\sqrt{n}$ as
\[
X_t = \E[\Delta_t \mid S_{t-1}] - \Delta_t
\leq \E[\Delta_t^+ \mid S_{t-1}] + \Delta_t^- \leq \frac{1}{\alpha} = \frac{\eps}{\sqrt{n}} = R.
\]
Next, we upper bound
\begin{equation*} 
\begin{aligned} 
\E[X^2_t \mid S_{t-1}] & \leq R \cdot \E[|X_t| \mid S_{t-1}]
\leq \frac{\eps}{\sqrt{n}} \Big(\E[\Delta^+_t \mid S_{t-1}] + \E[\Delta^-_t \mid S_{t-1}]\Big)
\leq \frac{\eps}{k\sqrt{n}} (1+\lambda+\eps),
\end{aligned}
\end{equation*}
where the last inequality follows from \eqref{eq:add-spec-exp} and \eqref{eq:del-spec-exp} that $\E[\Delta_t^+ \mid S_{t-1}] \leq 1/k$ and $\E[\Delta_t^- \mid S_{t-1}] \leq (\lambda+\eps)/k$.
Therefore, $W_\tau := \sum_{t=1}^\tau \E[X^2_t \mid S_{t-1}] \leq \tau \eps (1+\lambda+\eps) / (k\sqrt{n})$.
Applying Theorem~\ref{t:Freedman} with $R=\eps/\sqrt{n}$ and $\sigma^2 = \tau \eps (1+\lambda+\eps) / (k\sqrt{n})$,
it follows that
\[
\Pr(Y_\tau \geq \eta)
\leq \exp\left(-\frac{\eta^2/2}{\sigma^2 + R\eta/3}\right)
= \exp\left( -\frac{\eta^2 k\sqrt{n}/2}{\tau \eps(1+\lambda+\eps) + \eta k\eps/3} \right).
\]
The lemma follows as $Y_\tau \geq \eta$ is equivalent to 
$\sum_{t=1}^\tau \Delta_t \leq \left( \sum_{t=1}^\tau \E[\Delta_t \mid S_{t-1}]\right) - \eta$. 
\end{proof}

We are ready to prove that the algorithm terminates in a polynomial number of iterations with high probability.

\begin{theorem} \label{t:eigenvalue}
The probability that the iterative randomized swapping algorithm does not terminate in $qk/\eps$ iterations for $q \geq 2$ is at most $\exp(-\Omega(q\sqrt{n}))$.
\end{theorem}
\begin{proof}
Let $\tau = qk/\eps$.
Suppose $\lambda = \max_{0 \leq t \leq \tau} \lambda_{\min}(Z_t) < 1-2\eps$.
Then, Lemma~\ref{l:spec-exp-t} implies that 
\[
\sum_{t=1}^{\tau} \E \left[ \Delta_t \mid S_{t-1}\right] 
\geq \frac{\tau}{k}(1-\eps-\lambda)
= \frac{q}{\eps}(1-\eps-\lambda) > q,
\]
and the regret minimization bound in (\ref{e:regret-rank-two}) implies that 
\[
1 - 2\eps > \lambda_{\min}(Z_{\tau}) 
\geq \left(\sum_{t=1}^\tau \Delta_t\right) - 2\eps
\quad \implies \quad 
\sum_{t=1}^\tau \Delta_t < 1.
\]
Therefore,
\begin{equation*}
\begin{aligned}
\Pr\left[ \bigcap_{t=0}^{\tau} \Big(\lambda_{\min}(Z_t) < 1-2\eps \Big)\right]
& \leq
\Pr\left[ \sum_{t=1}^{\tau} \Delta_t < \left( \sum_{t=1}^{\tau} \E \left[ \Delta_t \mid S_{t-1}\right] \right) -(q-1) \right]
\\
& \leq \exp\left( -\frac{(q-1)^2 k\sqrt{n}/2}{(qk/\eps) \eps(1+(1-2\eps)+\eps) + (q-1) k\eps/3} \right)
\\
& \leq \exp(-\Omega(q\sqrt{n})),
\end{aligned}
\end{equation*}
where the second inequality is by Lemma~\ref{l:spec-whp-t} with $\eta=q-1$ and $\tau=qk/\eps$ and the last inequality is by the assumption that $q \geq 2$.
\end{proof}

So, for example, the probability that the algorithm does not terminate in $2k/\eps$ iterations is at most $\exp(-\Omega(\sqrt{n}))$ and the probability that it does not terminate in $k\sqrt{n}/\eps$ iterations is at most $\exp(-\Omega(n))$.

\subsubsection{Bounding the Linear Constraints} \label{ss:cost}

For an arbitrary non-negative linear constraint $c \in \R^m_+$,
the goal in this subsection is to prove that $c(S_t) \approx \inner{c}{x}$ with high probability for any $t$, where we recall that $c(S_t) := \sum_{i \in S_t} c_i$ is the ``cost'' of the solution at time $t$.
We first bound the expected change of the cost in an iteration.

\begin{lemma} \label{l:cost-1st-moment}
Suppose $\lambda_{\min}(Z_{t-1}) < 1$.
Then
\[
\frac{1}{k} \Big(\inner{c}{x} - c(S_{t-1})\Big) 
\leq \E[c_{j_t} - c_{i_t} \mid S_{t-1}]
\leq \frac{1}{k} \Big(\inner{c}{x} - c(S_{t-1}) + \frac{14nc_{\infty}}{\eps}\Big).
\]
\end{lemma}

\begin{proof}
We first bound the conditional expectation of $c_{j_t}$.
By the probability distribution of $j_t$,
\begin{eqnarray*}
\E[c_{j_t} \mid S_{t-1}] & = & 
\frac{1}{k} \sum_{j \in [m] \setminus S_{t-1}} c_j x_j (1+2\alpha \inner{v_j v_j^T}{A_t^{1/2}}) 
\\
& = & \frac{1}{k} \bigg(\inner{c}{x} - \sum_{j \in S_{t-1}} c_j x_j + 2\alpha \sum_{j \in [m] \backslash S_{t-1}} c_j x_j\inner{v_j v_j^T}{A_t^{1/2}} \bigg).
\end{eqnarray*}
Note that
\[
0 \leq  2\alpha \sum_{j \in [m] \backslash S_{t-1}} c_j x_j\inner{v_j v_j^T}{A_t^{1/2}} 
\leq 2\alpha c_{\infty} \sum_{j=1}^m x_j \inner{v_j v_j^T}{A_t^{1/2}} 
= 2\alpha c_{\infty} \tr(A_t^{1/2}) 
\leq \frac{2nc_{\infty}}{\eps},
\]
where the equality holds as $\sum_{j=1}^m x_j v_j v_j^T = I_n$ and the last inequality is by~\eqref{e:trace} and $\alpha=\sqrt{n}/\eps$. 
Therefore, 
\begin{equation} \label{eq:add-cost-exp}
\frac{1}{k} \Big( \inner{c}{x} - \sum_{i \in S_{t-1}} c_ix_i \Big)
\leq \E[c_{j_t} \mid S_{t-1}] 
\leq \frac{1}{k} \Big( \inner{c}{x}  - \sum_{i \in S_{t-1}} c_ix_i + \frac{2nc_{\infty}}{\eps}\Big).
\end{equation}
Next we bound the expectation of $c_{i_t}$.
By the probability distribution of $i_t$,
\begin{eqnarray*}
\E[c_{i_t} \mid S_{t-1}] 
& = & \frac{1}{k} \sum_{i \in S'_{t-1}} c_i (1-x_i)(1-2\alpha \inner{v_i v_i^T}{A_t^{1/2}})
\\ 
& = & \frac{1}{k} \Big(\sum_{i \in S'_{t-1}} c_i(1-x_i) - 2\alpha \sum_{i \in S'_{t-1}} c_i (1-x_i) \inner{v_i v_i^T}{A_t^{1/2}} \Big)
\\
& = & \frac{1}{k} \Big( c(S_{t-1}) - \sum_{i \in S_{t-1}} c_i x_i - \sum_{i \in S_{t-1} \backslash S'_{t-1}} c_i (1-x_i) - 2\alpha \sum_{i \in S'_{t-1}} c_i (1-x_i) \inner{v_i v_i^T}{A_t^{1/2}} \Big).
\end{eqnarray*}
We would like to bound the last two terms of the right hand side.
Recall that $S'_{t-1} := \{ i \in S_{t-1} \mid 4\alpha\inner{v_iv_i^T}{A_t^{1/2}} < 1\}$.
This implies that
\[
|S_{t-1} \setminus S'_{t-1}|
\leq \sum_{i \in S_{t-1} \setminus S'_{t-1}} 4\alpha \inner{v_i v_i^T}{A_t^{1/2}}
\leq 4\alpha \sum_{i \in S_{t-1}} \inner{v_i v_i^T}{A_t^{1/2}}
\leq 4\big(n + \alpha \sqrt{n} \cdot \lambda_{\min}(Z_{t-1})\big)
\leq \frac{8n}{\eps},
\]
where the second last inequality uses Lemma~\ref{l:cospectral} and the last inequality is by $\alpha = \sqrt{n}/\eps$ and the assumption that $\lambda_{\min}(Z_{t-1}) \leq 1$.
Since $x \in [0,1]^m$ and $c \geq 0$, it follows that the second last term is 
\[
0 \leq \sum_{i \in S_{t-1} \backslash S'_{t-1}} c_i (1-x_i) \leq c_{\infty} \cdot |S_{t-1} \backslash S'_{t-1}| \leq \frac{8nc_{\infty}}{\eps}.
\]
Similarly, for the last term,
\[
0 \leq 2\alpha \sum_{i \in S'_{t-1}} c_i (1-x_i) \inner{v_i v_i^T}{A_t^{1/2}} 
\leq 2 c_{\infty} \cdot \alpha \sum_{i \in S_{t-1}} \inner{v_i v_i^T}{A_t^{1/2}}
\leq 2c_{\infty}(n+\alpha\sqrt{n}\cdot \lambda_{\min}(Z_{t-1})) 
\leq \frac{4nc_{\infty}}{\eps}.
\]
Plugging back these upper and lower bounds for the last two terms, we obtain
\begin{equation} \label{eq:del-cost-exp}
\frac{1}{k} \Big( c(S_{t-1}) - \sum_{i \in S_{t-1}} c_ix_i - \frac{12n c_{\infty}}{\eps} \Big)
\leq \E[c_{i_t} \mid S_{t-1}] 
\leq \frac{1}{k} \Big( c(S_{t-1}) - \sum_{i \in S_{t-1}} c_ix_i \Big).
\end{equation}
The lemma follows by combining the bounds for the expectations of $c_{i_t}$ and $c_{j_t}$ in (\ref{eq:add-cost-exp}) and (\ref{eq:del-cost-exp}).
\end{proof}

To bound the difference between $c(S_t)$ and $\inner{c}{x}$, 
we consider the following sequences of random variables where
\begin{equation} \label{eq:XY}
Y_t := \frac{c(S_t) - \inner{c}{x}}{c_{\infty}}  {\rm~for~} t \geq 0
\quad {\rm and} \quad
X_t := Y_t - Y_{t-1} = \frac{c_{j_t} - c_{i_t}}{c_{\infty}} {\rm~for~} t \geq 1.
\end{equation}
Note that Lemma~\ref{l:cost-1st-moment} shows that the sequence $\{Y_t\}_t$ has the ``self-adjusting'' property that if $Y_t$ is (more) positive then $E[Y_{t+1}] - Y_t$ is (more) negative and vice versa, so intuitively $Y_t$ cannot be too far away from zero.
The sequence $\{Y_t\}_t$ is not a martingale, and so we cannot apply Freedman's inequality to prove concentration.
Instead, we will use Theorem~\ref{t:adjusting} to prove that the absolute value of $Y_t$ is small with high probability.
To apply Theorem~\ref{t:adjusting},
we need to bound the conditional second moment of $X_t$ and the moment generating function of the initial solution $S_0$.

\begin{lemma} \label{l:cost-2nd-moment}
Suppose $\lambda_{\min}(Z_{t-1}) < 1$.
Then
\[
\E[(c_{j_t}-c_{i_t})^2 \mid S_{t-1}] 
\leq \frac{c_{\infty}}{k} \cdot \Big( \inner{c}{x} + c(S_{t-1}) + \frac{2nc_{\infty}}{\eps} \Big).
\]
\end{lemma}
\begin{proof}
Since $c_i \geq 0$ for all $1 \leq i \leq m$,
\begin{eqnarray*}
\E[ (c_{j_t} - c_{i_t})^2 \mid S_{t-1}]  
& \leq & \max_{i_t, j_t} | c_{j_t} - c_{i_t} | \cdot \E[ |c_{j_t} - c_{i_t}| \mid S_{t-1}] 
\\
& \leq & c_{\infty} \cdot \E[ c_{j_t} + c_{i_t} \mid S_{t-1}]
\\
& \leq & \frac{c_{\infty}}{k} \Big( \inner{c}{x} + c(S_{t-1}) + \frac{2nc_{\infty}}{\eps}\Big),
\end{eqnarray*}
where the last inequality is by \eqref{eq:add-cost-exp} and \eqref{eq:del-cost-exp}.
\end{proof}

We use the fact that the initial solution $S_0$ is generated randomly to bound its moment generating function.

\begin{lemma} \label{l:mgf}
For $a \in [-1,1]$, 
\[
\E\left[ e^{a Y_0} \right] \leq e^{a^2 \inner{c}{x} / c_{\infty}}.
\]
\end{lemma}
\begin{proof}
Let $\chi_i$ be the indicator variable where $\chi_i=1$ if $i \in S_0$ and $\chi_i=0$ otherwise.
Since the algorithm constructs $S_0$ by sampling each vector independently with probability $x_i$, it follows that
\[
\E\left[e^{ac(S_0)/c_{\infty}}\right]
= \E\left[e^{a \sum_{i=1}^m \chi_i c_i /c_{\infty}}\right]
= \prod_{i=1}^m \E\left[e^{a \chi_i c_i /c_{\infty}}\right]
= \prod_{i=1}^m \left( 1-x_i + x_i e^{a c_i / c_{\infty}} \right).
\]
Note that $a c_i / c_{\infty} \leq 1$ as $a \in [-1,1]$ and $c_i/c_{\infty} \leq 1$, and thus $e^{a c_i / c_{\infty}} \leq 1 + a c_i/c_{\infty} + a^2 c_i^2 / c_{\infty}^2$ as $e^p \leq 1 + p + p^2$ for $p \leq 1$. Therefore,
\[
\E\left[e^{a c(S_0)/c_{\infty}} \right] 
\leq \prod_{i=1}^m \left(1 + \frac{a c_i x_i}{c_{\infty}} + \frac{a^2 c_i^2 x_i}{c_{\infty}^2} \right) 
\leq \exp\left( \sum_{i=1}^m \bigg(\frac{a c_i x_i}{c_{\infty}} + \frac{a^2 c_i x_i}{c_{\infty}} \bigg) \right) 
= \exp\left(\frac{(a+a^2)\inner{c}{x}}{c_{\infty}}\right),
\]
where the second inequality uses $1+p \leq e^p$ for $p \in \R$ and $c_i \leq c_{\infty}$ for $1 \leq i \leq m$. 
The claim follows as $Y_0 = (c(S_0) - \inner{c}{x})/c_{\infty}$.
\end{proof}

We are ready to apply Theorem~\ref{t:adjusting} to bound the cost.

\begin{theorem} \label{t:cost}
Suppose the iterative randomized swapping algorithm terminates at the $\tau$-th iteration.
Let $c \in \R^m_+$.
For any $\delta_1 \in [0,1]$, 
\[
\Pr\left[c(S_\tau) \leq (1+\delta_1) \inner{c}{x} + \frac{15nc_{\infty}}{\eps}\right] \geq 1- \exp\left[- \Omega\Big(\frac{\delta_1 n}{\eps}\Big)\right].
\]
Also, for any $\delta_2 \in [0,1]$ and $\delta_3 > 0$, 
\[
\Pr\Big[ c(S_\tau) \geq (1-\delta_2)\inner{c}{x} - \delta_3 n c_{\infty} \Big] \geq 1 - \exp\Big(-\Omega\left( \min\{\delta_2 \delta_3, \eps \delta_3^2\} \cdot n \right)\Big).
\]
\end{theorem}
\begin{proof}
As the algorithm terminates the first time when the minimum eigenvalue of the solution is at least $1-2\eps$, we can assume that $\lambda_{\min}(Z_t) < 1-2\eps < 1$ for $0 \leq t < \tau$.
We will apply Theorem~\ref{t:adjusting} on the sequences $\{X_t\}_t$ and $\{Y_t\}_t$ as defined in \eqref{eq:XY}.
Firstly, note that $|X_t| \leq 1$ by definition for all $t \geq 1$.
Secondly, as $\E[X_t \mid Y_0,...,Y_{t-1}] = \E[(c_{j_t} - c_{i_t})/c_{\infty} \mid S_{t-1}]$ and $Y_{t-1} = (c(S_{t-1}) - \inner{c}{x} )/c_{\infty}$, 
Lemma~\ref{l:cost-1st-moment} implies that
\[
\E\left[X_t \mid Y_0,...,Y_{t-1} \right] 
\leq \frac{1}{kc_{\infty}} \Big(\inner{c}{x} - c(S_{t-1}) + \frac{14nc_{\infty}}{\eps}\Big)  
= -\frac{Y_{t-1}}{k} + \frac{14n}{k\eps},
\]
and 
\[
\E\left[X_t \mid Y_0,...,Y_{t-1} \right] 
\geq \frac{1}{kc_\infty} \Big(\inner{c}{x} - c(S_{t-1})\Big)
= -\frac{Y_{t-1}}{k}.
\]
Thirdly, since $\E[X_t^2 \mid Y_0,...,Y_{t-1}] = \E[(c_{j_t} - c_{i_t})^2 / c_{\infty}^2 \mid S_{t-1}]$, Lemma~\ref{l:cost-2nd-moment} implies that
\[
\E[X_t^2 \mid Y_0, ..., Y_{t-1}] 
\leq \frac{1}{kc_{\infty}} \Big(\inner{c}{x} + c(S_{t-1}) + \frac{2nc_{\infty}}{\eps}\Big) 
= \frac{Y_{t-1}}{k} + \frac{2}{kc_{\infty}} \Big( \inner{c}{x} + \frac{nc_{\infty}}{\eps} \Big).
\]
Finally, Lemma~\ref{l:mgf} states that $\E[e^{aY_0}] \leq \exp(a^2 \inner{c}{x} / c_{\infty})$ for $a \in [-1,1]$.
By setting
\[
\gamma = \frac{1}{k}, \quad 
\beta_u = \frac{14n}{k\eps}, \quad
\beta_l = 0, \quad
\sigma = \frac{2}{kc_{\infty}} \Big( \inner{c}{x} + \frac{nc_{\infty}}{\eps} \Big),
\]
we can check that all the conditions of Theorem~\ref{t:adjusting} are satisfied.
Applying Theorem~\ref{t:adjusting} with $\eta = \delta_1 \inner{c}{x} / c_{\infty} + n/\eps$ for $\delta_1 \in [0,1]$,
\begin{eqnarray*}
\Pr\left[ c(S_t) \geq (1+\delta_1)\inner{c}{x} + \frac{15nc_{\infty}}{\eps} \right]
& = & \Pr\left[ Y_t \geq \frac{\beta_u}{\gamma} + \eta \right]
\\
& \leq & \exp\left[ - \frac{\eta^2}{4(\sigma+\beta_u)/\gamma +2\eta} \right]
\\
& = & \exp\left[ - \frac{c_{\infty} \eta^2}{8\inner{c}{x} + 64nc_{\infty}/\eps + 2\eta c_{\infty} } \right]
\\
& \leq & \exp\left[- \Omega\Big(\frac{\delta_1 n}{\eps}\Big) \right],
\end{eqnarray*}
where the last inequality is because $\eta c_{\infty} = O(\inner{c}{x} + nc_{\infty}/\eps)$ and thus the denominator is $\Theta(\inner{c}{x} + nc_{\infty}/\eps)$,
and the numerator is $\eta^2 c_{\infty} = \eta (\delta_1 \inner{c}{x} + nc_{\infty}/\eps) \geq (n /\eps) \delta_1(\inner{c}{x} + nc_{\infty}/\eps)$. 

Similarly, for the cost lower bound, we apply Theorem~\ref{t:adjusting} 
with $\eta = \delta_2 \inner{c}{x} / c_{\infty} + \delta_3 n$ for $\delta_2 \in [0,1]$ and $\delta_3 > 0$ to obtain
\begin{eqnarray*}
\Pr\left[ c(S_t) \leq (1-\delta_2)\inner{c}{x} - \delta_3 nc_{\infty} \right]
& = & \Pr\left[ Y_t \leq -\frac{\beta_l}{\gamma} - \eta\right]
\\
& \leq & \exp\left[ -\frac{\eta^2}{4\sigma/\gamma + \eta} \right]
\\
& = & \exp\left[ -\frac{\eta^2 c_{\infty}}{8(\inner{c}{x} + nc_{\infty}/\eps)+\eta c_{\infty}} \right]
\\
& \leq & \exp\left[ -\Omega \left( \frac{\delta_3 n(\delta_2 \inner{c}{x} + \delta_3n c_{\infty})}{\inner{c}{x} + nc_{\infty}/\eps + \delta_3 n c_{\infty}} \right) \right]
\\
& \leq & \exp\left[ -\Omega\left( \min\{\delta_2 \delta_3, \eps \delta_3^2\} \cdot n \right) \right],
\end{eqnarray*}
where the second last inequality 
is by similar calculations as in the previous case.
\end{proof}

\subsubsection{Exact One-Sided Spectral Rounding} \label{ss:exact}

Theorem~\ref{t:approx} follows directly from Theorem~\ref{t:eigenvalue} and Theorem~\ref{t:cost}.
This shows that the iterative randomized swapping algorithm will return a solution $S$ with $\sum_{i \in S} v_i v_i^T \succeq (1-2\eps)I_n$ and $c(S_t) \approx \inner{c}{x}$ with high probability for any $c \in \R^m_+$.

To prove Theorem~\ref{t:zero-one} where the goal is to return a solution $S$ with $\sum_{i \in S} v_i v_i^T \succeq I_n$, our idea is to scale up the fractional solution $x$ and then apply Theorem~\ref{t:approx}.
The following is the detailed description of the algorithm.

\begin{framed}
{\bf Exact One-Sided Spectral Rounding}

Input: $v_1, ..., v_m \in \R^n$ and $x \in [0, 1]^m$ with $\sum_{i=1}^m x_i v_i v_i^T = I_n$, and an error parameter $\eps \in (0, \frac14)$.

Output: a subset $S \subseteq [m]$ such that $\sum_{i \in S} v_i v_i^T \succcurlyeq I_n$ and $c(S) \approx \inner{c}{x}$ for any $c \in \R_+^m$ with high probability.
    \begin{enumerate}

    \item Define $y_i := x_i/(1-2\eps)$ and $u_i := \sqrt{1-2\eps} \cdot v_i$ for $i \in [m]$. Note that $\sum_{i=1}^m y_i  u_i u_i^T = I_n$.

    \item Let $S_{\rm big} := \{i \in [m] : y_i > 1 \}$, $S_{\rm small} := \{ i \in [m] : 0 \leq y_i \leq 1 \}$, and $Z_{\rm big} = \sum_{i \in S_{\rm big}} y_i  u_i u_i^T$.  
    
    \item Define $w_i := (I_n - Z_{\rm big})^{-\frac12} u_i$ for each $i \in S_{\rm small}$, so that $\sum_{i \in S_{\rm small}} y_i w_i w_i^T = I_n$\footnotemark.
    
    \item Apply the iterative randomized swapping algorithm with $\{ w_i \mid i \in S_{\rm small} \}$ and $\{y_i \mid i \in S_{\rm small}\}$ as input to obtain a solution set $S'_{\rm small} \subseteq S_{\rm small}$ with $\sum_{i \in S'_{\rm small}} w_i w_i^T\succeq (1-2\eps)I_n$.

    \item Return $S := S_{\rm big} \cup S'_{\rm small}$ as the solution. 
\end{enumerate}
\end{framed}

\footnotetext{If $I_n-B$ is singular, we first project the vectors to the orthogonal complement of the nullspace before applying the transformation.
We can add dummy coordinates to keep the vectors to have the same dimension $n$ for simplicity of the analysis.}

\begin{proofof}{Theorem~\ref{t:zero-one}}
We first analyze the spectral lower bound.
By the definitions of $w_i$ and $u_i$,
\[
\sum_{i \in S'_{\rm small}} w_i w_i^T \succcurlyeq (1-2\eps) I_n
\implies 
\sum_{i \in S'_{\rm small}} u_i u_i^T \succcurlyeq (1-2\eps)(I_n - Z_{\rm big}) 
\implies 
\sum_{i \in S'_{\rm small}} v_i v_i^T \succcurlyeq I_n - Z_{\rm big}.
\]
For the vectors in $S_{\rm big}$, as $x_i \in [0,1]$,
\[
\sum_{i \in S_{\rm big}} v_i v_i^T \succcurlyeq \sum_{i \in S_{\rm big}} x_i  v_i v_i^T = \sum_{i \in S_{\rm big}} y_i  u_i u_i^T = Z_{\rm big}.
\]
Therefore, it follows that
\[
\sum_{i \in S} v_i v_i^T 
= \sum_{i \in S_{\rm big} \cup S'_{\rm small}} v_i v_i^T 
= \sum_{i \in S'_{\rm small}} v_i v_i^T + \sum_{i \in S_{\rm big}} v_i v_i^T 
\succcurlyeq (I_n - Z_{\rm big}) + Z_{\rm big}
= I_n.
\]

Next, we prove that $c(S) \approx \inner{c}{x}$ with high probability for any vector $c \in \R^m_+$. 
Let $\inner{c}{x}_{\rm small} := \sum_{i \in S_{\rm small}} c_i x_i$ 
and $\inner{c}{x}_{\rm big} := \sum_{i \in S_{\rm big}} c_i x_i$. 
For the vectors in $S_{\rm big}$, as $y_i > 1$ for $i \in S_{\rm big}$ and $y_i = x_i / (1-2\eps)$ for all $i \in [m]$, it follows that
\[
\inner{c}{x}_{\rm big} \leq c(S_{\rm big}) \leq \inner{c}{y}_{\rm big} = \frac{\inner{c}{x}_{\rm big}}{1-2\eps}.
\]
For the vectors in $S_{\rm small}$, by Theorem~\ref{t:approx} with $\delta_1=\eps$, the returned set $S'_{\rm small}$ in Step $4$ satisfies the cost upper bound
\[
c(S'_{\rm small}) \leq (1+\eps)\inner{c}{y}_{\rm small} + \frac{15nc_{\infty}}{\eps}
= \frac{(1+\eps)\inner{c}{x}_{\rm small}}{1-2\eps}  + \frac{15nc_{\infty}}{\eps}
\]
with probability at least $1-\exp(-\Omega(n))$, 
which implies that for $\eps \in (0,\frac{1}{4})$,
\[
c(S) = c(S_{\rm big}) + c(S'_{\rm small}) 
\leq \frac{1+\eps}{1-2\eps} \big( \inner{c}{x}_{\rm big} + \inner{c}{x}_{\rm small} \big)  + \frac{15nc_{\infty}}{\eps}
\leq (1+6\eps)\inner{c}{x} + \frac{15nc_{\infty}}{\eps}.
\]
Similarly, by Theorem~\ref{t:approx} with $\delta_2=\eps$ and $\delta_3=\delta$ for some $\delta>0$, the returned set $S'_{\rm small}$ in Step $4$ satisfies the cost lower bound
\[
c(S'_{\rm small}) \geq (1-\eps)\inner{c}{y}_{\rm small} - \delta n c_{\infty}
= \frac{1-\eps}{1-2\eps} \inner{c}{x}_{\rm small} - \delta n c_{\infty}
\geq \inner{c}{x}_{\rm small} - \delta n c_{\infty}
\]
with probability at least $1-\exp(-\Omega( \min\{\eps \delta, \eps \delta^2\} \cdot n))$, which implies that
\[
c(S) = c(S_{\rm big}) + c(S'_{\rm small}) 
\geq \inner{c}{x}_{\rm big} + \inner{c}{x}_{\rm small} - \delta n c_{\infty}
= \inner{c}{x} - \delta n c_{\infty}.
\]
\end{proofof}

\subsubsection{Proof of the Concentration Inequality for Self-Adjusting Process (Theorem~\ref{t:adjusting})} \label{ss:concentration}

The proof is by computing the moment generating function of $Y_t$ and applying Markov's inequality, which is standard in concentration inequalities.
In the following, we write the conditional expectation as $\E_t[ \cdot ] := \E[ \cdot | Y_0, ..., Y_{t-1}]$ for simplicity. 

{\bf Upper Tail:}
We start with the proof for the upper tail. 
For any $a \in [0,1]$, the conditional moment generating function of $X_t$ with any given $Y_0, ..., Y_{t-1}$ is
\begin{eqnarray*}
\E_t\left[ e^{a X_t} \right] 
= \E_t\left[ \sum_{l=0}^{\infty} \frac{a^l X_t^l}{l!} \right]
& \leq & \E_t \left[ 1 + a X_t + X_t^2 \sum_{l=2}^{\infty} \frac{a^l}{l!} \right] 
\\
& = & 1 + a \E_t[X_t] + \E_t[X_t^2] \cdot (e^{a} - 1 - a)
\\
& \leq & 1 + a \E_t[X_t] + a^2 \E_t[X_t^2] 
\\
& \leq & 1 - a \gamma Y_{t-1} + a \beta_u + a^2\gamma Y_{t-1} + a^2 \sigma 
\\
& \leq & \exp\left( a^2 \sigma + a \beta_u - \gamma(1-a) a Y_{t-1}  \right),
\end{eqnarray*}
where the first inequality is by the bounded difference property that $|X_t| \leq 1$ always, 
the second inequality is because $e^p \leq 1+ p + p^2$ for $p \leq 1$, 
the third inequality is by the self-adjusting property and the bounded variance property and $a \geq 0$, and the last inequality uses $1+p \leq e^p$ for $p \in \R$.
Then we can bound the moment generating function of $Y_t$ as
\begin{eqnarray*}
\E_{Y_0,...,Y_{t}}\left[ e^{a Y_t} \right] 
& = & \E_{Y_0,...,Y_{t-1}}\left[ e^{a Y_{t-1}} \cdot \E_t\left[ e^{a X_t} \right] \right] 
\\
& \leq & \E_{Y_0,...,Y_{t-1}} \left[ \exp\big( a^2 \sigma + a \beta_u + (1 - \gamma(1-a)) a Y_{t-1} \big)\right]
\\ 
& \leq & \exp\left( a^2 \sigma + a \beta_u \right) \cdot \E_{Y_0,...,Y_{t-1}}\left[ \exp\left( a \left( 1 - (1-a)\gamma \right) Y_{t-1}\right) \right] 
\\
& = & \exp\left( a^2 \sigma + a \beta_u \right) \cdot \E_{Y_0,...,Y_{t-1}}\left[ \exp\left( f(a) \cdot Y_{t-1}\right) \right],
\end{eqnarray*}
where we define $f(a) := a( 1- (1-a)\gamma)$.
Note that $f(a) \in [0,a]$ by the assumptions $\gamma \in (0,\frac12]$ and $a \in [0,1]$. 
Define the sequence $a_{(0)} = a$ and $a_{(i)} = f(a_{(i-1)})$ for $i \geq 1$. 
Apply the same argument inductively, it follows that
\begin{equation*}
\E_{Y_0,...,Y_{t}} \left[ e^{a Y_t} \right] 
\leq \exp\left[ \sum_{i=0}^{t-1} \Big( a_{(i)}^2 \sigma + a_{(i)} \beta_u \Big) \right] \cdot \E_{Y_0} \left[ e^{a_{(t)} Y_0} \right]   
\leq \exp\left[ \sum_{i=0}^{t-1} \Big( a_{(i)}^2 \sigma + a_{(i)} \beta_u \Big) + \frac{a^2_{(t)} \sigma}{\gamma} \right],
\end{equation*}
where the last inequality follows from the initial concentration property of $Y_0$.
To bound the moment generating function,
we use the following claim whose proof follows from the definition of the sequence $\{a_{(i)}\}_i$.

\begin{claim} \label{cl:seq-u}
The sequence $\{a_{(i)}\}_{i \geq 0}$ is decreasing and dominated by the geometric sequence $\{ar^{i}\}_{i \geq 0}$ with common ratio $r := 1-(1-a)\gamma$.
The sequence $\{a_{(i)}^2\}_i$ is also decreasing and dominated by the geometric sequence $\{a^2 r^{2i}\}_{i \geq 0}$ with common ratio $r^2$.
Furthermore, $r^2 < r < 1$ when $a \in [0,1)$.
\end{claim}

Using Claim~\ref{cl:seq-u}, when $a \in [0,1)$, 
we can upper bound the moment generating function by
\begin{eqnarray*}
\E_{Y_0,...,Y_{t}}\left[ e^{a Y_t} \right] 
& \leq & \exp\left[ \big(a^2 \sigma + a \beta_u \big) \sum_{i=0}^{t-1} r^i + \frac{a^2 \sigma r^t}{\gamma} \right]
\\
& = & \exp\left[ \big(a^2 \sigma + a \beta_u \big) \cdot \frac{1-r^t}{1-r} + \frac{a^2 \sigma r^t}{\gamma} \right]
\\ 
& = & \exp\left[ \frac{a^2 \sigma + a \beta_u}{(1-a) \gamma} \cdot (1-r^t) + \frac{a^2 \sigma r^t}{\gamma}\right] 
\\
& \leq & \exp\left[ \frac{a^2 \sigma + a \beta_u }{(1-a)\gamma} \right],
\end{eqnarray*}
where the last inequality uses $a \in [0,1)$.
By Markov inequality, for any $a \in [0,1)$ and any $\eta > 0$,
\begin{eqnarray*}
\Pr\left[ Y_t \geq \frac{\beta_u}{\gamma} + \eta \right] 
= \Pr\left[ e^{a Y_t} \geq e^{a \left(\frac{\beta_u}{\gamma} + \eta\right)} \right]
& \leq & \E_{Y_0,...,Y_{t}}\left[e^{a Y_t}\right] / e^{a \left(\frac{\beta_u}{\gamma} + \eta\right)}
\\ 
& \leq & \exp\left[ \frac{a^2 \sigma + a \beta_u}{(1-a)\gamma} - a \Big(\frac{\beta_u}{\gamma} + \eta \Big) \right] 
\\
& = & \exp\left[\frac{a^2 (\sigma + \beta_u)}{(1-a)\gamma} - a \eta \right].
\end{eqnarray*}
To prove the best upper bound, we optimize over $a$ and set 
\[
a = 1 - \sqrt{\frac{(\sigma+\beta_u)/\gamma}{(\sigma+\beta_u)/\gamma + \eta}}
= 1 - \sqrt{\frac{\nu}{\nu+\eta}},
\]
where we use $\nu := (\sigma+\beta_u)/\gamma$ as a shorthand.
Notice that $a \in [0,1)$ as $\sigma,\gamma,\eta>0$ and $\beta_u \geq 0$, so the above probability bound applies.
Putting this choice of $a$ back into the exponent on the right hand side,
the exponent is
\begin{eqnarray*}
\frac{a^2 (\sigma + \beta_u)}{(1-a)\gamma} - a \eta 
& = & \frac{(1-\sqrt{\nu/(\nu+\eta)})^2\cdot \nu}{\sqrt{\nu/(\nu+\eta)}} - \left(1-\sqrt{\frac{\nu}{\nu+\eta}}\right) \cdot \eta 
\\
& = & \left(1 + \frac{\nu}{\nu+\eta} - 2 \sqrt{\frac{\nu}{\nu+\eta}}\right)\cdot \sqrt{\nu(\nu+\eta)} - \left(1-\sqrt{\frac{\nu}{\nu+\eta}}\right) \cdot \eta 
\\
& = & \sqrt{\nu(\nu+\eta)} + \nu \sqrt{\frac{\nu}{\nu+\eta}}  -2\nu - \eta + \eta \sqrt{\frac{\nu}{\nu+\eta}} 
\\
& = & -(2\nu+\eta) + 2\sqrt{\nu(\nu+\eta)} 
\\
& = & -(2\nu+\eta) + \sqrt{(2\nu+\eta)^2 - \eta^2} 
\\
& = & -(2\nu+\eta) + (2\nu+\eta) \sqrt{1 - \frac{\eta^2}{(2\nu+\eta)^2}} 
\\
& \leq & -\frac{\eta^2/2}{2\nu+\eta},
\end{eqnarray*}
where we used $\sqrt{1-p} \leq 1 - p/2$ for $p \in [0,1]$ in the last inequality. 
Therefore, we conclude that
\[\Pr(Y_t \geq \eta) 
\leq \exp\left( \frac{a^2 (\sigma + \beta_u)}{(1-a)\gamma} - a \eta \right) 
\leq \exp\left( - \frac{\eta^2 /2}{2(\sigma + \beta_u)/\gamma + \eta} \right),
\]
which completes the proof for the upper tail.

{\bf Lower Tail:}
The proof for the lower tail is quite similar to that for the upper tail. 
The main difference is that we work with the moment generating function $\E[e^{-a Y_t}]$, instead of $\E[e^{a Y_t}]$. 
For any $a \in [0,1]$, the conditional moment generating function of $-X_t$ is
\begin{eqnarray*}
\E_t\left[ e^{-a X_t} \right] 
= \E_t\left[ \sum_{l=0}^{\infty} \frac{(-a)^l X_t^l}{l!} \right] 
& \leq & \E_t \left[ 1 - a X_t + X_t^2 \sum_{l=2}^{\infty} \frac{a^l}{l!} \right] 
\\
& = & 1 - a \E_t[X_t] + \E_t[X_t^2] \cdot (e^{a} - 1 - a) 
\\
& \leq & 1 - a \E_t[X_t] + a^2 \E_t[X_t^2] 
\\
& \leq & 1 + a \gamma Y_{t-1} + a \beta_l + a^2\gamma Y_{t-1} + a^2 \sigma 
\\
& \leq & \exp\left( a^2 \sigma + a \beta_l + \gamma(1+a) a Y_{t-1} \right),
\end{eqnarray*}
where the first inequality is by the bounded difference property $|X_t| \leq 1$ and $a \geq 0$, 
the second inequality is because $e^p \leq 1 + p + p^2$ for $p \leq 1$,
the third inequality is by the self-adjusting property and the bounded variance property and $a \geq 0$,
and the last inequality is by $1+p \leq e^p$ for $p \in \R$.
Then we can bound the moment generating function of $Y_t$ as
\begin{eqnarray*}
\E_{Y_0,...,Y_{t}}\left[ e^{-a Y_t} \right] 
& = & \E_{Y_0,...,Y_{t-1}}\left[ e^{-a Y_{t-1}} \cdot \E_t\left[ e^{-a X_t} \right] \right] 
\\
& \leq & \E_{Y_0,...,Y_{t-1}} \left[ \exp\big( a^2 \sigma + a \beta_l - a(1 - \gamma(1+a)) Y_{t-1} \big)\right]
\\ 
& \leq & \exp\left( a^2 \sigma + a \beta_l \right) \cdot \E_{Y_0,...,Y_{t-1}}\left[ \exp\left( -a \left( 1 - (1+a)\gamma \right) Y_{t-1}\right) \right] 
\\
& = & \exp\left( a^2 \sigma + a \beta_l \right) \cdot \E_{Y_0,...,Y_{t-1}}\left[ \exp\left( -g(a) \cdot Y_{t-1}\right) \right],
\end{eqnarray*}
where we define $g(a) := a( 1- (1+a)\gamma)$.
Note that $g(a) \in [0,a]$ by the assumptions $\gamma \in (0,\frac12]$ and $a \in [0,1]$. 
Define the sequence $a_{(0)} = a$ and $a_{(i)} = g(a_{(i-1)})$ for $i \geq 1$. 
Apply the same argument inductively, it follows that
\begin{equation*}
\E_{Y_0,...,Y_{t}} \left[ e^{-a Y_t} \right] 
\leq \exp\left[ \sum_{i=0}^{t-1} \Big( a_{(i)}^2 \sigma + a_{(i)} \beta_l \Big) \right] \cdot \E_{Y_0} \left[ e^{-a_{(t)} Y_0} \right]   
\leq \exp\left[ \sum_{i=0}^{t-1} \Big( a_{(i)}^2 \sigma + a_{(i)} \beta_l \Big) + \frac{a^2_{(t)} \sigma}{\gamma} \right],
\end{equation*}
where the last inequality follows from the initial concentration property of $Y_0$.
To bound the moment generating function,
we use the following claim whose proof follows from the definition of the sequence $\{a_{(i)}\}_i$.

\begin{claim} \label{cl:seq-l}
The sequence $\{a_{(i)}\}_{i \geq 0}$ is decreasing and dominated by the geometric sequence $\{ar^{i}\}_{i \geq 0}$ with common ratio $r := 1-\gamma$.
The sequence $\{a_{(i)}^2\}_i$ is also decreasing and dominated by the geometric sequence $\{a^2 r^{2i}\}_{i \geq 0}$ with common ratio $r^2$.
Furthermore, $r^2 < r < 1$ when $a \in [0,1]$.
\end{claim}

Using Claim~\ref{cl:seq-l}, when $a \in [0,1]$, 
we can upper bound the moment generating function by
\begin{eqnarray*}
\E_{Y_0,...,Y_{t}}\left[ e^{-a Y_t} \right] 
& \leq & \exp\left[ \big(a^2 \sigma + a \beta_l \big) \sum_{i=0}^{t-1} r^i + \frac{a^2 \sigma r^t}{\gamma} \right]
\\
& = & \exp\left[ \big(a^2 \sigma + a \beta_l \big) \cdot \frac{1-r^t}{1-r} + \frac{a^2 \sigma r^t}{\gamma} \right]
\\ 
& = & \exp\left[ \frac{a^2 \sigma + a \beta_l}{\gamma} \cdot (1-r^t) + \frac{a^2 \sigma r^t}{\gamma}\right] 
\\
& \leq & \exp\left[ \frac{a^2 \sigma + a \beta_l }{\gamma} \right].
\end{eqnarray*}

By Markov inequality, for any $a \in [0,1]$ and any $\eta > 0$,
\begin{eqnarray*}
\Pr\left[ Y_t \leq -\frac{\beta_l}{\gamma} - \eta \right] 
= \Pr\left[ e^{-a Y_t} \geq e^{a \left(\frac{\beta_l}{\gamma} + \eta\right)} \right]
& \leq & \E_{Y_0,...,Y_{t}}\left[e^{-a Y_t}\right] / e^{a \left(\frac{\beta_l}{\gamma} + \eta\right)}
\\ 
& \leq & \exp\left[ \frac{a^2 \sigma + a \beta_l}{\gamma} - a \Big(\frac{\beta_l}{\gamma} + \eta \Big) \right] 
\\
& = & \exp\left[\frac{a^2 \sigma}{\gamma} - a \eta \right].
\end{eqnarray*}
When $\eta \leq 2\sigma/\gamma$, we set $a = (\eta \gamma) / (2\sigma) \in [0,1]$, so the above probability bound applies and gives
\[
\Pr\left[Y_t \leq -\frac{\beta_l}{\gamma} - \tau\right] 
\leq \exp\left[ -\frac{\eta^2 \gamma}{4\sigma} \right] 
\leq \exp\left[ -\frac{\eta^2}{4\sigma/\gamma + \eta} \right].
\]
When $\eta > 2\sigma/\gamma$, we simply set $a = 1$, and the above probability bound gives
\[
\Pr\left[Y_t \leq -\frac{\beta_l}{\gamma} - \eta\right] 
\leq \exp\left[ \frac{\sigma}{\gamma} - \eta \right] 
\leq \exp\left[ -\frac{\tau^2}{4\sigma/\gamma + \eta} \right],
\]
where the last inequality holds by the assumption that $\eta > 2\sigma/\gamma$.
This finishes the proof for the lower tail and thus the proof of Theorem~\ref{t:adjusting}.

\subsection{Two-Sided Spectral Rounding} \label{s:discrepancy}

In this section, we show that the two-sided spectral rounding result in Theorem~\ref{t:4std} can be extended to incorporate one non-negative linear constraint that is given as part of the input.

There is a standard reduction used in~\cite{SHS16} to construct spectral sparsifiers that satisfy additional linear constraints.
Suppose Corollary~\ref{c:4std} were to work for rank two matrices, then we can simply incorporate the linear constraint to the input matrices as $A_i := \begin{pmatrix}v_i v_i^T & 0 \\ 0 & c_i/\inner{c}{x} \end{pmatrix}$ so that $\sum_{i=1}^m x_i A_i = I_{n+1}$, and any $z \in \{0,1\}^m$ so that $\sum_{i=1}^m z_i A_i \approx I_{n+1}$ would have $\inner{c}{z} \approx \inner{c}{x}$.
But the rank one assumption is crucial in the proof of Theorem~\ref{t:4std} and it is an open problem to generalize it to work with higher rank matrices.

Our idea is to use the following signing trick, suggested to us by Akshay Ramachandran, to essentially carry out the same reduction using only rank one matrices.
We state the results in a more general form, where $\sum_{i=1}^m x_i v_i v_i^T$ is not necessarily equal to the identity matrix, so that we can also apply them to additive spectral sparsifiers in Section~\ref{s:additive}.

\begin{lemma} \label{l:Akshay}
Let $v_1, \ldots, v_m \in \R^n$, $x \in [0,1]^m$, and $c \in \R_+^m$.
Suppose $\norm{\sum_{i=1}^m x_i v_i v_i^T}_{\rm op} \leq \lambda$
and $\norm{v_i} \leq l$ for $1 \leq i \leq m$.
Then there exists a signing $s_1, \ldots, s_m \in \{\pm1\}$ such that if we let $u_i := \begin{pmatrix} v_i \\ s_i \sqrt{c_i \lambda/ \inner{c}{x}} \end{pmatrix} \in \R^{n+1}$ then $\norm{\sum_{i=1}^m x_i u_i u_i^T}_{\rm op} \leq \lambda + l \sqrt{\lambda}$.
\end{lemma}
\begin{proof}
By the definition of $u_i$,
\begin{eqnarray*}
\sum_{i=1}^m x_i u_i u_i^T & = &
\begin{pmatrix} \sum_{i=1}^m x_i v_i v_i^T & \sum_{i=1}^m s_i x_i \sqrt{\dfrac{c_i \lambda}{\inner{c}{x}}} v_i
\\
\sum_{i=1}^m s_i x_i \sqrt{\dfrac{c_i \lambda}{\inner{c}{x}}} v_i^T & \sum_{i=1}^m \dfrac{c_i x_i\lambda}{\inner{c}{x}}
\end{pmatrix}
\\
& = &
\begin{pmatrix} \sum_{i=1}^m x_i v_i v_i^T & 0
\\
0 & \lambda
\end{pmatrix}
+
\begin{pmatrix} 0 & \sum_{i=1}^m s_i x_i \sqrt{\dfrac{c_i \lambda}{\inner{c}{x}}} v_i
\\
\sum_{i=1}^m s_i x_i \sqrt{\dfrac{c_i \lambda}{\inner{c}{x}}} v_i^T & 0
\end{pmatrix}
\end{eqnarray*}
The operator norm of the second matrix is bounded by $\norm{\sum_{i=1}^m s_i x_i \sqrt{c_i \lambda / \inner{c}{x}} v_i}$. It follows from triangle inequality that
$\norm{\sum_{i=1}^m x_i u_i u_i^T}_{\rm op} \leq \lambda + \norm{\sum_{i=1}^m s_i x_i \sqrt{c_i \lambda / \inner{c}{x}} v_i}$.
We show that there is a signing $s_1, \ldots, s_m \in \{\pm 1\}$ such that $\norm{\sum_{i=1}^m s_i x_i \sqrt{c_i \lambda / \inner{c}{x}} v_i} \leq l \sqrt{\lambda}$ and this will complete the proof.
Take a uniform random signing and consider
\begin{align*}
\E_{s \in \{\pm 1\}^m} \norm{\sum_{i=1}^m s_i x_i \sqrt{\frac{c_i \lambda}{\inner{c}{x}}} v_i}^2
= &~ \sum_{i=1}^m \E_s \left[s_i^2 x_i^2 \norm{v_i}^2 \frac{\lambda c_i}{\inner{c}{x}}\right] + 
\sum_{i \neq j} \E_s\left[s_i s_j x_i x_j \inner{v_i}{v_j} \frac{\lambda \sqrt{c_i c_j}}{\inner{c}{x}} \right]
\\
= &~ \sum_{i=1}^m x_i^2 \norm{v_i}^2 \frac{\lambda c_i}{\inner{c}{x}}
~ \leq ~ l^2 \sum_{i=1}^m \frac{\lambda c_i x_i}{\inner{c}{x}}
~ = ~ l^2 \lambda,
\end{align*}
where the second line uses that $s_i^2=1$, $\E[s_is_j]=\E[s_i] \cdot \E[s_j]=0$, and $x_i \in [0,1]$, $\norm{v_i} \leq l$ in the inequality.
This implies that there exists such a signing.
\end{proof}

We apply the signing in Lemma~\ref{l:Akshay} to incorporate one non-negative linear constraint into the two-sided spectral rounding result of Kyng, Luh and Song~\cite{KLS19}.

\begin{theorem} \label{t:two-sided-cost}
Let $v_1, \ldots, v_m \in \R^n$, $x \in [0,1]^m$, and $c \in \R_{+}^m$.
Suppose $\norm{\sum_{i=1}^m x_i v_i v_i^T}_{\rm op} \leq \lambda$ and
$\norm{v_i} \leq l$ for $1 \leq i \leq m$.  
Suppose further that $c_{\infty} \leq l^2 \inner{c}{x}/\lambda$ and $l \leq \sqrt{\lambda}$.
Then there exists $z \in \{0,1\}^m$ such that
\[
\norm{\sum_{i=1}^m x_i v_i v_i^T - \sum_{i=1}^m z_i v_i v_i^T}_{\rm op} \leq 8l\sqrt{\lambda}
\quad {\rm and} \quad
| \inner{c}{x} - \inner{c}{z} | \leq \frac{8l}{\sqrt{\lambda}} \inner{c}{x}
\]
\end{theorem}
\begin{proof}
Let $u_i = \begin{pmatrix} v_i \\ s_i \sqrt{c_i \lambda / \inner{c}{x}} \end{pmatrix}$ for $1 \leq i \leq m$, where $s_1, \ldots, s_m$ is the signing given in Lemma~\ref{l:Akshay}.
By the assumption that $c_{\infty} \leq l^2 \inner{c}{x} / \lambda$, it follows that $\norm{u_i}^2 = \norm{v_i}^2 + c_i \lambda / \inner{c}{x} \leq 2l^2$. 
Let $\xi_i$ be a zero-one random variable with probability $x_i$ being one.
Applying Theorem~\ref{t:4std} on $u_1, \ldots, u_m$ and $\xi_1, \ldots, \xi_m$,
there exists $z \in \{0,1\}^m$ such that
\[
\norm{\sum_{i=1}^m x_i u_i u_i^T - \sum_{i=1}^m z_i u_i u_i^T}_{\rm op}
\leq 4\norm{\sum_{i=1}^m {\bf Var}[\xi_i] (u_i u_i^T)^2}_{\rm op}^{1/2}
\leq 4\norm{\sum_{i=1}^m x_i \norm{u_i}^2 u_i u_i^T}_{\rm op}^{1/2}
\leq 4\sqrt{2l^2(\lambda + l \sqrt{\lambda})},
\]
where we use that ${\bf Var}[\xi_i] =x_i(1-x_i) \leq x_i$, $\norm{u_i}^2 \leq 2l^2$ and $\norm{\sum_{i=1}^m x_i u_i u_i^T}_{\rm op} \leq \lambda + l \sqrt{\lambda}$ by Lemma~\ref{l:Akshay}.
By looking at the top left $n \times n$ block, this implies that $\norm{\sum_{i=1}^m x_i v_i v_i^T - \sum_{i=1}^m z_i v_i v_i^T}_{\rm op} \leq 4\sqrt{2l^2(\lambda+l\sqrt{\lambda})} \leq 8 l \sqrt{\lambda}$ where we use the assumption that $l \leq \sqrt{\lambda}$.
By looking at the bottom right entry, we have
\[
\left| \sum_{i=1}^m \frac{x_i c_i \lambda}{\inner{c}{x}} - \sum_{i=1}^m \frac{z_i c_i \lambda}{\inner{c}{x}} \right| \leq 4\sqrt{2l^2(\lambda+l\sqrt{\lambda})} \leq 8 l \sqrt{\lambda}
\quad \implies \quad
\left| \inner{c}{x} - \inner{c}{z} \right| \leq  \frac{8l}{\sqrt{\lambda}} \inner{c}{x}.
\]
\end{proof}

This proves Theorem~\ref{t:two-sided} that incorporates one non-negative linear constraint into Corollary~\ref{c:4std}, by plugging $\lambda=1$ and $l=\eps$ into Theorem~\ref{t:two-sided-cost}.

\subsection{Tight Examples} \label{ss:integrality}

We provide two examples showing the tightness of Theorem~\ref{t:zero-one}.

First, consider the following simple example, which shows the $n c_{\infty}$ additive error term is necessary.

\begin{example} \label{ex:bad}
There are $m=2n$ vectors $v_{11}, v_{12}, ..., v_{n1}, v_{n2} \in \R^n$, a vector $x \in [0,1]^m$, a vector $c \in \R_+^m$, and a parameter $\eps$.
They are defined as follows
\begin{align*}
& x_{i1} = 1, \quad v_{i1} = \sqrt{1-\eps} \cdot e_i, \quad c_{i1} = 0 \qquad \text{and} \\
& x_{i2} = \eps, \quad v_{i2} = e_i, \quad c_{i2} = c_{\infty}, \quad \forall i \in \{1,...,n\}.
\end{align*}
Note that $\inner{c}{x} = \eps n c_{\infty}$ and $\sum_{i=1}^n \sum_{j = 1,2} x_{ij} v_{ij} v_{ij}^T = I_n$. 
\end{example}

\begin{claim}
For any constant $\alpha > 1$, any $z \in \{0,1\}^m$ satisfying the spectral lower bound in Example~\ref{ex:bad} must have $\inner{c}{z} \geq \alpha \inner{c}{x} + \Omega(n c_{\infty})$.
\end{claim}
\begin{proof}
Note that the only vector $z \in \{0,1\}^m$ that satisfies the spectral lower bound exactly is $z = \vec{1}_m$.
This implies that $\inner{c}{z} - \alpha \inner{c}{x} = nc_{\infty} - \alpha \eps n c_{\infty} = (1-\alpha \eps) nc_{\infty}$.
For any $\alpha > 1$, there exists $\eps$ such that $\inner{c}{z} - \alpha \inner{c}{x}$ is at least say $nc_{\infty}/2$.
\end{proof}

Next, we modify an integrality gap example in~\cite{NST19} to show that, even if $c = \vec{1}$ and we are allowing integral-solution instead of zero-one solution, the additive error $O(nc_{\infty}/\eps)$ in Theorem~\ref{t:zero-one} is best possible.

\begin{example} \label{ex:Mohit}
The example contains $m=\binom{n}{2}$ vectors $v_1, ..., v_m \in \R^{n-1}$, a vector $x \in [0,1]^m$ and a vector $c = \vec{1}_m$. 
Let $\Pi \in \R^{(n-1) \times n}$ be the orthogonal projection onto the $(n-1)$-dimensional subspace orthogonal to the all-one vector. 
Given some parameter $k$, we define
\[
v_{ij} = \sqrt{\frac{n-1}{2k}} \cdot \Pi (\chi_i - \chi_j) \qquad \text{and} \qquad x_{ij} = \frac{2k}{n(n-1)}, \qquad \forall 1 \leq i < j \leq n.
\]
Note that $\inner{c}{x} = k$ and $\sum_{i<j} x_{ij} v_{ij} v_{ij}^T = I_{n-1}$ and $x$ has the smallest $\norm{x}_1$ among all vectors satisfying $\sum x_{ij} v_{ij} v_{ij}^T \succeq I_{n-1}$.
\end{example}

We will use the following result from~\cite{NST19}.

\begin{theorem}[Theorem C.2 in \cite{NST19}] \label{t:GeneralizedAB}
Let $G = (V , E)$ be a graph with average degree $d_{\rm avg} = 2m/n$, and let $L_G$ be its unnormalized
Laplacian matrix. Then, as long as $d_{\rm avg}$ is large enough, and $n$ is large enough with respect to $d_{\rm avg}$,
\[
\lambda_2(L_G) \leq d_{\rm avg} - \rho \sqrt{d_{\rm avg}},
\]
where $\lambda_2(L_G)$ is the second smallest eigenvalue of $L_G$, and $\rho > 0$ is an absolute constant. Furthermore, the upper bound for $\lambda_2(L_G)$ still holds for graphs with parallel edges.
\end{theorem}

Using the above theorem, we can prove the following lemma.

\begin{lemma} \label{l:lower}
Let $\{v_{ij}\}, c, x$ be defined as in Example~\ref{ex:Mohit}. For any $z \in \mathbb{Z}_+^m$, if
$\sum_{1\leq i < j \leq n} z_{ij} v_{ij} v_{ij}^T \succcurlyeq I_{n-1}$,
then we have 
\[
\inner{c}{z} \geq k + \Omega(\sqrt{k n} + n). 
\]
\end{lemma}
\begin{proof}
Given any $z \in \mathbb{Z}_+^m$, let $G_z$ be the multi-graph corresponding to $z$ with Laplacian matrix 
\[
L_z = \sum_{1\leq i < j \leq n} z_{ij} (\chi_i-\chi_j)(\chi_i-\chi_j)^T = \frac{2k}{n-1} \Pi \left(\sum_{1 \leq i < j \leq n} z_{ij} v_{ij}^T v_{ij}^T \right) \Pi \succeq \frac{2k}{n-1} \left(I_n - \frac{1}{n} J_n \right),
\]
where the last inequality holds by the assumption on $z$. Therefore, $\lambda_2(L_z) \geq 2k/(n-1)$.

On the other hand, since the average degree of $G_z$ is $d_{\rm avg} = 2\norm{z}_1/n$, we apply Theorem~\ref{t:GeneralizedAB} with properly chosen $n$, for some constant $\rho$ we have
\[
\lambda_2(L_z) \leq d_{\rm avg} - \rho \sqrt{d_{\rm avg}} \qquad \Longrightarrow \qquad \lambda_2(L_z) \leq \frac{2\|z\|_1}{n} - \rho \sqrt{\frac{2\|z\|_1}{n}}.
\]
Combining with $\lambda_2(L_z) \geq 2k/(n-1)$, we have
\[
\frac{2k}{n-1} \leq \frac{2\|z\|_1}{n} - \rho \sqrt{\frac{2\|z\|_1}{n}} \quad \Longrightarrow \quad 2k \leq 2\|z\|_1 - \rho \sqrt{2 n \|z\|_1}.
\]
For the quadratic inequality $2y^2 - \rho \sqrt{2n} y - 2k \geq 0$, we know that the nonnegative solution for $y$ should satisfy
\[
y \geq \frac{\rho \sqrt{2n} + \sqrt{2\rho^2 n +16k}}{4}.
\]
Therefore, letting $y = \sqrt{\norm{z}_1}$, we have
\begin{eqnarray*}
\inner{c}{z} = \|z\|_1 & \geq & (\rho \sqrt{2n} + \sqrt{2\rho^2n + 16k})^2/16
\\
& = & \rho^2 n/4 + k + \rho \sqrt{4\rho^2n + 32kn}/8 
\\
& \geq & k + \rho \sqrt{2kn}/2 + \rho^2n/4
\\
& \geq & k + \Omega(\sqrt{kn} + n).
\end{eqnarray*}
\end{proof}

Suppose we set the parameter $k = qn$ for $q > 16$ in Example~\ref{ex:Mohit}.
If we apply Theorem~\ref{t:zero-one} to the vectors $v_1, ..., v_m \in \R^{n-1}$ and $x \in [0,1]^m$ defined in Example~\ref{ex:Mohit} with $\eps = \sqrt{n/k} < 1/4$, then there exists a $z \in \{0,1\}^m$ such that
\[
\sum_{1\leq i < j \leq n} z_{ij} v_{ij} v_{ij}^T \succcurlyeq I_{n-1} \qquad \text{and} \qquad \inner{c}{z} \leq (1+6\eps) \inner{c}{x} + \frac{15nc_{\infty}}{\eps} = k + O(\sqrt{kn}),
\]
where the last equality uses $\inner{c}{x}=k$.
Note that if the additive error term $O(\eps \inner{c}{x} + nc_{\infty}/\eps)$ has a better dependency on $\eps$, then we can set $\eps$ accordingly such that the cost upper bound will contradict with the lower bound in Lemma~\ref{l:lower}.
For example, if Theorem~\ref{t:zero-one} were improved to $\inner{c}{z} \leq (1+6\eps)\inner{c}{x} + 15nc_{\infty}/\sqrt{\eps}$, then we could set $\eps = (n/k)^{2/3}$ which would imply that $\inner{c}{z} \leq k + O(k^{1/3} n^{2/3})$, contradicting with the lower bound $\inner{c}{z} \geq k + \Omega(\sqrt{kn})$ when $k$ is large enough.
This shows Theorem~\ref{t:zero-one} is tight up to a constant factor in the additive error term $nc_{\infty}/\eps$.

\section{Applications} \label{s:applications}

In this section, we will show that the spectral rounding results in Section~\ref{s:rounding} have many applications including survivable network design (Section~\ref{s:network}), experimental design (Section~\ref{s:experimental}), network design with spectral properties (Section~\ref{s:spectral}) and unweighted spectral sparsification (Section~\ref{s:additive}).

\subsection{General Survivable Network Design} \label{s:network}

We will show that the spectral rounding results provide a new approach to design algorithms for the survivable network design problem.
The main advantage of this approach is that it significantly extends the scope of useful properties that can be incorporated into survivable network design.

The organization of this subsection is as follows.
We begin by writing a large convex program that incorporates many useful constraints into survivable network design in Section~\ref{ss:convex}, and explain how the spectral rounding results can be used to find a solution for this general survivable network design problem in Section~\ref{ss:rounding}.
Then we will see the implications of Theorem~\ref{t:zero-one} to network design in Section~\ref{ss:one-sided} and of Theorem~\ref{t:two-sided} to network design in Section~\ref{ss:two-sided}.
Finally, we discuss how these new results make some progress towards Bansal's question~\cite{Ban19} of designing an approximation algorithm for survivable network design with concentration property in Section~\ref{ss:Bansal}.

\subsubsection{Convex Programming Relaxation} \label{ss:convex}

We can write a convex programming relaxation for the general network design problem incorporating all these constraints as discussed in Section~\ref{s:applications-intro}.
In the following, the input graph is $G=(V,E)$ with $|V|=n$ and $|E|=m$.
The fractional solution is $x \in \R^m$ where the intended solution is to set $x_e=1$ if we choose edge $e$ and $x_e=0$ otherwise.
We first present the convex program and then explain the constraints below.
\begin{alignat}{2} \label{eq:P}
    & \min_x & &~ \inner{c}{x} \nonumber \\
    &  & & \begin{aligned}
          & x(\delta(S)) \geq f(S)  & & \quad \forall S \subseteq V &  & \text{\quad (connectivity constraints)} \\
          & x(\delta(v)) \leq d_v  & & \quad \forall v \in V & & \text{\quad (degree constraints)}\\
          & A x \leq a & & \quad A \in \R_+^{p \times m}, a \in \R_+^p & & \text{\quad (linear packing constraints)}\\
          & B x \geq b & & \quad B \in \R_+^{q \times m}, b \in \R_+^q & & \text{\quad (linear covering constraints)}\\
        & \Reff_{x}(u,v) \leq r_{uv} & & \quad \forall u,v \in V& & \text{\quad (effective resistance constraints)}\\
        & L_{x} \succcurlyeq M  & & \quad M \succcurlyeq 0 & & \text{\quad (spectral constraints)}\\
        & \lambda_2(L_x) \geq \lambda & & & & \text{\quad (algebraic connectivity constraint)} \\
        & 0 \leq x_e \leq 1 & & \quad \forall e \in E & & \text{\quad (capacity constraints)}
    \end{aligned} \tag{CP}
\end{alignat}
Let us explain the constraints one by one.
For the connectivity constraints, we have a connectivity requirement $f_{u,v}$ that there are at least $f_{u,v}$ edge-disjoint paths between every pair $u,v$ of vertices. 
For each subset $S \subseteq V$, we let $f(S) := \max_{u,v: u \in S, v \notin S} f_{u,v}$ and write a constraint that at least $f(S)$ edges in $\delta(S)$ should be chosen, where $x(\delta(S))$ denotes $\sum_{e \in \delta(S)} x_e$.
By Menger's theorem, if an integral solution satisfies all these constraints, then all the connectivity requirements are satisfied.
For the degree constraints, each vertex has a degree upper bound $d_v$ and we write a constraint that at most $d_v$ edges in $\delta(v)$ can be chosen, where $x(\delta(v)) := \sum_{e \in \delta(v)} x_e$.
For the linear packing and covering constraints, all the entries in $A,B,a,b$ are nonnegative, and we assume that $A,B$ have at most a polynomial number of rows in $n,m$.
For effective resistance constraints, we have an upper bound $r_{u,v}$ on the effective resistance between every pair $u,v \in V$.
As in Section~\ref{s:graphs}, we write $\Reff_{x}(u,v) = b_{st}^T L_x^{\dagger} b_{st}$ as the effective resistance between $u$ and $v$ in the fractional solution $x$ where each edge $e$ has conductance $x_e$.
In the spectral and the algebraic connectivity constraints, we write $L_x := \sum_{e \in E} x_e L_e$ as the Laplacian matrix of the fractional solution $x$ where $L_e$ is the Laplacian matrix of an edge as defined in Section~\ref{s:graphs}.
In the spectral constraint, we require that $L_x \succeq M$ for a positive semidefinite matrix $M$.
One could have polynomially many constraints of this form (just as linear packing and covering constraints), but we only write one for simplicity.
In the algebraic connectivity constraint, we require the second smallest eigenvalue of the Laplacian matrix of the solution is at least $\lambda$, which is related to the graph expansion of the fractional solution as described in Section~\ref{s:graphs}.

This convex program can be solved by the ellipsoid method in polynomial time in $n$ and $m$.
There are exponentially many connectivity constraints but we can use a max-flow min-cut algorithm as a polynomial time separation oracle for these constraints~(see e.g.~\cite{Jai01}).
Other linear constraints can easily be checked efficiently, as we assume there are only polynomially many of them.
Next we consider the non-linear constraints.
For the effective resistance constraints, it is known~\cite{GBS08} that $\Reff_x(u,v)$ is a convex function in $x$. 
For the algebraic connectivity constraint, it is known~\cite{GB06} that $\lambda_2$ is a concave function in $x$.
For the spectral constraint, the feasible set is a positive semidefinite cone and is convex in $x$.
So the feasible set for these non-linear constraints form a convex set.
Also, these non-linear constraints can all be checked in polynomial time using standard numerical computations.
Therefore, we can use the ellipsoid algorithm to find an $\eps$-approximate solution to this convex program in polynomial time in $n$ and $m$ with dependency on $\eps$ being $\log(1/\eps)$.

\subsubsection{Spectral Rounding} \label{ss:rounding}

Suppose we are given an optimal solution $x$ to the convex programming relaxation \eqref{eq:P}.
To design approximation algorithms,
the task is to round this fractional solution $x$ into an integral solution $z$ so that $z$ satisfies all the constraints and $\inner{c}{z}$ is close to $\inner{c}{x}$.
There are many different types of constraints and it seems difficult to handle them simultaneously.
In the spectral approach, the main observation is that if we can find an integral solution $z$ such that $\sum_{e \in E} z_e L_e \approx \sum_{e \in E} x_e L_e$ and $\inner{c}{x} \approx \inner{c}{z}$, then all the constraints can be (approximately) satisfied simultaneously.
We state this observation in the following lemma. 

\begin{lemma} \label{l:spectral-constraints}
Let $x \in \R_+^m$ be a feasible solution $x$ to \eqref{eq:P}.
For $\eps \in [0,\frac12]$, any $z \in \Z_+^m$ satsifies
\begin{align*}
\sum_{e \in E} z_e L_e \succeq (1-\eps) \sum_{e \in E} x_e L_e
\quad \implies \quad 
	\left\{ \begin{aligned}
         & z(\delta(S)) \geq (1-\eps)f(S) {\rm~for~all~} S \subseteq V  \\
         & \Reff_z(u,v) \leq (1+2\eps)r_{u,v} {\rm~for~all~} u,v \in V \\
        & L_z \succeq (1-\eps)M, \\
        & \lambda_2(L_z) \geq (1-\eps)\lambda.
    \end{aligned} \right.
\end{align*}

For $\eps \in [0,1]$, any $z \in \Z_+^m$ satifies
\begin{align*}
\sum_{e \in E} z_e L_e \preceq (1+\eps) \sum_{e \in E} x_e L_e
\quad \implies \quad 
	\begin{aligned}
         & z(\delta(v)) \leq (1+\eps)d_v {\rm~for~all~} v \in V.
    \end{aligned}
\end{align*}
\end{lemma}
\begin{proof}
Let $L_x := \sum_{e \in E} x_e L_e$ and $L_z := \sum_{e \in E} z_e L_e$.
We start with the connectivity constraints.
For any $S \subseteq V$, let $\chi_S \in \R^n$ be the characteristic vector of $S$ with $\chi_S(i)=1$ if $i \in S$ and zero otherwise.
It is well-known that 
\[\chi_S^T L_z \chi_S = \chi_S^T \left(\sum_{e \in E} z_e L_e\right) \chi_S = \sum_{e \in E} z_e \chi_S^T L_e \chi_S = \sum_{e \in \delta(S)} z_e = z(\delta(S))\]
and similarly $\chi_S^T L_x \chi_S = x(\delta(S))$.
So, if $L_z \succeq (1-\eps) L_x$, then for all $S \subseteq V$ we have 
\[z(\delta(S)) = \chi_S^T L_z \chi_S \geq (1-\eps)\chi_S^T L_x \chi_S = (1-\eps) x(\delta(S)) \geq (1-\eps) f(S).
\]
For the effective resistance constraints, 
since $L_z \succeq (1-\eps)L_x$, 
it implies that $L_z^\dagger \preceq (1-\eps)^{-1} L_x \preceq (1+2\eps)L_x$ for $\eps \in [0,\frac12]$,
and thus
\[
\Reff_z(u,v) = b_{uv}^T L_z^\dagger b_{uv} \leq (1+2\eps)b_{uv}^T L_x^\dagger b_{uv} = (1+2\eps)\Reff_x(u,v) \leq (1+2\eps)r_{u,v}.
\]
The statements about the spectral lower bound and the algebraic connectivity constraint follows directly from the assumption that $L_z \succeq (1-\eps)L_x$.
Finally, for the degree constraints, suppose we are given $L_z \preceq (1+\eps)L_x$, then it follows that
\[
z(\delta(v)) = \chi_{v}^T L_z \chi_{v} 
\leq (1+\eps)\chi_{v}^T L_x \chi_{v} = (1+\eps) x(\delta(v)) \leq (1+\eps)d_v.
\]
\end{proof}

Lemma~\ref{l:spectral-constraints} says that if $z$ satisfies the spectral lower bound $L_z \succeq L_x$, then the solution $z$ will simultaneously satisfy all connectivity constraints, effective resistance constraints, spectral constraints, and the algebraic connectivity constraint exactly.
Moreover, if $z$ also satisfies the spectral upper bound approximately, then the solution $z$ will approximately satisfy all degree constraints as well.

\subsubsection{Applications of One-Sided Spectral Rounding} \label{ss:one-sided}

We apply Theorem~\ref{t:zero-one} to design approximation algorithms for network design problems that significantly extend the scope of existing techniques.
\begin{alignat}{2} \label{eq:one-sided}
    & {\sf cp} := \min_x & &~ \inner{c}{x} \nonumber \\
    &  & & \begin{aligned}
          & x(\delta(S)) \geq f(S)  & & \quad \forall S \subseteq V &  & \text{\quad (connectivity constraints)} \\
          & A x \leq a & & \quad A \in \R_+^{p \times m}, a \in \R_+^p & & \text{\quad (linear packing constraints)}\\
          & B x \geq b & & \quad B \in \R_+^{q \times m}, b \in \R_+^q & & \text{\quad (linear covering constraints)}\\
        & \Reff_{x}(u,v) \leq r_{uv} & & \quad \forall u,v \in V& & \text{\quad (effective resistance constraints)}\\
        & L_{x} \succcurlyeq M  & & \quad M \succcurlyeq 0 & & \text{\quad (spectral constraint)}\\
        & \lambda_2(L_x) \geq \lambda & & & & \text{\quad (algebraic connectivity constraint)} \\
        & 0 \leq x_e \leq 1 & & \quad \forall e \in E & & \text{\quad (capacity constraints)}
    \end{aligned} \tag{CP1}
\end{alignat}

In network design, a zero-one solution corresponds to a subset of edges where each edge is used at most once (satisfying the capacity constraints).
The following theorem is a consequence of Theorem~\ref{t:zero-one}.

\begin{theorem} \label{t:network-zero-one}
Suppose we are given an optimal solution $x$ to the convex program~\eqref{eq:one-sided}.
For any $\eps \in (0,\frac14)$, there is a polynomial time randomized algorithm to return a zero-one solution $z \in \{0,1\}^m$ to \eqref{eq:one-sided} satisfying all the constraints exactly with probability at least $1-\exp(-\Omega(n))$ except for the linear constraints. The solution $z$ has objective value
\[
\inner{c}{z} \leq (1+6\eps) {\sf cp} +\frac{15 n c_{\infty}}{\eps}
\]
with probability at least $1-\exp(-\Omega(n))$, and satisfies
\[
\inner{A_i}{z} \leq (1+6\eps) a_i + \frac{15 n \norm{A_i}_{\infty}}{\eps}
\]
where $A_i$ is the $i$-th row of $A$, with probability at least $1-\exp(-\Omega(n))$ for each linear packing constraint, and satisfies
\[
\inner{B_j}{z} \geq b_j - \delta n \norm{B_j}_{\infty},
\]
where $B_j$ is the $j$-th row of $B$, with probability at least $1- \exp(-\min\{\eps \delta, \eps \delta^2 \} \cdot \Omega(n))$ for any $\delta > 0$ for each linear covering constraint.
\end{theorem}
\begin{proof}
We apply the following standard transformation to reduce to the one-sided spectral rounding problem.
We assume without loss of generality that the graph $G_x$ formed by the support of the fractional solution $x$ is connected, 
and so $L_{x}$ has rank $n-1$.
Let $\Pi = I_n- \frac{1}{n} J_n$ be the orthogonal projection onto the $n-1$ dimensional subspace orthogonal to the all-one vector, where $J_n$ is the $n \times n$ all-one matrix. 
For each edge $e \in E$, we define a vector $v_e := L_{G_x}^{\dagger/2} \Pi b_e$ which is contained in the $n-1$ dimensional subspace orthogonal to the all-one vector. 
Then
\[
\sum_{e \in E} x_e v_e v_e^T 
= L_{x}^{\dagger/2} \Pi \left( \sum_{e \in E} x_e b_e b_e^T \right) \Pi L_{x}^{\dagger/2} 
= L_{x}^{\dagger/2} \Pi L_x \Pi L_{x}^{\dagger/2}
= I_{n-1}.
\]
For any $\eps \in (0,\frac14)$, we apply Theorem~\ref{t:zero-one} to $x \in [0,1]^m$, $\{v_e\}_{e \in E}$ to find a zero-one solution $z \in \{0,1\}^m$ such that $\sum_{e \in E} z_e v_e v_e^T \succcurlyeq I_{n-1}$ with probability at least $1-\exp(-\Omega(n))$, which implies $\sum_{e \in E} z_e b_e b_e^T \succcurlyeq L_{G_x}$, thus the zero-one solution $z$ satisfies all the constraints in \eqref{eq:one-sided} except for the linear constraints by Lemma~\ref{l:spectral-constraints}.

Theorem~\ref{t:zero-one} also guarantees that with probability at least $1-\exp(-\Omega(n))$ the objective value of $z$ is at most
\[
\inner{c}{z} \leq (1+6\eps) \inner{c}{x} + \frac{15n c_{\infty}}{\eps}.
\]

The guarantees for the linear packing constraints follow the same way as for the objective function, and the guarantees for the linear covering constraints follow from the lower bound part of Theorem~\ref{t:zero-one}.
\end{proof}

We demonstrate the use of Theorem~\ref{t:network-zero-one} in some concrete settings.
The first example shows that Theorem~\ref{t:network-zero-one} provides a spectral alternative to Jain's iterative rounding algorithm to achieve $O(1)$-approximation for a fairly general subclass of the survivable network design problem.

\begin{example}
Theorem~\ref{t:network-zero-one} is a constant factor approximation algorithm as long as $nc_{\infty} = O({\sf cp})$.
Suppose that in our network design problem the average degree is at least $d_{\rm avg}$ and the costs on edges are positive integers with $c_{\infty} = O(d_{\rm avg})$ (e.g. in the minimum $k$-edge-connected subgraph problem every vertex has degree at least $k$ and $1 \leq c_e \leq O(k)$ for $e \in E$, or the solution requires a connected subgraph and $1 \leq c_{e} \leq O(1)$ for $e \in E$, etc).
Then ${\sf cp} \geq \Omega(d_{\rm avg} n) \geq \Omega(c_{\infty} n)$ and Theorem~\ref{t:network-zero-one} provides a constant factor approximation algorithm.
\end{example}

The additive error term $nc_{\infty}$ is the reason that we could not achieve constant factor approximation in general, but this term is unavoidable in the one-sided spectral rounding setting when we need to satisfy the spectral lower bound exactly.
See Section~\ref{ss:integrality} for examples showing the limitations.
Heuristically, we can compute ${\sf cp}$ and if $n c_{\infty} = O({\sf cp})$ then we know Theorem~\ref{t:network-zero-one} will provide good approximate solutions.

The second example shows that Theorem~\ref{t:network-zero-one} returns good approximate solution to survivable network design while incorporating many other constraints simultaneously.

\begin{example}
Suppose the connectivity requirement is to find a $k$-edge-connected subgraph, or more generally $f_{u,v} \geq k$ for all $u,v \in V$.
Assume the cost $c_e$ of each edge $e$ is at least one.
Then ${\sf cp} \geq \Omega(kn)$.

When the cost function satisfies $c_{\infty} = O(k)$, 
then Theorem~\ref{t:network-zero-one} implies that there is a polynomial time randomized algorithm to return a simple $k$-edge-connected subgraph satisfying all the constraints in~\eqref{eq:one-sided} except for the linear constraints (with some non-trivial guarantees), and the cost of the subgraph is at most a constant factor of the optimal value.

When the cost function satisfies $c_{\infty} = O(1)$,
then Theorem~\ref{t:network-zero-one} implies that there is a polynomial time randomized algorithm to return a $k$-edge-connected subgraph satisfying all the constraints in~\eqref{eq:one-sided} except for the linear constraints, and the cost of the subgraph is at most $1+O(1/\sqrt{k})$ factor of the optimal value by setting $\eps = \Theta(1/\sqrt{k})$.
\end{example}

The third example shows when the linear packing and covering constraints can be satisfied up to a multiplicative constant factor.  See also Section~\ref{ss:Bansal} for a related question asked by Bansal~\cite{Ban19}.

\begin{example} \label{ex:covering}
For linear covering constraints, suppose they are of the form $\sum_{e \in F} x_e \geq b_j$ for some subset $F \subseteq E$ where $b_j \geq n$, 
then the returned solution $z$ will almost satisfy this constrint as $\sum_{e \in F} z_e \geq b_j - \delta n \norm{B_j}_{\infty} \geq (1-\delta)b_j$ for some $\delta > 0$.
So, these unweighted covering constraints with large right hand side can be incorporated into survivable network design, even though they can be unstructured.
By a similar argument, any unweighted packing constraints with large right hand side will be only violated by at most a multiplicative constant factor with high probability.
It was not known that Jain's iterative rounding can be adapted to incorporate these linear covering and packing constraints.
\end{example}

We will present more applications of Theorem~\ref{t:network-zero-one} in Section~\ref{s:spectral}, where they can be used to design approximation algorithms for network design problems with spectral requirements.
These problems were studied in the literature before but not much is known about approximation algorithms with performance guarantees.

\subsubsection{Applications of Two-Sided Spectral Rounding} \label{ss:two-sided}

If we can achieve two-sided spectral rounding in network design, 
then we can also approximately satisfy the degree constraints by Lemma~\ref{l:spectral-constraints}.
However, to apply Theorem~\ref{t:two-sided}, we need to satisfy the assumption that the vector lengths are small.
It is known that the vector lengths in the spectral rounding setting corresponds to the effective resistance of the edges in the fractional solution $x$.
In the following, we describe when two-sided spectral rounding can be applied, and discuss what are the implications for network design.
\begin{alignat}{2} \label{eq:two-sided}
    & {\sf cp} := \min_x & &~ \inner{c}{x} \nonumber \\
    &  & & \begin{aligned}
          & x(\delta(S)) \geq f(S)  & & \quad \forall S \subseteq V &  & \text{\quad (connectivity constraints)} \\
          & x(\delta(v)) \leq d_v  & & \quad \forall v \in V & & \text{\quad (degree constraints)}\\
        & \Reff_{x}(u,v) \leq r_{uv} & & \quad \forall u,v \in V& & \text{\quad (effective resistance constraints)}\\
        & L_{x} \succcurlyeq M  & & \quad M \succcurlyeq 0 & & \text{\quad (spectral lower bound)}\\
        & \lambda_2(L_x) \geq \lambda & & & & \text{\quad (algebraic connectivity constraint)} \\
        & 0 \leq x_e \leq 1 & & \quad \forall e \in E & & \text{\quad (capacity constraints)}
    \end{aligned} \tag{CP2}
\end{alignat}

\begin{theorem} \label{t:network-two-sided}
Suppose we are given an optimal solution $x$ to the convex program~\eqref{eq:two-sided}.
For any $\eps \in [0,1]$, if $\Reff_x(u,v) \leq \eps^2$ for every $uv \in E$ and $c_{\infty} \leq \eps^2 \inner{c}{x}$,
then there exists a zero-one solution $z \in \{0,1\}^m$ 
\[(1-O(\eps))L_x \preceq L_z \preceq (1+O(\eps))L_x
\quad {\rm and} \quad 
(1-O(\eps))\inner{c}{x} \leq \inner{c}{z} \leq (1+O(\eps))\inner{c}{x}
\]
This implies that all the constraints of~\eqref{eq:two-sided} will be approximately satisfied by $z$ (e.g. $z(\delta(S)) \geq (1-O(\eps))f(S)$ for all $S \subseteq V$ and $z(\delta(v)) \leq (1+O(\eps))d_v$ for all $v \in V$) and the objective value of $z$ is at most $(1+O(\eps)) {\sf cp}$.
\end{theorem}
\begin{proof}
We apply the same standard transformation as in Theorem~\ref{t:network-zero-one} to reduce to the two-sided spectral rounding problem.
Let $\Pi = I_n - \frac{1}{n}J_n$ as defined in Theorem~\ref{t:network-zero-one}.
For each edge $e$, we define a vector $v_e := L_{x}^{\dagger/2} \Pi b_e$ which is contained in the $n-1$ dimensional subspace orthogonal to the all-one vector. 
Then $\sum_{e \in E} x_e v_e v_e^T = I_{n-1}$ as in Theorem~\ref{t:network-zero-one}. 
Using the assumption that $\Reff_x(i,j) \leq \eps^2$ for every edge $ij \in E$, it follows that 
\[
\|v_{ij}\|^2 = b_{ij}^T L_{x}^{\dagger} b_{ij} = \Reff_x(i,j) \leq \eps^2 {\rm~for~all~} ij \in E,
\]
and thus the assumption in Theorem~\ref{t:two-sided} is satisfied.
We can then apply Theorem~\ref{t:two-sided} on $\{v_e\}_e$ and $c$ to conclude that there exists $z \in \{0,1\}^m$ such that
\[
(1-O(\eps)) I_{n-1} \preccurlyeq \sum_{e \in E} z_e v_e v_e^T \preccurlyeq (1+O(\eps)) I_{n-1} \qquad \text{and} \qquad \langle c, z \rangle \leq (1+O(\eps)) \langle c, x \rangle.
\]
By the definition of $v_e = L_{x}^{\dagger/2} \Pi b_e$, this implies that
\[
(1-O(\eps))L_{x} \preccurlyeq L_{z} = \sum_{e \in E} z_e b_e b_e^T \preccurlyeq (1+O(\eps)) L_{x}. 
\]
By Lemma~\ref{l:spectral-constraints}, the zero-one solution $z$ satisfies all the constraints of \eqref{eq:two-sided} approximately.
\end{proof}

In the following, we compare Theorem~\ref{t:network-two-sided} to Theorem~\ref{t:network-zero-one}.

\begin{enumerate}
\item
{\bf Approximation guarantees:}
When Theorem~\ref{t:network-two-sided} applies,
it can handle degree constraints as well and basically preserves all properties of the fractional solution (e.g. upper bound and lower bound on every cut).
It also gives strong approximation guarantee for the objective value, getting arbitrarily close to the optimal value.
However, the constraints are only approximately satisfied, while in Theorem~\ref{t:network-zero-one} they are exactly satisfied.
Theorem~\ref{t:network-two-sided} can only handle one linear constraint, which is used for the objective function, while Theorem~\ref{t:network-zero-one} can handle many linear constraints simultaneously with an additive error term.

\item
{\bf Assumptions:}
Theorem~\ref{t:network-zero-one} apply without any assumptions, but Theorem~\ref{t:network-two-sided} only applies when $\Reff_x(u,v) \leq \eps^2$ for all $uv \in E$ and $c_{\infty} \leq \eps^2 \inner{c}{x}$.
The assumption about the cost is moderate, as it only requires the maximum cost of an edge is at most $\eps^2$ fraction of the total cost of the solution, which should be satisfied in many applications with small $\eps$.
The main restriction is the first assumption about effective resistances, which may not be satisfied in network design applications,
and we would like to provide some combinatorial characterizations under which the assumption will hold.
Let $\Reff_{\rm diam} := \max_{u,v} \Reff(u,v)$ be the effective resistance diameter of a graph; note that the maximum is taken over all pairs (not just for edges as required in Theorem~\ref{t:network-two-sided}).
For example, it is known that~\cite{CRR+96} a $d$-regular graph with constant expansion has  $\Reff_{\rm diam} \leq O(1/d)$. 
So, if the fractional solution $x$ is close to a $d$-regular expander graph, then Theorem~\ref{t:network-two-sided} can be applied with $\eps \geq 1/\sqrt{d}$.
It is proved in~\cite{AALO18} that a much milder expansion condition guarantees small effective resistance diameter.
For example, in a $d$-regular graph $G$, as long as for some $0<\delta \leq 1/2$,
\[|\delta(S)| \geq \Omega\left((d|S|)^{\frac12+\delta}\right) {\rm~for~all~} S \subseteq V
\quad \implies \quad 
\Reff_{\rm diam} \leq O\left(\frac{1}{d^{2\delta}}\right).
\]
Note that a $d$-regular graph with constant expansion satisfies the much stronger assumption that $|\delta(S)| \geq \Omega(d|S|)$.
Informally, the above result only requires $|\delta(S)|$ to be roughly the square root of $d|S|$ to show that the graph has a small effective resistance diameter (e.g. $3$-dimensional mesh).
So, as long as the fractional solution $x$ is a mild expander as defined in~\cite{AALO18}, the assumption in Theorem~\ref{t:network-two-sided} will be satisfied with small $\eps$.
As another example, if the algebraic connectivity $\lambda_2(L_x)$ of the fractional solution is at least say $1/2\eps^2$, then we have $\Reff_{\rm diam} \leq \eps^2$ so that Theorem~\ref{t:network-two-sided} can be applied.
Heuristically, if one could add the constraints that $\Reff_x(u,v) \leq \eps^2$ for $uv \in E$ so that the convex program~\eqref{eq:two-sided} is still feasible without increasing the objective value too much, then one could then apply Theorem~\ref{t:network-two-sided} to bound the integrality gap of the convex program.
\item
{\bf Algorithms:}
There are polynomial time algorithms to return the solutions guaranteed in Theorem~\ref{t:network-zero-one}, while the proof of Theorem~\ref{t:network-two-sided} is non-constructive.
In network design, Theorem~\ref{t:network-zero-one} give us approximation algorithms, while Theorem~\ref{t:network-two-sided} only gives us integrality gap results for the convex programming relaxation (that there exists a zero-one solution almost satisfying all the constraints with objective value close to the optimal value).
\end{enumerate}

\subsubsection{Concentration Property in Survivable Network Design} \label{ss:Bansal}

Recently, Bansal~\cite{Ban19} designed a rounding technique that achieves the guarantees by iterative rounding and randomized rounding simultaneously.
Suppose there is an iterative rounding algorithm for a problem satisfying some technical assumptions.
Bansal's algorithm will satisfy essentially the same guarantees of the iterative rounding algorithm, and simultaneously the following concentration property with $\beta = O(1)$ with respect to linear constraints as if the algorithm does independent randomized rounding.

\begin{definition}[$\beta$-concentration]
Let $\beta \geq 1$. For a vector valued random variable $X=(X_1, ..., X_m)$, where $X_i$ are possible dependent $0$-$1$ random variables, we say $X$ is $\beta$-concentrated around the mean $x=(x_1, ..., x_m)$ where $x_i = \E[X_i]$, if for every $a \in \R^n$ with $M:= \max_i |a_i|$, $\langle a, X \rangle$ is well-concentrated and satisfies Bernstein's inequality up to a factor of $\beta$ in the exponent, i.e.
\[
\Pr\left( \langle a, X \rangle - \inner{a}{x} \geq t \right) \leq \exp\left( - \frac{t^2/\beta}{2(\sum_{i=1}^m a_i^2 x_i (1-x_i) + Mt/3)} \right).
\]
\end{definition}

Bansal showed various interesting application of his techniques, with $x$ being the fractional solution to the linear programming relaxation and $X$ being the zero-one solution output by the approximation algorithm.
However, he left it as an open question whether there is an $O(1)$-approximation algorithm for survivable network design (the guarantee achieved by Jain's iterative rounding algorithm) with $O(1)$-concentration property.

Our iterative randomized swapping algorithms satisfy similar but weaker concentration properties.
Let $x \in [0,1]^m$ be the fractional solution to the one-sided spectral rounding problem.
The algorithm in Theorem~\ref{t:zero-one} will output a vector-valued random variable $X \in \{0,1\}^m$ such that for any $a \in \R_+^n$ with $M:=\max_i a_i$,
\[
\E[\inner{a}{X}] \leq (1+O(\eps))\inner{a}{x} + O\Big(\frac{n M}{\eps}\Big)
\quad {\rm and} \quad
\Pr( \inner{a}{X} - \E[\inner{a}{X}] \geq \eta ) \leq \exp\left[-\Omega\left(\frac{\eta^2}{\sigma^2 + M \eta} \right) \right],
\]
where $n$ is the dimension of the problem (i.e.~the dimension of the vectors) and $\sigma^2=O(M(\inner{a}{x} + n M/\eps))$ is a term related to the variance of the randomized swapping process. 
In other words, the random variable $\inner{a}{X}$ is concentrated around the expected value $\E[\inner{a}{X}]$, but the expected value $\E[\inner{a}{X}]$ could derivate from $\inner{a}{x}$ by $O(\eps \inner{a}{x} + nM/\eps)$ and the concentration property is weaker than the one required in $\beta$-concentration, as the upper bound of $\sigma^2$ we can obtain is larger than the term $\sum_{i=1}^m a_i^2 x_i(1-x_i)$ in the $\beta$-concentration definition.
We note that both Bansal's proof and our proof use Freedman's concentration inequality or its variant.
Using Theorem~\ref{t:network-zero-one}, we made some progress towards Bansal's question.

\begin{corollary} \label{c:network-concentration}
Let $x \in [0,1]^m$ be an optimal fractional solution to the survivable network design problem (i.e. \eqref{eq:one-sided} with only connectivity and capacity constraints).
Suppose $nc_{\infty} = O(\inner{c}{x})$.
Then there is a randomized polynomial time algorithm to return a solution $z \in \{0,1\}^m$ to the survivable network design problem so that $\inner{c}{z} \leq O(\inner{c}{x})$ with probability at least $1 - \exp(-\Omega(n))$.
Furthermore, for any $a \in \R_+^n$ and $\delta \in (0,1)$ it holds that $\inner{a}{x} - \delta n\norm{a}_\infty \leq \inner{a}{z} \leq O(\inner{a}{x} + n \norm{a}_{\infty})$ with probability at least $1-O(\exp(-\Omega(\delta^2 n)))$.
\end{corollary}

We remark that one can add linear constraints $a$ to the convex program in our framework before we apply the rounding, so that we have some control over $\inner{a}{x}$ of the fractional solution $x$ and hence some control over $\inner{a}{z}$ of the zero-one solution $z$.
But it may not be possible to add linear constraints to the relaxation in Bansal's setting, as adding constraints may make the underlying iterative rounding algorithm stops working (e.g. we do not know of an iterative rounding algorithm for the survivable network design problem with additional linear packing or covering constraints).
See Example~\ref{ex:covering} for a related discussion.
Our results suggest that the spectral approach is perhaps more suitable for achieving concentration property for survivable network design.

\subsection{Experimental Design} \label{s:experimental}

In this subsection, we will apply the one-sided spectral rounding results to design approximation algorithms for weighted experimental design problems, extending the work in~\cite{AZLSW17c,AZLSW20} for (unweighted) experimental design problems.
The presentations will mostly follow those in~\cite{AZLSW17c,AZLSW20}.

\subsubsection{Previous Work}

Experimental design is classical in statistics and has found new applications in machine learning~\cite{Puk06,Ang88,AZLSW20}.
In the general problem, we would like to select $k$ points from a large design pool $\{v_1, \ldots, v_m\} \in \R^n$ to maximize the statistical efficiency regressed on the selected $k$ design points.
This can be formulated as a discrete optimization problem of choosing a subset $S \subseteq [m]$ of at most $k$ vectors, so that its covariance matrix $\Sigma_S := \sum_{i \in S} v_i v_i^T$ has the smallest function value $f(\Sigma_S)$ for some objective function $f$.
Some popular choices of $f$ include
\begin{itemize}
    \item A(verage)-optimality: $f_A(\Sigma) = \tr(\Sigma^{-1})/n$,
    \item D(eterminant)-optimality: $f_D(\Sigma) = (\det \Sigma)^{-1/n}$,
    \item E(igen)-optimality: $f_E(\Sigma) = \lambda_{\max}(\Sigma^{-1})$,
    \item V(ariance)-optimality: $f_V(\Sigma) = \tr(V \Sigma^{-1} V^T)$,
    \item G-optimality: $f_G(\Sigma) = \max \diag(V \Sigma^{-1} V^T)$.
\end{itemize}
Many of these optimization problems are known to be NP-hard~\cite{CM09,NST19},
and we are interested in designing approximation algorithms for these problems.
There are two settings in experimental design.
\begin{enumerate}
\item With Repetition: A vector can be chosen multiple times.
This is equivalent to finding a vector $z \in \Z_+^m$ to minimize $f(\sum_{i=1}^m z_i v_i v_i^T)$ subject to the constraint that $\sum_{i=1}^m z_i \leq k$.
This is common in statistic literature, where multiple measurements with respect to the same design point lead to different values with statistically independent noise.
\item Without Repetition: A vector can be chosen at most once.
This is equivalent to finding a vector $z \in \{0,1\}^m$ to minimize $f(\sum_{i=1}^m z_i v_i v_i^T)$ subject to the constraint that $\sum_{i=1}^m z_i \leq k$.
This is more relevant in machine learning applications, as same data points often give the same result.
\end{enumerate}
To design approximation algorithms for these discrete optimization problems, the approach in~\cite{AZLSW17c,AZLSW20} is to first solve a convex programming relaxation to obtain $x \in [0,1]^m$ that minimizes $f(\sum_{i=1}^m x_i v_i v_i^T)$ subject to the constraint that $\sum_{i=1}^m x_i \leq k$, and then round it to $z \in \{0,1\}^m$ with $\sum_{i=1}^m z_i \leq k$ and $f(\sum_{i=1}^m z_i v_i v_i^T) \leq \gamma f(\sum_{i=1}^m x_i v_i v_i^T)$ for some small constant $\gamma \geq 1$.
Under some mild assumptions on the objective function $f$ (which are satisfied for all the popular choices above), Allen-Zhu, Li, Singh and Wang~\cite{AZLSW17c,AZLSW20} showed that designing a polynomial time $\gamma$-approximation algorithm for the experimental design problem can be reduced to the following one-sided spectral rounding problem.
\begin{problem}
Given $x \in [0,1]^m$ with $\sum_{i=1}^m x_i \leq k$ and $\sum_{i=1}^m x_i v_i v_i^T = I_n$, find $z \in \{0,1\}^m$ with $\sum_{i=1}^m z_i \leq k$ and $\sum_{i=1}^m x_i v_i v_i^T \succeq \frac{1}{\gamma} I_n$ in polynomial time.
\end{problem}
Theorem~\ref{t:swap} proves that this one-sided spectral rounding problem is always solvable with $\gamma=1+\eps$ as long as $k \geq \Omega(n/\eps^2)$, using the regret minimization framework.
This implies a $(1+\eps)$-approximation algorithm for a large class of experimental design problem as long as $k \geq \Omega(n/\eps^2)$, in both the with repetition and without repetition settings.
The assumption that $k \geq \Omega(n/\eps^2)$ is shown to be necessary in achieving a $(1+\eps)$-approximation for E-optimal design~\cite{NST19}.
For some other objective functions, it is possible to relax the assumption $k = \Omega(n/\eps^2)$:
Singh and Xie~\cite{SX18} and Madan, Singh, Tantipongpipat and Xie~\cite{MSTX19} gave $(1+\eps)$-approximation algorithms for D-optimal design when $k = \Omega(n/\eps)$, and Nikolov, Singh and Tantipongpipat~\cite{NST19} gave a $(1+\eps)$-approximation algorithm for A-optimal design when $k = \Omega(n/\eps)$.

\subsubsection{Weighted Experimental Design}

We consider the generalization of the experimental design problem where different design points may have different costs.
In this problem, we are given design points $\{v_1, \ldots, v_m\} \in \R^n$ and a cost vector $c \in \R_+^m$ and a cost budget $C$, the objective is to choose a subset $S \subseteq [m]$ that minimizes $f(\sum_{i \in S} v_i v_i^T)$ subject to the constraint that $\sum_{i \in S} c_i \leq C$.
The problem in the previous subsection is the special case when $c$ is the all-one vector and $C=k$.
We believe that this more general problem will be useful in applications, as it is natural that different experiments have different operation costs.

The approximate spectral rounding Theorem~\ref{t:approx} imply the following one-sided spectral rounding results that satisfy the more general cost constraint (which includes the cardinality constraint as a special case).

\begin{theorem} \label{t:budget}
Let $v_1, \ldots, v_m \in \R^n$ and $x \in [0,1]^m$.
Let $c \in \R_+^m$ and $C = \inner{c}{x}$.
Suppose $\sum_{i=1}^m x_i v_i v_i^T = I_n$ and $C \geq 15nc_{\infty}/\eps^2$.
For any $\eps \in (0,\frac12]$, there is a randomized polynomial time algorithm that returns an integral solution $z \in \{0,1\}^m$ such that $\inner{c}{z} \leq C$ and $\sum_{i=1}^m z_i v_i v_i^T \succeq (1-4\eps)I_n$ with probability at least $1-\exp(-\Omega(n))$.
\end{theorem}
\begin{proof}
The idea is to scale down $x$ then apply Theorem~\ref{t:approx}.
We let $\alpha = 1-2\eps$ and set 
\[y := \alpha x \quad {\rm and} \quad u_i := \frac{v_i}{\sqrt{\alpha}}
\quad \implies \quad
\inner{c}{y} = \alpha \inner{c}{x} = \alpha C 
\quad {\rm and} \quad 
\sum_{i=1}^m y_i u_i u_i^T = \sum_{i=1}^m x_i v_i v_i^T = I_n.
\]
We apply Theorem~\ref{t:approx} on $u_1, \ldots, u_m$ and $y,c$ with $\delta_1 = \eps, q = \sqrt{n}$ to obtain $z \in \{0,1\}^m$ so that
\[
\sum_{i=1}^m z_i u_i u_i^T \succeq (1-2\eps) I_n \quad \implies \quad
\sum_{i=1}^m z_i v_i v_i^T \succeq \alpha(1-2\eps) I_n \succeq (1-4\eps) I_n
\]
and
\[
\inner{c}{z} \leq (1+\eps) \inner{c}{y} + 15nc_{\infty}/\eps \leq 
(1+\eps)(1-2\eps) C + \eps C < C,
\]
where we use the assumptions that $15nc_{\infty}/\eps^2 \leq C$.
The failure probability is at most $\exp(-\Omega(n))$.
\end{proof}

Using the same reduction in~\cite{AZLSW17c,AZLSW20}, Theorem~\ref{t:budget} implies the following approximation algorithms for weighted experimental design, including the weighted version of A/D/E/V/G-design.

\begin{theorem} \label{t:exact-design}
Suppose we are given $m$ design points that are represented by $n$-dimensional vectors $v_1, ..., v_m \in \R^n$, a cost vector $c \in \R^m_+$ and a cost budget $C \in \R_+$.
Assuming that the objective function $f$ satisfies the monotonicity, reciprocal sub-linearity and the polynomial time approximability conditions as described in~\cite{AZLSW17c,AZLSW20} (which hold for A/D/E/V/G-design), we have the following approximation results for weighted experimental design. 

For any fixed $\eps \leq \frac15$, if $C \geq 15 n c_{\infty}/\eps^2$, then there exists a polynomial time randomized algorithm that returns an integral vector $z \in \{0,1\}^m$ such that with probability at least $1-\exp(-\Omega(n))$ it holds that 
    \[
        f\left( \sum_{i=1}^m z_i v_i v_i^T \right) \leq (1+O(\eps)) \cdot \min_{y \in [0,1]^m: \inner{c}{y} \leq C} f\left( \sum_{i=1}^m y_i v_i v_i^T \right) \qquad \text{and} \qquad \inner{c}{z} \leq C.
    \]
\end{theorem}

We note that the algorithms in Theorem~\ref{t:exact-design} can incorporate some additional linear packing and covering constraints, with the same guarantees as in Theorem~\ref{t:network-zero-one}.

Finally, we mention that the two-sided spectral rounding result can also be applied to weighted experimental design.
Assuming all the vectors have length at most $\eps$, it shows that there is a zero-one solution which achieves $(1+O(\eps))$-approximation in weighted experimental design, but it does not provide a polynomial time algorithm to find such a zero-one solution.

\subsection{Spectral Network Design} \label{s:spectral}

There are several previous work on network design problems with spectral requirements.
In this section, we will see that these problems are special cases of the general network design problem and the weighted experimental design problem in Section~\ref{s:network} and Section~\ref{s:experimental}, and our results provide improved approximation algorithms for these problems and also generalize these problems to incorporate many additional constraints.

\subsubsection{Maximizing Algebraic Connectivity}

Ghosh and Boyd~\cite{GB06} study the problem of choosing a subgraph that maximizes the algebraic connectivity (the second smallest eigenvalue of its Laplacian matrix) subject to a cost constraint.
The problem is formulated as follows:
\begin{equation} \label{e:max_lambda}
    \begin{aligned}
            & \lambda_{\rm opt} := \max_{x \in \R^{|E|}} & & \lambda_2\left( \sum_{e \in E} x_e b_e b_e^T \right) \\
            & \text{subject to} & & \sum_{e \in E} c_e x_e \leq C, \\
            & & & x_e \in \{0,1\}, \forall e \in E,
    \end{aligned}
\end{equation}
where $c_e$ is the cost of edge $e$ for $e \in E$ and $C$ is the given cost budget.
As mentioned in~\cite{GB06}, the algebraic connectivity is a good measure on the well-connectedness of a graph, as 
\[\lambda_2(L_G) \leq \min_{S \subseteq V} \frac{n |\delta(S)|}{|S| |\bar{S}|} \leq 2 \min_{0 \leq |S| \leq \frac{n}{2}} \frac{|\delta(S)|}{|S|}
\]
where the first inequality is proved in \cite{FKP03}.
Thus, any graph with large $\lambda_{\opt}$ has no sparse cuts, which also implies that the mixing time of random walks is small.

Ghosh and Boyd show that if the constraint $x_e \in \{0,1\}$ is relaxed to $x_e \in [0,1]$, then the relaxation is convex and can be written as a semidefinite program.
They proposed a greedy heuristic based on the Fiedler vector for the zero-one cost setting (where $c_e \in \{0,1\}$ for all $e$), but they do not provide any approximation guarantee of their heuristic algorithm.

Kolla, Makarychev, Saberi and Teng~\cite{KMS+10} provide the first algorithm with non-trivial approximation guarantee in the zero-one cost setting.
Using subgraph sparsification techniques, they give an algorithm that returns a solution which violates the cost constraint by a factor of at most $8$ and having algebraic connectivity at least $\Omega(\lambda^2_{\opt} / \Delta)$ where $\Delta$ is the maximum degree of the graph.

We observe that if we project the vectors $b_e$ onto the rank $n-1$ subspace orthogonal to the all-one vector, then the objective function of \eqref{e:max_lambda} is simply the reciprocal of the objective function of the E-optimal design problem described in Section~\ref{s:experimental}.
This immediately implies that the result of Allen-Zhu, Li, Singh and Wang~\cite{AZLSW20} can be applied to give a $(1+\eps)$-approximation algorithm for the unweighted problem as long as $C \geq 5n/\eps^2$, although this connection was not made before.
Theorem~\ref{t:exact-design} implies the following approximation result for general non-negative cost function.

\begin{theorem} \label{t:lambda1}
Suppose $C \geq 15nc_{\infty}/\eps^2$ for some $\eps \in (0,\frac12]$.
There is a polynomial time randomized algorithm which returns a zero-one solution $z \in \{0,1\}^m$ for \eqref{e:max_lambda} with with probability at least $1-\exp(-\Omega(n))$ such that 
\[\lambda_2\left( \sum_{e \in E} z_e b_e b_e^T \right) \geq (1-O(\eps)) \lambda_{\opt}
\quad {\rm and} \quad 
\sum_{e \in E} c_e z_e \leq C.
\]
\end{theorem}

As shown in Section~\ref{s:network}, the constraint $\lambda_2(\sum_{e \in E} x_e b_e b_e^T) \geq \lambda_{\opt}$ can be incorporated into network design, and so Theorem~\ref{t:network-zero-one} implies the following result.

\begin{theorem} \label{t:lambda2}
There is a polynomial time randomized algorithm which returns a zero-one solution $z \in \{0,1\}^m$ with probability at least $1-\exp(-\Omega(n))$ such that 
\[
\lambda_2\left( \sum_{e \in E} z_e b_e b_e^T \right) \geq \lambda_{\opt}
\quad {\rm and} \quad 
\sum_{e \in E} c_e z_e \leq (1+O(\eps)) C + O\Big(\frac{nc_{\infty}}{\eps}\Big).
\]
Furthermore, this can be done while incorporating other constraints (e.g. connectivity constraints) as described in Theorem~\ref{t:network-zero-one}.
\end{theorem}

\subsubsection{Minimizing Total Effective Resistance} \label{ssec:totReff}

Ghosh, Boyd and Saberi~\cite{GBS08} study the problem of designing a network that minimizes the total effective resistance.
The problem is formulated as follows.
\begin{equation} \label{e:ERMP}
    \begin{aligned}
            & R_{\opt} := \min_{x \in \R^{|E|}} & & \frac{1}{2} \sum_{u, v \in V} \Reff_x(u,v) \\
            & \text{subject to} & & \sum_{e \in E} x_e \leq k, \\
            & & & x_e \in \{0,1\}, \forall e \in E.
    \end{aligned}
\end{equation}
They showed that if the constraint $x_e \in \{0,1\}$ is relaxed to $x_e \in [0,1]$, then the relaxation is convex and can be written as a semidefinite program.
They did not provide any result for the discrete optimization version in \eqref{e:ERMP}. 

Ghosh, Boyd and Saberi~\cite{GBS08} also show that the total effective resistance is a useful measure in different problems, e.g. average commute time, power dissipation in a resistor network, elmore delay in a RC Circuit, total time constant of an averaging network, and euclidean variance.
Furthermore, they established a connection between \eqref{e:ERMP} and the A-design problem described in Section~\ref{s:experimental}. 
To see this, note that the objective of \eqref{e:ERMP} can be written as
\begin{align*}
    \frac12 \sum_{u,v \in V} \Reff_x(u,v) & = \frac12 \sum_{u \neq v \in V} b_{uv}^T L_{x}^{\dagger} b_{uv} = \left\langle L_{x}^{\dagger},  \frac12 \sum_{u \neq v \in V} b_{uv} b_{uv}^T\right\rangle 
    =  \left\langle L_{x}^{\dagger},  n I_n - 1_n 1_n^T \right\rangle = n \tr\left( L_{x}^{\dagger} \right),
\end{align*}
where the last equality follows as $L_{x}^{\dagger}$ is orthogonal to $1_n$.
Hence, minimizing total effective resistance is equivalent to minimizing $\tr(L_{G_x}^{\dagger}) = \tr( \sum_{e \in E} x_e b_e b_e^T)^{\dagger}$, which is the same as the A-design objective function after we project the vectors onto the subspace orthogonal to the all-one vector.

With this connection, all the recent algorithms for the A-optimal design can be applied to solve \eqref{e:ERMP}. 
For instances, the regret minimization algorithm in \cite{AZLSW20} gives a $(1+\eps)$-approximation algorithm when $k \geq \Omega(n/\eps^2)$, and the proportional volume sampling in~\cite{NST19} achieves $(1+\eps)$-approximation with weaker assumption $k \geq \Omega(n/\eps)$.

Theorem~\ref{t:exact-design} implies the following approximation result for the more general weighted setting, where every edge has a cost $c_e$ and we are given a cost budget $C$ as in \eqref{e:max_lambda}.

\begin{theorem} \label{t:A1}
Suppose $C \geq 15nc_{\infty}/\eps^2$.
There is a polynomial time randomized $(1+O(\eps))$-approximation algorithm for the weighted version of \eqref{e:ERMP}.
\end{theorem}

As shown in Section~\ref{s:network}, the effective resistance constraints can be incorporated into network design, and so Theorem~\ref{t:network-zero-one} implies the following result.

\begin{theorem} \label{t:A2}
There is a polynomial time randomized algorithm which returns a zero-one solution $z \in \{0,1\}^m$ with probability at least $1-\exp(-\Omega(n))$ such that 
\[
\frac{1}{2} \sum_{u,v \in V} \Reff_z(u,v) \leq R_{\opt}
\quad {\rm and} \quad 
\sum_{e \in E} c_e x_e \leq (1+O(\eps)) C + O\Big(\frac{nc_{\infty}}{\eps}\Big).
\]
Furthermore, this can be done while incorporating other constraints (e.g. connectivity constraints) as described in Theorem~\ref{t:network-zero-one}.
\end{theorem}

\subsubsection{Network Design for Effective Resistances}

In~\cite{CLSWZ19}, together with Chan, Schild and Wong, we consider the following new problem about network design for $s$-$t$ effective resistance.
Given a graph $G=(V,E)$ and two vertices $s,t \in V$, find a subgraph $H$ with at most $k$ edges to minimize the effective resistance between $s$ and $t$.
The main result in~\cite{CLSWZ19} is a constant factor approximation algorithm for the problem.
This result motivates the current paper.

Using the results in Section~\ref{s:experimental}, we can generalize the problem by allowing the edges to have costs and considering the sum of effective resistance of multiple pairs.
Using the results in Section~\ref{s:network}, we can add the effective resistance constraints for multiple pairs with the objective of minimizing the cost of the solution subgraph, while incorporating other constraints as described in Theorem~\ref{t:network-zero-one}.

\subsection{Unweighted Spectral Sparsification} \label{s:additive}

We show that the spectral rounding results can also be applied to the study of unweighted spectral sparsification.

\subsubsection{Previous Work}

Batson, Spielman, and Srivastava~\cite{BSS12} proved that any graph has a $(1\pm \eps)$-spectral sparsifier with only $O(n/\eps^2)$ edges, by carefully reweighting the edges of the original graph where different edges may have different weights.
If we require all the edges to have the same weight, then there are simple examples (e.g. barbell graphs) showing that linear-sized spectral sparsification is not always possible.
In a recent paper~\cite{BST19}, Bansal, Svensson and Trevisan ask whether there is a non-trivial notion of unweighted spectral sparsification with which linear-sized spectral sparsification is always possible.
They study a notion suggested by Oveis Gharan.

\begin{definition}[Additive Unweighted Spectral Sparsifier]
Given a graph $G=(V,E)$ with $n$ vertices, $m$ edges and maximum degree $d$, a subgraph $\tilde{G}=(V,F)$ with $\tilde{m}$ edges is an additive spectral sparsifier with error $\eps \in [0,1]$ if
\[
-\eps d I \preceq \frac{m}{\tilde{m}} L_{\tilde{G}} - L_G \preceq \eps d I.
\]
\end{definition}

Bansal, Svensson and Trevisan~\cite{BST19} prove that sparse additive unweighted spectral sparsification is always possible, and they provide both deterministic and randomized algorithms for constructing these sparsifiers.

\begin{theorem}[Randomized Construction~\cite{BST19}] \label{t:BST-randomized}
Given a graph $G=(V,E)$ with $n$ vertices, $m$ edges, maximum degree $d$, and $\eps \in (0,1)$, there is a polynomial time randomized algorithm that finds a subset of edges $F \subseteq E$ with size $\tilde{m} = |F| = O(n \log(1/\eps)^3 /\eps^2)$ such that $\tilde{G} = (V,F)$ satisfies
\[
- \eps d I \preceq \frac{m}{\tilde{m}} L_{\tilde{G}} - L_G \preceq \eps d I.
\]
\end{theorem}

\begin{theorem}[Deterministic Construction \cite{BST19}] \label{t:BST-deterministic}
Given a graph $G=(V,E)$ with $n$ vertices, $m$ edges, maximum degree $d$, and $\eps \in (0,1)$, there is a polynomial time deterministic algorithm that finds a multi-set $F$ of edges with size $\tilde{m} = |F| = O(n/\eps^2)$ such that $\tilde{G} = (V,F)$ satisfies
\[
2\frac{m}{\tilde{m}} D_{\tilde{G}} - 2 D_G -\eps d I \preceq \frac{m}{\tilde{m}} L_{\tilde{G}} - L_G \preceq \eps d I,
\]
where $D_G$ is the diagonal degree matrix of $G$.
\end{theorem}

The proof of Theorem~\ref{t:BST-randomized} is by Lov\'asz local lemma and the converse of expander mixing lemma by Bilu and Linial.
The proof of Theorem~\ref{t:BST-deterministic} is by the regret minimization framework of Allen-Zhu, Liao and Orecchia~\cite{AZLO15}.

Note that Theorem~\ref{t:BST-deterministic} has a slightly weaker spectral lower bound guarantee than Theorem~\ref{t:BST-randomized}.
Also, Theorem~\ref{t:BST-deterministic} can only return a multi-set solution where some edges can be used more than once, and so the sparsifier is integer weighted rather than unweighted where every edge has the same weight.

\subsubsection{Nonconstructive Spectral Rounding and Unweighted Spectral Sparsification}

We show that the existence of a linear-sized additive unweighted spectral sparsifier follows from the two-sided rounding result in Theorem~\ref{t:two-sided-cost}.
The idea is to view the original graph as a fractional solution where every edge $e$ has $x_e = \tilde{m}/m$, and then use Theorem~\ref{t:two-sided-cost} to round this fractional solution to a zero-one solution while preserving the spectral properties of the original graph. 
The additional linear constraint in Theorem~\ref{t:two-sided-cost} allows us to bound the number of edges in the sparsifier.

\begin{theorem} \label{t:additive-cost}
Suppose we are given a graph $G=(V,E)$ with $n$ vertices, $m$ edges, and maximum degree $d$.
Let $\tilde{m} = n/\eps^2$.
For any $\eps \in (0,1]$, there exists a subset of edges $F \subseteq E$ with $|F| \leq 8n/\eps^2 $ such that
\[
-8\sqrt{2} \eps d I_n \preceq L_G - \frac{m}{\tilde{m}} \sum_{e \in F} b_e b_e^T \preceq 8\sqrt{2} \eps d I_n.
\]
\end{theorem}
\begin{proof}
The plan is to apply Theorem~\ref{t:two-sided-cost} with $v_e := b_e$, $x_e := \tilde{m}/m$ and $c := \vec{1}_m$. 
We will first define the parameters $\lambda$ and $l$ and check that the assumptions $l \leq \sqrt{\lambda}$ and $c_{\infty} \leq l^2 \inner{c}{x} / \lambda$ in Theorem~\ref{t:two-sided-cost} are satisfied.
Note that
\[
\norm{\sum_{e \in E} x_e v_e v_e^T}_{\rm op} = \frac{\tilde{m}}{m} \norm{L_G}_{\rm op} \leq \frac{2d \tilde{m}}{m} 
\quad \text{and} \quad 
\norm{v_e} = \sqrt{2} \text{ for all } e \in E.
\]
So we define $\lambda := 2d\tilde{m}/m$ and $l:=\sqrt{2}$.
We check that $\lambda = 2d\tilde{m}/m = 2dn/(\eps^2 m) \geq 2/\eps^2 \geq 2 = l^2$, and $l^2 \inner{c}{x} / \lambda = 2 \tilde{m} / (2d\tilde{m}/m) = m/d \geq 1 = c_{\infty}$.
Therefore, we can apply Theorem~\ref{t:two-sided-cost} to conclude that there exists a subset of edges $F \subseteq E$ (corresponding to the zero-one solution $z$) such that
\[
\norm{\sum_{e \in E} x_e v_e v_e^T - \sum_{e \in F} v_e v_e^T}_{\rm op} \leq 16\sqrt{\frac{d \tilde{m}}{m}} \qquad \text{and} \qquad \left| \sum_{e \in E} x_e c_e - \sum_{e \in F} c_e \right| \leq 8\sqrt{\frac{m}{d \tilde{m}}} \cdot \inner{c}{x}.
\]
Plugging in $x_e = \tilde{m}/m$ and $c = \vec{1}$ and $\tilde{m}=n/\eps^2$, the first statement implies that
\[
\norm{L_G - \frac{m}{\tilde{m}}\sum_{e \in F} v_e v_e^T}_{\rm op} =
\norm{\sum_{e \in E} v_e v_e^T - \frac{m}{\tilde{m}}\sum_{e \in F} v_e v_e^T}_{\rm op} \leq  16\sqrt{\frac{d m}{\tilde{m}}}  
= 16\sqrt{\frac{\eps^2 d m}{n}} \leq 8\sqrt{2} \eps d,
\]
where the last inequality uses $m \leq dn/2$ as the maximum degree is $d$.
Finally, the second statement implies that
\[
\left| \tilde{m} - |F| \right| 
\leq 8 \sqrt{\frac{\eps^2 m}{dn}} \cdot \tilde{m} \leq 4\sqrt{2} \eps \tilde{m}
\quad \implies \quad
|F| \leq (1+4\sqrt{2}\eps) \tilde{m} \leq \frac{8n}{\eps^2}.
\]
\end{proof}

Note that Theorem~\ref{t:additive-cost} improves Theorem~\ref{t:BST-randomized} slightly by removing a factor of $\log^3(1/\eps)$ in the number of edges of the sparsifier.
This confirms the existence of unweighted additive spectral sparsifiers with $O(n/\eps^2)$ edges, which was not known before.
More generally, we can use the same proof with a cost function $c$ with $c_{\infty} \leq \norm{c}_1 / d$ to obtain a sparsifier with $\tilde{m}=n/\eps^2$ and
\[
\norm{L_G - \frac{m}{\tilde{m}}\sum_{e \in F} v_e v_e^T}_{\rm op} \leq 8\sqrt{2}\eps d
\quad {\rm and} \quad
(1-4\sqrt{2} \eps) \sum_{e \in E} c_e \leq \frac{m}{\tilde{m}} \sum_{e \in F} c_e \leq (1+4\sqrt{2}\eps) \sum_{e \in E} c_e.
\]
We remark that the same reduction in~\cite{BST19} can be used to replace $dI$ by $(D_G + d_{\rm avg})I$ where $D_G$ is the diagonal degree matrix of $G$ and $d_{\rm avg}$ is the average degree in $G$.

The main disadvantage of Theorem~\ref{t:additive-cost} is that it does not provide a polynomial time algorithm to find such a sparsifier.
It is a major open problem to make the method of interlacing polynomials used in~\cite{MSS15a,MSS15b,KLS19} constructive.

\subsubsection{Constructive Spectral Rounding and Unweighted Spectral Sparsification}

For the determinstic algorithm, using similar techniques in~\cite{AZLSW17c,AZLSW20} which proves Lemma~\ref{l:cospectral}, we can strengthen Theorem~\ref{t:BST-deterministic} by returning a subgraph with no parallel edges.

\begin{theorem} \label{t:additive-deterministic}
Given a graph $G=(V,E)$ with $n$ vertices, $m$ edges, maximum degree $d$, and $\eps \in (0,1)$, there is a polynomial time deterministic algorithm that finds a {\em subset} $F$ of edges with size $\tilde{m} = |F| = O(n/\eps^2)$ such that $\tilde{G} = (V,F)$ satisfies
\[
2\frac{m}{\tilde{m}} D_{\tilde{G}} - 2 D_G - O(\eps) d I_n \preceq \frac{m}{\tilde{m}} L_{\tilde{G}} - L_G \preceq O(\eps) d I_n.
\]
\end{theorem}

The algorithm is a slight modification of the algorithm in~\cite{BST19},
which is a greedy algorithm based on the regret minimization framework. 
The feedback matrices are of the following form
\[
F_0 = 0 \qquad \text{and} \qquad F_t = \begin{pmatrix}
        L_G - m L_e& \\ & L^+_G - m L^+_e
\end{pmatrix} - 2d I_{2n} \quad \text{for some $e \in E$ and $t \geq 1$},
\]
where $L_G$ is the Laplacian matrix of the original graph, $L^+_G := D_G + A_G$ is the signless-Laplacian of the original graph, and $L_e$ and $L^+_e$ are the Laplacian and signless-Laplacian matrix of a single edge $e$. 
Note that we always have $F_t \preccurlyeq 0$, as $L_G \preccurlyeq 2d I_n$ and $L^+_G \preccurlyeq 2d I_n$ for a graph $G$ of maximum degree $d$. 
\begin{framed}
{\bf Greedy Additive Spectral Sparsification}

Input: An error parameter $\eps \in (0,1)$, and a graph $G = (V,E)$ with $n$ vertices, $m \geq 2n/\eps^2$ edges and maximum degree $d$.

Output: A subgraph $\tilde{G}$ of $G$ with $\tilde{m} = O(n/\eps^2)$ edges satisfying
\[
2\frac{m}{\tilde{m}} D_{\tilde{G}} - 2 D_G -\eps d I_n \preceq \frac{m}{\tilde{m}} L_{\tilde{G}} - L_G \preceq \eps d I_n.
\]
\begin{enumerate}
\item Initialization: Set $S_0 := \emptyset$, $F_0 := 0$, $\tau := n/\eps^2$, and $\alpha = q\eps/\sqrt{dm}$ for some small $q > 0$.

\item For $t = 1$ to $\tau$ do
\begin{enumerate}
	\item Compute the action matrix $A_t = (\alpha\sum_{j=0}^{t-1} F_j + l_t I)^{-2}$, where $l_t \in \R$ is the unique value such that $A_t \succ 0$ and $\tr(A_t)=1$.
	\item Select an edge $e_t \in E \backslash S_{t-1}$ such that
	\[
	    \left\langle A_t, \begin{pmatrix}
	            L_G - m L_{e_t}& \\ & L^+_G - mL^+_{e_t} 
	    \end{pmatrix} \right\rangle 
	    \geq - \frac{2\sqrt{n}}{\alpha m} = -O(\eps) d.
	\]
	
	\item Set \[F_t := \begin{pmatrix} L_G - m L_{e_t} & \\ & L^+_G - m L^+_{e_t} \end{pmatrix} - 2d I_{2n} \qquad \text{and} \qquad S_t := S_{t-1} \cup \{e_t\}.\]
\end{enumerate}
\item Return $\tilde{G} = (V, S_\tau)$ as the solution.
\end{enumerate}
\end{framed}

Note that we can assume $m \geq 2n/\eps^2 = 2\tau$, as otherwise we can simply return $\tilde{G} = G$ as our solution.
The only difference with the algorithm in~\cite{BST19} is in Step 2(b), where we insist on choosing an edge $e_t \in E \setminus S_{t-1}$ to guarantee that the returned solution is a simple subgraph.
If there is no such restriction, then a simple averaging argument in~\cite{BST19} shows that there is an edge $e \in E$ with the inner product in Step 2(b) being non-negative.
With this restriction, we will use the closed-form of the action matrix and Lemma~\ref{l:cospectral} to show that there is still an edge with the inner product in Step 2(b) being not too small.
The following lemma is the new ingredient for the proof of Theorem~\ref{t:additive-deterministic}. 

\begin{lemma} \label{l:select}
For each $1 \leq t \leq \tau$, there always exists an edge $e \in E \backslash S_{t-1}$ such that
	\[
	    \left\langle A_t, \begin{pmatrix}
	            L_G - m L_e & \\ & L^+_G - m L^+_e
	    \end{pmatrix} \right\rangle 
            \geq - \frac{2\sqrt{n}}{\alpha m} \geq -O(\eps)d.
	\]
\end{lemma}
\begin{proof}
The sum of the inner product over all edges in $E \backslash S_{t-1}$ is
\begin{eqnarray*}
& & \sum_{e \in E \backslash S_{t-1}} \left\langle A_t, \begin{pmatrix}
	            L_G - m L_e & \\ & L^+_G - m L^+_e
	    \end{pmatrix} \right\rangle \\
& = & \sum_{e \in E}  \left\langle A_t, \begin{pmatrix}
	            L_G - m L_e & \\ & L^+_G - m L^+_e
	    \end{pmatrix} \right\rangle - \sum_{e \in S_{t-1}}  \left\langle A_t, \begin{pmatrix}
	            L_G - m L_e & \\ & L^+_G - m L^+_e
	    \end{pmatrix} \right\rangle \\
& = & \left\langle A_t, \begin{pmatrix}
	            m L_G - m \sum_{e \in E} L_e & \\ & m L^+_G - m \sum_{e \in E} L^+_e
	    \end{pmatrix} \right\rangle - \left\langle A_t, \sum_{e \in S_{t-1}} \begin{pmatrix}
	            L_G - m L_e & \\ & L^+_G - m L^+_e
	    \end{pmatrix} \right\rangle \\
& = & - \left\langle A_t, \sum_{e \in S_{t-1}} \begin{pmatrix}
	            L_G - m L_e & \\ & L^+_G - m L^+_e
	    \end{pmatrix} \right\rangle,
\end{eqnarray*}
where the last equality follows from $\sum_{e \in E} L_e = L_G$ and $\sum_{e \in E} L^+_e = L^+_G$.
Let 
\[
Z_{t-1} := \sum_{e \in S_{t-1}} \begin{pmatrix}
	            L_G - m L_e & \\ & L^+_G - m L^+_e
	    \end{pmatrix},
\]
and let the eigenvalues of $Z_{t-1}$ be $\lambda_1, ..., \lambda_{2n}$. 
Note that $\lambda_{\min}(Z_{t-1}) \leq 0$ as $\tr(L_G) = \tr(L^+_G) = 2m$ which implies that $\tr(Z_{t-1}) = 0$.

Observe that $A_t = (l_t I_{2n} + \alpha Z_{t-1} - 2 \alpha (t-1) d I_{2n})^{-2}$ and so $A_t$ and $Z_{t-1}$ have the same eigenbasis, and the $i$-th eigenvalue of $A_t$ is $(l_t + \alpha \lambda_i - 2\alpha(t-1)d)^{-2}$.
It follows that
\begin{eqnarray*}
- \langle A_t, Z_{t-1} \rangle 
& = & \sum_{i=1}^{2n} \frac{-\lambda_i}{(l_t + \alpha \lambda_i - 2\alpha (t-1)d)^2}
\\
& = & \sum_{i=1}^{2n} \frac{l_t/\alpha - 2 (t-1)d}{(l_t + \alpha \lambda_i - 2\alpha (t-1)d)^2} - \frac{1}{\alpha} \sum_{i=1}^{2n} \frac{l_t + \alpha \lambda_i - 2\alpha (t-1)d }{(l_t + \alpha \lambda_i - 2\alpha (t-1)d)^2} 
\\
& = & \frac{l_t}{\alpha} - 2(t-1)d - \frac{\tr(A_t^{1/2})}{\alpha}
\\
& \geq & -\lambda_{\min}(Z_{t-1}) - \frac{\tr(A_t^{1/2})}{\alpha}
\\
& \geq & -\frac{\sqrt{n}}{\alpha},
\end{eqnarray*}
where the third equality is because $\tr(A_t)=1$ and $(l_t+\alpha \lambda_i - 2\alpha(t-1)d)^{-1}$ is the $i$-th eigenvalue of $A_t^{1/2}$,
the first inequality is by $A_t \succ 0$ which implies that $l_t > 2\alpha(t-1)d - \lambda_{\min}(Z_{t-1})$, and the last inequality is by $\lambda_{\min}(Z_{t-1}) \leq 0$ and $\tr(A_t^{1/2}) \leq \sqrt{n}$ from \eqref{e:trace}.

Since $|E \setminus S_{t-1}| = m-t+1$, an averaging argument shows that there exists an edge $e \in E \backslash S_{t-1}$ such that 
\[
\left\langle A_t, \begin{pmatrix}
	    L_G - m L_e & \\ & L^+_G - m L^+_e
	    \end{pmatrix} \right\rangle 
            \geq - \frac{\sqrt{n}}{\alpha(m-t+1)} 
            \geq - \frac{2\sqrt{n}}{\alpha m},
\]
where the last inequality is because $m-t+1 \geq m- \tau + 1 \geq m/2$ by our assumption $\tau = n/\eps^2 \leq m/2$.
Finally, when $\alpha = q\eps/\sqrt{dm}$ for some constant $q > 0$, 
\[
\frac{2\sqrt{n}}{\alpha m} 
= O\left( \frac{1}{\eps} \sqrt{\frac{dn}{m}}\right) 
\leq O\left(\sqrt{d}\right) 
\leq O(\eps d),
\]
where the first inequality is by our assumption $\tau = n/\eps^2 \leq m/2$,
and the second inequality follows from $dn \geq m \geq 2n/\eps^2$ which implies $\eps \geq \sqrt{2/d}$.
\end{proof}

Given Lemma~\ref{l:select}, the rest of the proof is almost the same as that in~\cite{BST19}, but we include the proofs for completeness.
The following lemma bounds the width term, which is essentially the same as Claim~$1$ in~\cite{BST19} with minor modification.

\begin{lemma} \label{l:width}
If $\alpha = q\eps/ \sqrt{dm}$ for a sufficiently small constant $q>0$, then
\[
\alpha \norm{A_t^{\frac{1}{4}} F_t A_t^{\frac{1}{4}}}_{\rm op} \leq \min \left\{ \frac{1}{4}, \eps \right\}.
\]
\end{lemma}
\begin{proof}
Since the feedback matrices $F_t$ have a block diagonal structure, by the closed-form solution of the action matrix in \eqref{e:closed-form}, 
$A_t$ also has the same block diagonal structure
\[
A_t = \begin{pmatrix}
        B_t & \\ & C_t
\end{pmatrix}, \quad \text{where} \quad 0 \preccurlyeq B_t, C_t \preccurlyeq I_n.
\]
Therefore, 
\[
\norm{A_t^{\frac{1}{4}} F_t A_t^{\frac{1}{4}}}_{\rm op}  
= \max\left\{ \norm{B^{\frac{1}{4}}_t \big(L_G - m L_{e_t} - 2dI_n\big) B^{\frac{1}{4}}_t}_{\rm op}, 
\norm{C^{\frac{1}{4}}_t \big(L^+_G - m L^+_{e_t} -  2dI_n\big) C^{\frac{1}{4}}_t}_{\rm op} \right\}.
\]
We will just bound the first term, as the second term can be bounded the same way.
By triangle inequality and the facts that $0 \preccurlyeq B_t \preceq I_n$ and $0 \preccurlyeq L_G \preccurlyeq 2d I$, it follows that
\[
\norm{B^{\frac{1}{4}}_t \big(L_G - m L_{e_t} -  2dI_n\big) B^{\frac{1}{4}}_t}_{\rm op} 
\leq m \norm{B^{\frac{1}{4}}_t L_{e_t} B^{\frac{1}{4}}_t}_{\rm op} + 2d,
\]
By the choice of edge $e_t$ as guaranteed by Lemma~\ref{l:select}, 
\[
- \frac{2\sqrt{n}}{\alpha m}
\leq  \left\langle A_t, \begin{pmatrix}
        L_G - m L_{e_t} & \\ & L^+_G - m L^+_{e_t}
\end{pmatrix} \right\rangle
=      
\left\langle A_t, \begin{pmatrix}
  L_G & \\ & L^+_G
\end{pmatrix} \right\rangle 
- m \langle L_{e_t}, B_t \rangle - m \langle L^+_{e_t}, C_t \rangle.
\]
Since $\tr(A_t)=1$, $L_G, L^+_G \preccurlyeq 2d I_n$ and $\inner{L_{e_t}^+}{C_t} \geq 0$, the above inequality implies that
\[
m \langle L_{e_t}, B_t \rangle \leq 2d + \frac{2\sqrt{n}}{\alpha m}.
\]
Let $B_t = \sum_{i=1}^n \lambda_i y_i y_i^T$ be the eigendecomposition of $B_t$, and let $w = \sqrt{m} \cdot b_{e_t}$ so that $w w^T = m L_{e_t}$ and $\norm{w}_2 = \sqrt{2m}$. 
Then
\begin{eqnarray*}
m\norm{B_t^{\frac{1}{4}} L_{e_t} B^{\frac{1}{4}}_t}_{\rm op} 
= w^T B^{\frac{1}{2}}_t w
& = & \sum_{i=1}^n \sqrt{\lambda_i} \cdot \inner{w}{y_i}^2 
\\
& \leq & \sqrt{\sum_{i=1}^n \langle w, y_i \rangle^2} \cdot \sqrt{\sum_{i=1}^n \lambda_i \cdot \langle w, y_i \rangle^2} 
\\
& = & \norm{w}_2 \cdot \sqrt{w^T B_t w} 
\\
& \leq & \sqrt{2m} \cdot \sqrt{2d + \frac{2\sqrt{n}}{\alpha m}} 
\\
& = & 2\sqrt{dm + \frac{\sqrt{n}}{\alpha}},
\end{eqnarray*}
where the first inequality is by Cauchy-Schwartz, 
and the second inequality follows from $w^T B_t w = m \inner{B_t}{L_{e_t}}$ and the above upper bound on $m\inner{B_t}{L_{e_t}}$.
The same arguments gives the same upper bound on $m\|C_t^{1/4} L^+_{e_t} C^{1/4}_t\|$.
Therefore, for $\alpha = q\eps/\sqrt{dm}$, 
\begin{eqnarray*}
\alpha \norm{A^{\frac{1}{4}}_t F_t A^{\frac{1}{4}}_t}_{\rm op} 
& \leq & 2\alpha d + 2\sqrt{\alpha^2 dm + \alpha \sqrt{n}} 
\\
& = & 2q\eps \sqrt{\frac{d}{m}} + 2 \sqrt{q^2 \eps^2 + q \eps \sqrt{ \frac{n}{dm} }}
\\
& \leq & O(\sqrt{q}\eps), 
\end{eqnarray*}
where the last inequality follows from $d/m \leq n/m \leq \eps^2/2$ by the assumption that $m \geq 2n/\eps^2$.
Then the lemma follows when $q$ is sufficiently small.
\end{proof}

We are ready to prove Theorem~\ref{t:additive-deterministic},
whose proof is essentially the same as in~\cite{BST19}.

\begin{proofof}{Theorem~\ref{t:additive-deterministic}}
By Lemma~\ref{l:width}, when $\alpha = q\eps/\sqrt{dm}$ for a small enough constant $q>0$, then $\alpha\|A^{1/4}_t F_t A^{1/4}_t\| \leq \frac{1}{4}$ for any $t$ which also implies that $\alpha A_t^{1/4} F_t A_t^{1/4} \succcurlyeq - \frac14 I$ for any $t$. 
Therefore, we can apply Theorem~\ref{t:regret-psd-nsd} and get
\begin{align} \label{e:additive-regret}
R_\tau = \sum_{t=1}^\tau \inner{A_t}{F_t} - \lambda_{\min}\left( \sum_{t=1}^\tau F_t \right) & \leq O(\alpha) \sum_{t=1}^\tau \inner{A_t}{|F_t|} \cdot \norm{A^{\frac{1}{4}}_t F_t A^{\frac{1}{4}}_t}_{\rm op} + \frac{2\sqrt{n}}{\alpha}.
\end{align}
By Lemma~\ref{l:select} and Lemma~\ref{l:width} and the fact that $F_t \preceq 0$, it holds that
\[
\inner{A_t}{F_t} \geq -(2+O(\eps))d, \qquad  \inner{A_t}{|F_t|} \leq (2+O(\eps))d \qquad \text{and} \qquad \alpha \|A_t^{\frac{1}{4}} F_t A_t^{\frac{1}{4}}\| \leq \min\left\{ \frac{1}{4}, \eps \right\}.
\]

Together with $\alpha = \Theta(\eps/\sqrt{dm}) \geq \Omega(\eps/d \sqrt{n}) = \Omega(\sqrt{n}/(\eps d \tau))$, the regret minimization bound \eqref{e:additive-regret} implies that
\[
-(2+O(\eps))\tau d - \lambda_{\min} \left( \begin{pmatrix}
\tau L_G - m \sum_{t=1}^\tau L_{e_t} &  \\ & \tau L^+_G - m \sum_{t=1}^\tau L^+_{e_t} 
\end{pmatrix} - 2\tau d I \right) \leq O(\eps) \tau d.
\]
Let $\tilde{m} := \tau$ and $L_{\tilde{G}} = \sum_{t=1}^\tau L_{e_t}$.
From the first block, we have
\[
\tau L_G - \sum_{t=1}^\tau m L_{e_t} \succcurlyeq - O(\eps) \tau dI_n 
\quad \implies \quad
\frac{m}{\tilde{m}} L_{\tilde{G}} - L_G \preceq O(\eps) dI_n.
\]
From the second block, we have
\[
\tau L^+_G - \sum_{t=1}^\tau m L^+_{e_t} \succcurlyeq - O(\eps) \tau dI_n
\quad \implies \quad
\frac{m}{\tilde{m}} L_{\tilde G} - L_G \succeq 2\frac{m}{\tilde{m}} D_{\tilde{G}} - 2D_G - O(\eps)dI_n,
\]
where we used that $L^+_G = 2D_G - L_G$.
\end{proofof}

\section*{Concluding Remarks}

We propose a spectral approach to design approximation algorithms for network design problems.
We show that the techniques developed in spectral graph theory and discrepancy theory can be used to significantly extend the scope of network design problems that can be solved.
We believe that this connection will bring new techniques and stronger results for network design, and will also introduce new formulations and interesting questions to spectral graph theory and discrepancy theory.
It also gives extra motivation to design a constructive algorithm for the method of interlacing polynomials, as this will lead to very strong approximation algorithms for network design.
We leave it as an open question to improve the spectral approach to fully recover Jain's result.  

\subsection*{Acknowledgement} 

We thank Akshay Ramachandran for many useful discussions, for bringing~\cite{KLS19} to our attention, and for his suggestion of the signing trick in Lemma~\ref{l:Akshay} to proving Theorem~\ref{t:two-sided}.
We also thank Shayan Oveis Gharan for comments that improve the presentation of the paper.


\bibliographystyle{plain}

\end{document}